\documentclass[final,3p,authoryear]{elsarticle} 

\usepackage{amsfonts}
\usepackage{amsmath}
\usepackage{amssymb}
\usepackage{amsthm}
\usepackage{bm}

\usepackage{graphicx,psfrag,epsf}
\usepackage{enumerate}
\usepackage{enumitem}
\usepackage{natbib}
\usepackage{url} 

\usepackage{subfigure}
\usepackage{mathtools}
\usepackage{slashbox}
\usepackage{textcomp} 

\usepackage{tikz}
\usepackage{color}
\usepackage{colortbl}
\usepackage{ctable}

\usepackage{lineno}

\newcommand{\var}{\text{Var}}
\newcommand{\Skew}{\text{Skew}}
\newcommand{\Kurt}{\text{Kurt}}

\def\btheta{\boldsymbol{\theta}}

\def\bzero{\boldsymbol{0}}

\theoremstyle{definition}
\newtheorem{defi}{Definition}

\newtheorem{theo}{Theorem} 
\newtheorem{pro}{Proposition} 
\renewenvironment{proof}{\textit{Proof.}}{\qed}

\journal{arXiv}

\biboptions{authoryear}

\begin{document}

\begin{frontmatter}

\title{Asymmetric Laplace scale mixtures for the distribution of cryptocurrency returns}
\author[Antonio]{Antonio Punzo} 
\ead{antonio.punzo@unict.it}
\author[Luca]{Luca Bagnato\corref{cor}} 
\ead{luca.bagnato@unicatt.it}
\cortext[cor]{Corresponding author: Luca Bagnato -- Dipartimento di Scienze Economiche e Sociali, Universit\`{a} Cattolica del Sacro Cuore, Via Emilia Parmense 84, 29122, Piacenza, Italy.
Email: \texttt{luca.bagnato@unicatt.it}, 
Phone: +39-0523-599326, 
Fax: +39-0523-599437}
\address[Antonio]{Dipartimento di Economia e Impresa, Universit\`{a} di Catania, Italy}
\address[Luca]{Dipartimento di Scienze Economiche e Sociali, Universit\`{a} Cattolica del Sacro Cuore, Piacenza, Italy}

%
%
%

\begin{abstract}
Recent studies about cryptocurrency returns show that its distribution can be highly-peaked, skewed, and heavy-tailed, with a large excess kurtosis. 
To accommodate all these peculiarities, we propose the asymmetric Laplace scale mixture (ALSM) family of distributions.
Each member of the family is obtained by dividing the scale parameter of the conditional asymmetric Laplace (AL) distribution by a convenient mixing random variable taking values on all or part of the positive real line and whose distribution depends on a parameter vector $\btheta$ providing greater flexibility to the resulting ALSM.
Advantageously with respect to the AL distribution, the members of our family allow for a wider range of values for skewness and kurtosis.  
For illustrative purposes, we consider different mixing distributions; they give rise to ALSMs having a closed-form probability density function where the AL distribution is obtained as a special case under a convenient choice of $\btheta$.
We examine some properties of our ALSMs such as hierarchical and stochastic representations and moments of practical interest. 
We describe an EM algorithm to obtain maximum likelihood estimates of the parameters for all the considered ALSMs.
We fit these models to the returns of two cryptocurrencies, considering several classical distributions for comparison.
The analysis shows how our models represent a valid alternative to the considered competitors in terms of AIC, BIC, and likelihood-ratio tests.

\end{abstract}


\end{frontmatter}


\section{Introduction}
\label{sec:Introduction}

In the recent years, cryptocurrencies have grown to be alternative investments for investors looking to hedge against stock market crashes and optimize their earnings.
As a matter of fact, many established cryptocurrencies are currently in the market, and over the years many of these assets have appreciated radically.
This justifies the growing number of studies about the distributional peculiarities of cryptocurrency returns.
However, from this literature, it is not yet possible to infer formal and rigorous ``stylized facts'' for the distribution of these returns.
Nevertheless, some common findings are already known.
Without going into details about the cryptocurrencies considered in each work we will cite, these findings show that cryptocurrency returns -- as usual for financial returns in general -- are clearly non-normally distributed (\citealp{Bariviera:2017}, \citealp{Zhang:2018}, and \citealp{Taka:PhyA:2018}); instead, their distribution can be highly peaked, skewed, and leptokurtic (\citealp{chu2015statistical}, \citealp{Oste:2017}, \citealp{Zhang:2018}, and \citealp{Szcz:Kara:Zare:Appl:2019}),
with high levels of excess kurtosis (always greater than 3 in the analyses by \citealp{Chan:2017}, \citealp{Zhang:2018}, \citealp{Bariviera:2017}, \citealp{Phillip:2018}, \citealp{Szcz:Kara:Zare:Appl:2019}, and \citealp{Punz:Bagn:PhyA:2021}).

The asymmetric Laplace (AL) distribution could be a good candidate to accommodate all these findings; unfortunately, its skewness and kurtosis can only assume values in the intervals $\left(-2,2\right)$ and $\left[3,6\right)$, respectively (refer to Section~\ref{subsec:Preliminaries: the asymmetric Laplace distribution}), and this limits its performance when fitted to the empirical distribution of cryptocurrency returns.
To overcome these issues, we introduce the family of asymmetric Laplace scale mixtures (refer to Section~\ref{sec:AL scale mixture}).
The underlying idea consists in dividing the scale parameter, of the conditional AL distribution, by a convenient mixing random variable, whose distribution takes values on all or part of the positive real line and depends on a parameter vector $\btheta$ further governing the shape of the unconditional mixture.
This confers more flexibility to the conditional AL distribution, allowing for a larger range of values for skewness excess kurtosis, without modifying its peculiar peaked shape.  
As an alternative, our proposal can be seen as a generalization of the (symmetric) Laplace scale mixture proposed in \citet{Punz:Bagn:PhyA:2021} to allow for skewness.   
 
For illustrative purposes, in Section~\ref{sec:Examples of AL scale mixtures} we consider seven members of our family obtained by choosing convenient mixing distributions.
The mixing distributions we choose have the advantage to produce a compound model having a closed-form probability density function (pdf); this advantage should not to be undervalued since, as emphasized by \citet{Shev:Calc:2010}, simple closed form expressions for the pdf of compound models are often not available.

As usual in the literature, we consider the maximum likelihood (ML) approach to estimate the parameters of our models.
To obtain these estimates, in Section~\ref{sec:Maximum likelihood estimation} we describe an expectation-maximization (EM) algorithm which is general enough to be easily extended to mixing distributions beyond those considered in this paper.
From the updates of the parameters of the nested AL distribution, obtained in the M-step of the algorithm, we also show how the influence of observations associated to large absolute distances is automatically reduced (downweighted) in the estimation phase.

In Section~\ref{sec:Real data analysis}, we summarize the results of an analysis where our models are fitted to the returns of two cryptocurrencies (Bitcoin EUR and TRON EUR).
We consider several classical symmetric distributions as competitors.
The analysis shows how the proposed models 
provide a significant improvement with respect to the considered competitors in terms of AIC, BIC and likelihood-ratio tests.

We conclude the paper, in Section~\ref{sec:Conclusions and future works}, with further insights and possible future works related to the proposed models.

\section{Asymmetric Laplace scale mixture}
\label{sec:AL scale mixture}

\subsection{Preliminaries: the asymmetric Laplace distribution}
\label{subsec:Preliminaries: the asymmetric Laplace distribution}


As well-documented in \citet{Kotz:Kozu:Podg:TheL:2012}, various skewed versions of the Laplace distribution have appeared in the literature.
Among them, the asymmetric Laplace (AL) distribution by \citet{Kozu:Podg:Asym:2000} is the reference one; for details about its properties and advantages, see \citet[][pages~134--136]{Kotz:Kozu:Podg:TheL:2012}. 
The classical and simplest parameterization of the pdf of the AL distribution -- given in Equation~(3.0.8) of \citet{Kotz:Kozu:Podg:TheL:2012} -- is 
\begin{equation}
f_{\text{AL}}\left(x;\mu,\beta,\kappa\right)=\frac{1}{\beta}\frac{\kappa}{1+\kappa^2} 
\begin{cases}
    e^{-\frac{\kappa}{\beta}\left(x-\mu\right)}  & \text{for $x\geq \mu$}, \\
    e^{-\frac{1}{\kappa\beta}\left(\mu-x\right)} & \text{for $x < \mu$},
  \end{cases}
\label{eq:asymmetric Laplace distribution}
\end{equation}
where $\mu\in \left(-\infty,\infty\right)$, $\beta>0$ and $\kappa>0$ are location, scale, and asymmetry parameters, respectively; if $X$ has the pdf in \eqref{eq:asymmetric Laplace distribution}, then we compactly write $X\sim \mathcal{AL}\left(\mu,\beta,\kappa\right)$.
The parameter $\kappa$ is scale invariant while $\mu$ coincides with the mode of the distribution.
Furthermore, if $\kappa=1$, then we obtain the classical (symmetric) Laplace distribution; in this case we write $X\sim \mathcal{L}\left(\mu,\beta\right)$.

The moments of greatest practical interest of $X\sim \mathcal{AL}\left(\mu,\beta,\kappa\right)$, namely the mean, variance, skewness, and kurtosis, are respectively given by
\begin{align}
	\text{E}\left(X\right)    & = \mu + \beta\left(\frac{1}{\kappa} - \kappa \right), \label{eq:AL mean}\\
	\var\left(X\right)  & = \beta^2 \left(\frac{1}{\kappa^2} + \kappa^2\right), \label{eq:AL var}\\
	\Skew\left(X\right) & = 2\frac{\displaystyle\frac{1}{\kappa^3}-\kappa^3}{\displaystyle\left(\frac{1}{\kappa^2}+\kappa^2\right)^\frac{3}{2}}, \label{eq:AL skew}\\
	\Kurt\left(X\right) & = 9 - \frac{12}{\displaystyle\left(\frac{1}{\kappa^2}+\kappa^2\right)^2}. \label{eq:AL kurt}
\end{align}
Among them, the mean in \eqref{eq:AL mean} is the only one depending on all the parameters of the AL distribution. 
The mean is given by the mode plus a term whose sign depends on the asymmetry parameter $\kappa$; this term can be negative (for $\kappa > 1$), null (for $\kappa = 1$), or positive (for $\kappa < 1$).
The variance does not depend on the mode $\mu$ and it is given by $\beta^2$ -- which is half of the variance $2\beta^2$ of $\mathcal{L}\left(\mu,\beta\right)$ --  multiplied by a positive factor only depending on $\kappa$.   
Instead, skewness and kurtosis only depend on $\kappa$, and their behavior as a function of $\kappa$ is displayed in \figurename~\ref{fig:AL skewness and kurtosis}.
\begin{figure}[!ht]
\centering
\includegraphics[width=0.5\textwidth]{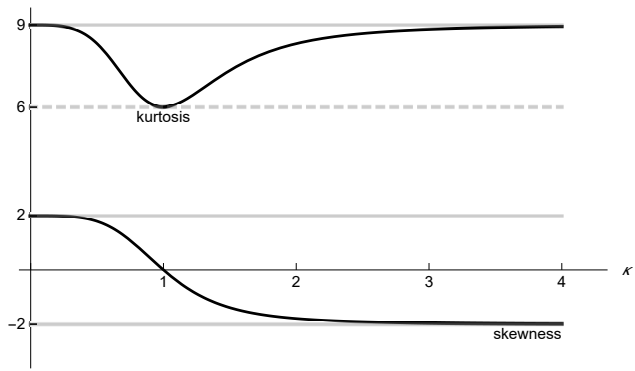}
\caption{
\footnotesize 
Skewness and kurtosis of $\mathcal{AL}\left(0,1,\kappa\right)$ as a function of $\kappa$.
Grey solid lines represent asymptotes. 
\label{fig:AL skewness and kurtosis}
}
\end{figure}
The skewness lies in the interval $\left(-2,2\right)$; as $\kappa$ increases, the corresponding value of $\text{Skew}\left(X\right)$ decreases monotonically.
As for the kurtosis, the AL distribution is leptokurtic and $\text{Kurt}\left(X\right)$ varies from 6 -- the value we have for the symmetric Laplace distribution ($\kappa = 1$) -- to 9 (the greatest value attained when either $\kappa \rightarrow 0$ or $\kappa \rightarrow \infty$).
For further details about moments and related quantities of $\mathcal{AL}\left(\mu,\beta,\kappa\right)$, see \citet[][Section~3.1.5]{Kotz:Kozu:Podg:TheL:2012}.

  %

\subsection{Probability density function}
\label{subsec:Probability density function}

In Definition~\ref{def:ALSM} we introduce the asymmetric Laplace scale mixture in terms of its pdf.
 
\begin{defi}[Probability density function]\label{def:ALSM}
A random variable $X$ is said to have the asymmetric Laplace scale mixture (ALSM) distribution with location $\mu\in \left(-\infty,\infty\right)$, scale $\beta>0$, asymmetry parameter $\kappa>0$ and tailedness $\btheta$, in symbols $X\sim \mathcal{ALSM}\left(\mu,\beta,\kappa,\btheta\right)$, if its pdf is given by
\begin{equation}
f_{\text{ALSM}}\left(x;\mu,\beta,\kappa,\btheta\right) 
=
\int_{S_h} f_{\text{AL}}\left(x;\mu,\beta/w,\kappa\right)h\left(w;\btheta\right)dw,
\quad -\infty < x < \infty,
\label{eq:asymmetric Laplace scale mixture}
\end{equation}
where $h\left(w;\btheta\right)$ is the mixing probability density (or mass) function of $W$, with support $S_h \subseteq \left(0,\infty\right)$, depending on the vector of parameters $\btheta$; in symbols, $W\sim \mathcal{H}\left(\btheta\right)$.
\end{defi}

According to \eqref{eq:asymmetric Laplace scale mixture}, the distribution of $X$ can be thought of as a composite distribution constructed by taking a finite/infinite set of AL component distributions with the same location $\mu$ and asymmetry parameter $\kappa$, but with a different scale $\beta/w$.
The component AL distributions are not taken uniformly from the set, but according to a set of ``weights'' determined by the distribution of $W$ (random counterpart of $w$).
In the context of modelling currency returns, the mixing variable $W$ could be interpreted as a shock that arises from new information and impacts the volatility of the market.
The resulting mixture is not itself an AL distribution due to a more flexible tail behavior governed by the tailedness parameter $\btheta$.
However, the distribution preserves the characteristic towering peak of the AL distribution.
Note that, if $W$ is degenerate in 1 ($W \equiv 1$), then $\mathcal{ALSM}\left(\mu,\beta,\kappa,\btheta\right)$ reduces to $\mathcal{AL}\left(\mu,\beta,\kappa\right)$.
While, if $\kappa=1$, then $\mathcal{ALSM}\left(\mu,\beta,\kappa,\btheta\right)$ reduces to the Laplace scale mixture $\mathcal{LSM}\left(\mu,\beta,\btheta\right)$ distribution proposed by \citet{Punz:Bagn:PhyA:2021}.

If we want model~\eqref{eq:asymmetric Laplace scale mixture} to embed $\mathcal{AL}\left(\mu,\beta,\kappa\right)$ as a special or limiting case, we need to consider a mixing random variable $W$ degenerating in 1 under a convenient choice of $\btheta$.
This is commonly required to models that aim to make more flexible the behavior of the tails of a conditional distribution (as, e.g., the
well-known normal scale mixture model). 
This allows to use inferential procedures such as the likelihood-ratio test to compare the ALSM with the nested AL distribution (cf.~Section~\ref{sec:Real data analysis}).

\subsection{Representations}
\label{subsec:Representations}

The ALSM admits the following useful representations as mixture of normal distributions (Proposition~\ref{prop:representation as N mixture}) and as mixture of AL distributions (Proposition~\ref{prop:representation as AL mixture}).

\begin{pro}[Mixture of normal distributions]\label{prop:representation as N mixture}
$X\sim \mathcal{ALSM}\left(\mu,\beta,\kappa,\btheta\right)$ admits the representation
\begin{equation}
X \stackrel{d}{=} \mu + \beta\left(\frac{1}{\kappa}-\kappa\right)\frac{V}{W} + \beta\frac{\sqrt{2V}}{W}Z,
\label{eq:representation as N mixture}
\end{equation}
where $Z$ is standard normal ($Z\sim\mathcal{N}\left(0,1\right)$) and $V$ is standard exponential ($V\sim\mathcal{E}\left(1\right)$).
Note that $Z$, $V$, and $W$ are independent.
\end{pro}

\begin{pro}[Mixture of AL distributions]\label{prop:representation as AL mixture}
$X\sim \mathcal{ALSM}\left(\mu,\beta,\kappa,\btheta\right)$ admits the representation
\begin{equation}
X \stackrel{d}{=} \mu + \beta \frac{Y}{W}, 
\label{eq:representation as AL mixture}
\end{equation}
where $Y\sim \mathcal{AL}\left(0,1,\kappa\right)$ is independent of $W$.

\end{pro}



\subsection{Moments}
\label{subsec:Moments}

Let $X\sim \mathcal{ALSM}\left(\mu,\beta,\kappa,\btheta\right)$.
For any integer $r > 0$, it can be proved -- starting from the representation given in Proposition~\ref{prop:representation as AL mixture}, from Eq.~(3.1.25) in \citet{kotz2001laplace}, and noticing that $Y$ and $W$ are stochastically independent -- that the $r$th moment of $X$ about $\mu$ is
\begin{align}
\text{E}\left[\left(X -\mu\right)^r\right] & = \text{E}\left[\left(Y-\mu\right)^r\right] \beta^r \text{E}\left(\frac{1}{W^r}\right) \nonumber \\
& = r!\left(\frac{\beta}{\kappa}\right)^r p_r\left(\kappa\right) \text{E}\left(\frac{1}{W^r}\right),
\label{eq:centered moment r}
\end{align}
where 
\begin{align}
p_r\left(\kappa\right)  & = \sum_{j=0}^{r} \left(-1\right)^j \kappa^{2j} \nonumber \\
& =  \frac{\displaystyle 1 + \left(-1\right)^r \kappa^{2(r+1)}}{1+\kappa^2}.
\label{eq:polynomial}
\end{align}
The moment in \eqref{eq:centered moment r} exists if $\text{E}\left(1/W^r\right)< \infty$.

\subsection{Absolute moments}
\label{subsec:Absolute moments}


Let $X\sim \mathcal{ALSM}\left(\mu,\beta,\kappa,\btheta\right)$.
For any $a>-1$, it can be proved -- starting from the representation given in Proposition~\ref{prop:representation as N mixture}, from Eq.~(3.1.26) in \citet{kotz2001laplace}, and noticing that $Y$ and $W$ are stochastically independent -- that the absolute moment of order $a$ of $X$ about $\mu$ is
\begin{align}
\text{E}\left[\left|X -\mu\right|^a\right] &=   \text{E}\left[\left|Y-\mu\right|^a\right] \beta^a E\left(\frac{1}{W^a}\right) \nonumber \\
& =  \left(\frac{\beta}{\kappa}\right)^a  \Gamma\left(a+1\right) \frac{\displaystyle 1 +   \kappa^{2(a+1)}}{1+\kappa^2} \text{E}\left(\frac{1}{W^a}\right).
\label{eq:absolute moment}
\end{align}
The moment in \eqref{eq:absolute moment} exists if $\text{E}\left(1/W^a\right) < \infty$.

\subsection{Mean, variance, skewness, and kurtosis}
\label{subsec:Mean, variance, skewness, and kurtosis}

Mean, variance, skewness, and kurtosis of $X\sim \mathcal{ALSM}\left(\mu,\beta,\kappa,\btheta\right)$ can be obtained from \eqref{eq:representation as AL mixture} and from the corresponding moments of $\mathcal{AL}\left(\mu,\beta,\kappa\right)$.
These moments can be expressed as a function of convenient moments of $1/W$, which depend on $\btheta$, in the following way:
\begin{align}
	\text{E}\left(X\right) &= \mu + \beta \text{E}\left( Y  \right) \text{E}\left( \frac{1}{W}  \right)\nonumber \\
&= \mu + \frac{\beta}{\kappa} p_1\left(\kappa\right) \text{E}\left( \frac{1}{W}  \right) \nonumber \\
&= \mu + \beta \left(\frac{1}{\kappa}-\kappa\right) \text{E}\left( \frac{1}{W}  \right)
\label{eq:ALSM mean}
\end{align}
\begin{align}
	\var\left(X\right)  &= \beta^2 \text{Var}\left(\frac{Y}{W}\right)\nonumber\\
&=   \frac{\beta^2}{\kappa^2}  \left[2p_2(\kappa)\text{E}\left(\frac{1}{W^2}\right)- p_1(\kappa)^2 \text{E}\left(\frac{1}{W}\right)^2\right], \nonumber \\
&=
\frac{\beta^2}{\kappa^2\left(1+\kappa^2\right)}  \left[2  \left(1+\kappa^6\right) \text{E}\left(\frac{1}{W^2}\right)- \left(1-\kappa^4\right)\left(1-\kappa^2\right)  \text{E}\left(\frac{1}{W}\right)^2\right]
\label{eq:ALSM var}
\end{align}
\begin{align}
	\Skew\left(X\right) &=   
2
\frac{
3p_3(\kappa)\text{E}\left(\frac{1}{W^3}\right) -3p_2(\kappa)p_1(\kappa)\text{E}\left(\frac{1}{W^2}\right)\text{E}\left(\frac{1}{W}\right) + p_1(\kappa)^3 \text{E}\left(\frac{1}{W}\right)^3}{\left[2p_2(\kappa)\text{E}\left(\frac{1}{W^2}\right)- p_1(\kappa)^2\text{E}\left(\frac{1}{W}\right)^2 \right]^\frac{3}{2}},
\nonumber \\
&= 2
\frac{
3\left(1+\kappa^2\right)^2\left(1-\kappa^8\right)\text{E}\left(\frac{1}{W^3}\right) 
-3\left(1+\kappa^2\right) \left(1+\kappa^6\right)\left(1-\kappa^4\right)\text{E}\left(\frac{1}{W^2}\right)\text{E}\left(\frac{1}{W}\right) 
+ \left(1-\kappa^8\right)^3 \text{E}\left(\frac{1}{W}\right)^3}
{\left[
2\left(1+\kappa^2\right) \left(1+\kappa^6\right)\text{E}\left(\frac{1}{W^2}\right)
- \left(1-\kappa^4\right)^2\text{E}\left(\frac{1}{W}\right)^2 \right]^\frac{3}{2}},
\label{eq:ALSM skew}
\end{align}
\begin{align}
	\Kurt\left(X\right) &=   
12
\frac{
2p_4(\kappa)\text{E}\left(\frac{1}{W^4}\right) 
-2p_3(\kappa)p_1(\kappa)\text{E}\left(\frac{1}{W^3}\right)\text{E}\left(\frac{1}{W}\right)
+ p_2(\kappa)^2 \text{E}\left(\frac{1}{W^2}\right)^2
}{
\left[2p_2(\kappa)\text{E}\left(\frac{1}{W^2}\right) - p_1(\kappa)^2\text{E}\left(\frac{1}{W}\right)^2 \right]^2
}
-3 \nonumber \\
	&=   
12
\frac{
2 \left(1+\kappa^2\right)\left(1+\kappa^{10}\right)\text{E}\left(\frac{1}{W^4}\right) 
-2\left(1-\kappa^8\right)\left(1-\kappa^4\right)\text{E}\left(\frac{1}{W^3}\right)\text{E}\left(\frac{1}{W}\right)
+ \left(1+\kappa^6\right)^2 \text{E}\left(\frac{1}{W^2}\right)^2
}{
\left[2\left(1+\kappa^6\right)\text{E}\left(\frac{1}{W^2}\right) -  \left(1-\kappa^2\right)\left(1-\kappa^4\right)\text{E}\left(\frac{1}{W}\right)^2 \right]^2
}
-3. \label{eq:ALSM kurt} 
\end{align}
As we can note, the mean depends on all the parameters of the ALSM distribution, the variance does not depend on $\mu$, while skewness and kurtosis only depend on $\kappa$ and $\btheta$.

\subsection{Mode}
\label{subsec:Mode}

In Theorem~\ref{theo:1} we show that the ALSM distribution is unimodal hump-shaped, with mode in $\mu$; the result is a slight modification of Theorem~1 in \citet{Punz:Bagn:Maru:Comp:2018} and \citet{Punz:Bagn:PhyA:2021}. 
\begin{theo}\label{theo:1}
The pdf of $X\sim \mathcal{ALSM}\left(\mu,\beta,\kappa,\btheta\right)$ is unimodal hump-shaped in $x=\mu$.
\end{theo}
\begin{proof}
The first derivative, with respect to $x$, of $f_{\text{ALSM}}\left(x;\mu,\beta,\kappa,\btheta\right)$ is
\begin{equation*}
  f_{\text{ALSM}}'\left(x;\mu,\beta,\kappa,\btheta\right) = \int_{S_h} f_{\text{AL}}'\left(x;\mu,\beta/w,\kappa\right)h\left(w;\btheta\right)dw, \quad x > 0,
\label{eq:derivp}
\end{equation*} 
where 
\begin{displaymath}
	f_{\text{AL}}'\left(x;\mu,\beta/w,\kappa\right) = 
	\frac{w^2}{\beta^2}\frac{\kappa^2}{1+\kappa^2} (\mu-x)
\begin{cases}
    e^{-\frac{w\kappa}{\beta}\left(x-\mu\right)}  & \text{for $x\geq \mu$}, \\
    e^{-\frac{w}{\kappa\beta}\left(\mu-x\right)} & \text{for $x < \mu$}.
  \end{cases}
\end{displaymath}
Hence, $f_{\text{AL}}'\left(x;\mu,\beta/w,\kappa\right)>0$ for $x<\mu$, $f_{\text{AL}}'\left(x;\mu,\beta/w,\kappa\right)<0$ for $x>\mu$, and $f_{\text{AL}}'\left(x;\mu,\beta/w,\kappa\right)=0$ for $x=\mu$.
Using these results and recalling that $h\left(w;\btheta\right)> 0$ for $w>0$, the theorem is straightforwardly proven.
\end{proof}

\section{Examples of asymmetric Laplace scale mixtures}
\label{sec:Examples of AL scale mixtures}

The ALSM family of distributions is extremely flexible.
In this section we introduce some members of the family obtained under a convenient choice of the mixing probability density (or mass) function $h\left(w;\btheta\right)$.
These members satisfy two requirements: 1) the existence of a closed-form pdf for the ALSM and 2) the peculiarity to nest $\mathcal{AL}\left(\mu,\beta,\kappa\right)$ under a suitable choice of $\btheta$.  

\subsection{Two-point asymmetric Laplace distribution}
\label{subsec:Two-point mixing distribution}

A simple but effective special case of model~\eqref{eq:asymmetric Laplace scale mixture} is obtained if we consider the following mixing two-point (TP) distribution 
\begin{equation}
W = \begin{cases}
      1 & \text{with probability $\theta_1$}\\
      1/\theta_2 & \text{with probability $1-\theta_1$},\\
 	\end{cases} 
\label{eq:W Bernoulli}
\end{equation}
where $\theta_1\in\left(0,1\right)$ and $\theta_2>1$.
The probability mass function of $W$ is
\begin{equation}
h\left(w;\btheta\right)=\theta_1^{\frac{w-1/\theta_2}{1-1/\theta_2}}\left(1-\theta_1\right)^{\frac{1-w}{1-1/\theta_2}},\quad w\in\left\{1/\theta_2,1\right\},
\label{eq:pmf Bernoulli}
\end{equation}
where $\btheta=\left(\theta_1,\theta_2\right)'$; in symbols, $W\sim \mathcal{TP}_{\left\{1/\theta_2,1\right\}}\left(\theta_1\right)$.
Using \eqref{eq:pmf Bernoulli}, model~\eqref{eq:asymmetric Laplace scale mixture} becomes
\begin{equation}
f_{\text{ALSM}}\left(x;\mu,\beta,\kappa,\btheta\right)=\theta_1 f_{\text{AL}}\left(x;\mu,\beta,\kappa\right) + \left(1-\theta_1\right) f_{\text{AL}}\left(x;\mu,\theta_2\beta,\kappa\right),\quad -\infty<x<\infty.
\label{eq:contaminated AL pdf}
\end{equation} 
Hereafter, we will refer to the resulting distribution as two-point asymmetric Laplace (TP-AL).

If $W\sim \mathcal{TP}_{\left\{1/\theta_2,1\right\}}\left(\theta_1\right)$, then the $r$th moment of $1/W$ is 
\begin{equation}
\text{E}\left(\frac{1}{W^r}\right) = \theta_1 + \left(1-\theta_1\right)\theta_2^r.
\label{eq:Bern rth moment}
\end{equation}
If we put in \eqref{eq:absolute moment}--\eqref{eq:ALSM kurt} the first four raw moments computed according to \eqref{eq:Bern rth moment}, mean, variance, skewness and kurtosis of the TP-AL distribution become
\begin{equation}
\text{E}\left(X\right) = \mu + \beta\left(\frac{1}{\kappa} - \kappa \right) \left[\theta_1 + \left(1-\theta_1\right)\theta_2\right], 
\label{eq:BernAL mean}
\end{equation}
\begin{equation}
\var\left(X\right) = \frac{\beta^2}{\kappa^2} 
\left\{\frac{2 \left[\theta_1 + \left(1-\theta_1\right)\theta_2\right] \left(1+\kappa^6\right)}{1+\kappa^2}
-\frac{\left[\theta_1 + \left(1-\theta_1\right)\theta_2\right]^2 \left(1-\kappa^4\right)^2}{\left(1+\kappa^2\right)^2}\right\},
\label{eq:BernAL variance}
\end{equation}
\begin{equation}
\Skew\left(X\right) = 
\frac{
2 \left\{\left[\theta_1 + \left(1-\theta_1\right)\theta_2\right]^2 \left(1-\kappa^4\right)^3-3 \left[\theta_1 + \left(1-\theta_1\right)\theta_2\right]  \left(1+\kappa^2\right) \left(1+\kappa^6\right) \left(1-\kappa^4\right)+3 \left(1+\kappa^2\right)^2 \left(1-\kappa^8\right)\right\}
}{
\sqrt{\theta_1 + \left(1-\theta_1\right)\theta_2} 
\left\{2 \left(1+\kappa^2\right) \left(1+\kappa^6\right) - \left[\theta_1 + \left(1-\theta_1\right)\theta_2\right] \left(1-\kappa^4\right)^2\right\}^{3/2}
},
\label{eq:BernAL skewness}
\end{equation}
and
\begin{equation}
\Kurt\left(X\right) = 
\frac{
12 \left(1+\kappa^2\right)^2 \left\{\left[\theta_1 + \left(1-\theta_1\right)\theta_2\right]^2 \left(1+\kappa^6\right)^2
-
2 \left[\theta_1 + \left(1-\theta_1\right)\theta_2\right] \left(1-\kappa^4\right) \left(1-\kappa^8\right)+2 \left(1+\kappa^2\right) \left(1+\kappa^{10}\right)\right\}
}{
\left[\theta_1 + \left(1-\theta_1\right)\theta_2\right] \left\{2 \left(1+\kappa^2\right) \left(1+\kappa^6\right) - \left[\theta_1 + \left(1-\theta_1\right)\theta_2\right] \left(1-\kappa^4\right)^2\right\}^2
}
-3.
\label{eq:BernAL kurtosis}
\end{equation}

With respect to $\mathcal{AL}\left(\mu,\beta,\kappa\right)$, which can be considered as the reference distribution for the model (for a discussion about the concept of reference distribution, see \citealp{Davi:Gath:Thei:1993}), the additional parameters $\theta_1$ and $\theta_2$ have an interpretation of practical interest.
\begin{description}
	\item[--] $\theta_1$ is the proportion of points from the reference distribution.
For applications of model~\eqref{eq:contaminated AL pdf} in robustness studies, it could be natural to assume, and we will do that hereafter, that $\theta_1\in\left(1/2,1\right)$, i.e., at least half of the observations come from the reference distribution (\citealp[][p.~250]{Henn:Fixe:2002}, \citealp{Punz:McNi:Robu:2016}, \citealp{Mazz:Punz:Mode:2019}, \citealp{templ2019evaluation} and \citealp{Punz:Bagn:PhyA:2021}).  
	\item[--] $\theta_2$ denotes the degree of contamination and, because of the assumption $\theta_2>1$, it can be meant as the increase in variability due to the points which do not come from the reference distribution; hence, it is an inflation parameter.
\end{description}
Model~\eqref{eq:contaminated AL pdf} reduces to $\mathcal{AL}\left(\mu,\beta,\kappa\right)$ when $\theta_1\rightarrow 1$ and $\theta_2\rightarrow 1$, and can also be seen as the natural asymmetric counterpart of the contaminated Laplace distribution introduced by \citet{Punz:Bagn:PhyA:2021}.
For this reason, we will also refer to the TP-AL distribution as contaminated asymmetric Laplace distribution.

An interesting characteristic of model~\eqref{eq:contaminated AL pdf} is the possibility to determine whether a generic observation $x$ comes or not from the reference distribution via the \textit{a~posteriori} probability
\begin{equation}
P\left(\text{$x$ comes from $\mathcal{AL}\left(\mu,\beta,\kappa\right)$}\left|\mu,\beta,\kappa,\btheta\right.\right)=\frac{\theta_1f_{\text{AL}}\left(x;\mu,\beta,\kappa\right)}{f_{\text{ALSM}}\left(x;\mu,\beta,\kappa,\btheta\right)}.
\label{eq:probability good}
\end{equation}
Specifically, $x$ will be considered coming from the reference distribution if the probability in \eqref{eq:probability good} is greater than 1/2.

\figurename~\ref{fig:TPAL skew and kurt1} shows examples of behaviors of $\Skew\left(X\right)$ and $\Kurt\left(X\right)$, as functions of $\kappa$, at various levels of $\theta_1$, with $\theta_2=5$.
\begin{figure}[!ht]
\centering
\subfigure[Skewness \label{fig:TPALskew1}]
{\resizebox{0.48\textwidth}{!}{\includegraphics{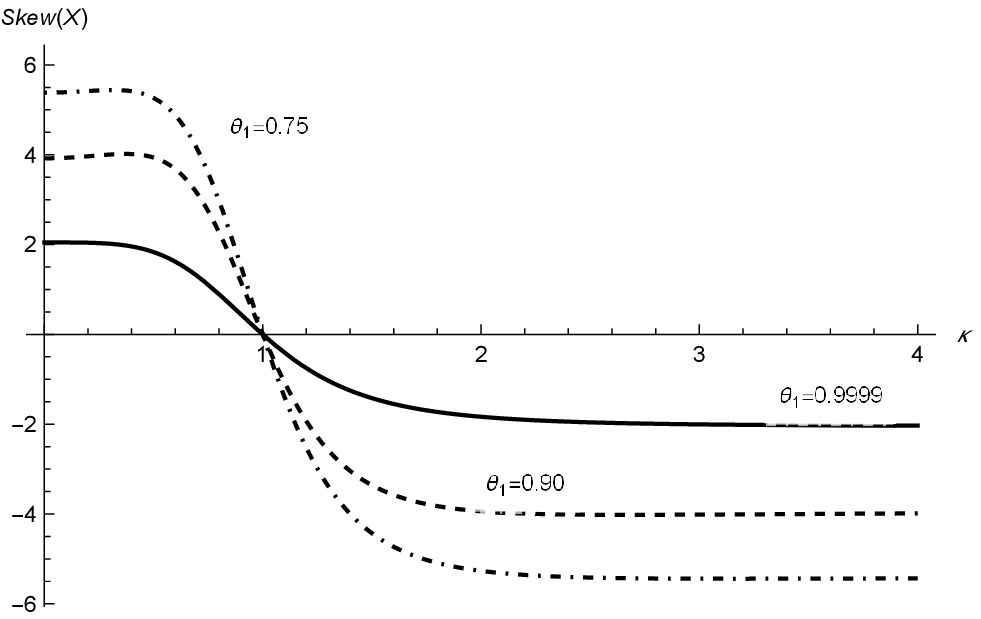}}}
\subfigure[Kurtosis\label{fig:TPALkurt1}]
{\resizebox{0.48\textwidth}{!}{\includegraphics{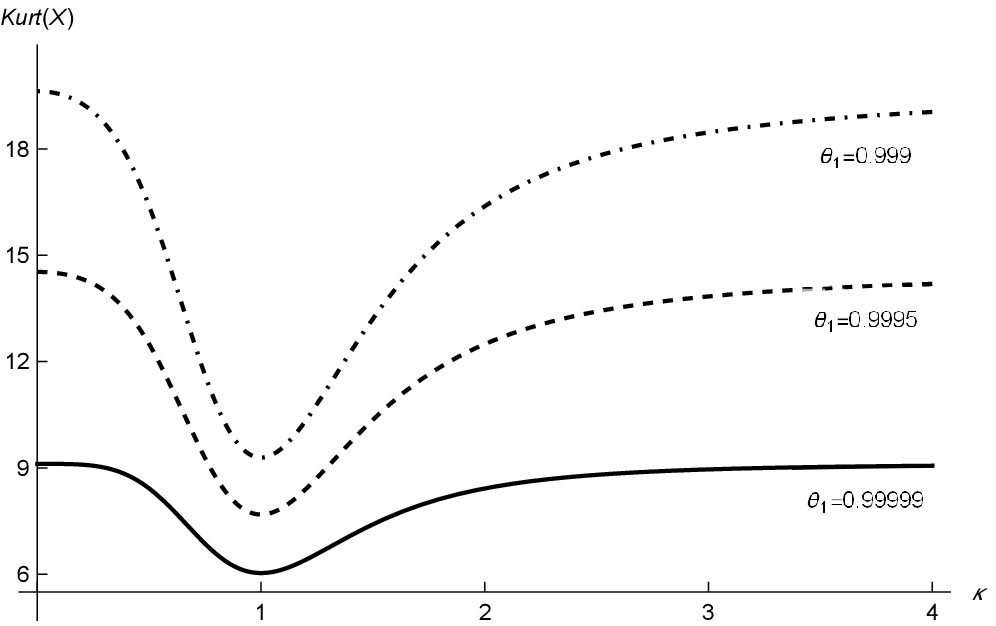}}}
\caption{
\footnotesize 
Examples of behavior of $\Skew\left(X\right)$ (on the left) and $\Kurt\left(X\right)$ (on the right), as functions of $\kappa$, at various levels of $\theta_1$ for the TP-AL distribution, with $\theta_2=5$.
\label{fig:TPAL skew and kurt1}
}
\end{figure}
From \figurename~\ref{fig:TPALskew1} we realize that: 1) values of $\theta_1$ close to 1 tend to produce the plot of $\Skew\left(X\right)$ we obtained for $\mathcal{AL}\left(\mu,\beta,\kappa\right)$ (refer to \figurename~\ref{fig:AL skewness and kurtosis}); and 2) as $\theta_1$ decreases (to a minimum value of 0.5), the range of possible values of $\Skew\left(X\right)$ increases.
Similarly, from \figurename~\ref{fig:TPALkurt1} we realize that, fixed $\theta_2$: 1) large values of $\theta_2$ (with a maximum of 1) tend to produce the plot of $\Kurt\left(X\right)$ we obtained for $\mathcal{AL}\left(\mu,\beta,\kappa\right)$ (refer to \figurename~\ref{fig:AL skewness and kurtosis}); 
and 2) $\kappa$ kept fixed, the lower the value $\theta_1$ (with a minimum of 0.5), the higher the kurtosis.
This means that the curve obtained for $\theta_1\rightarrow 1$ acts as a lower bound.

\figurename~\ref{fig:TPAL skew and kurt2} shows examples of behaviors of $\Skew\left(X\right)$ and $\Kurt\left(X\right)$, as functions of $\kappa$, at various levels of $\theta_2$, with $\theta_1=0.95$.
\begin{figure}[!ht]
\centering
\subfigure[Skewness \label{fig:TPALskew2}]
{\resizebox{0.48\textwidth}{!}{\includegraphics{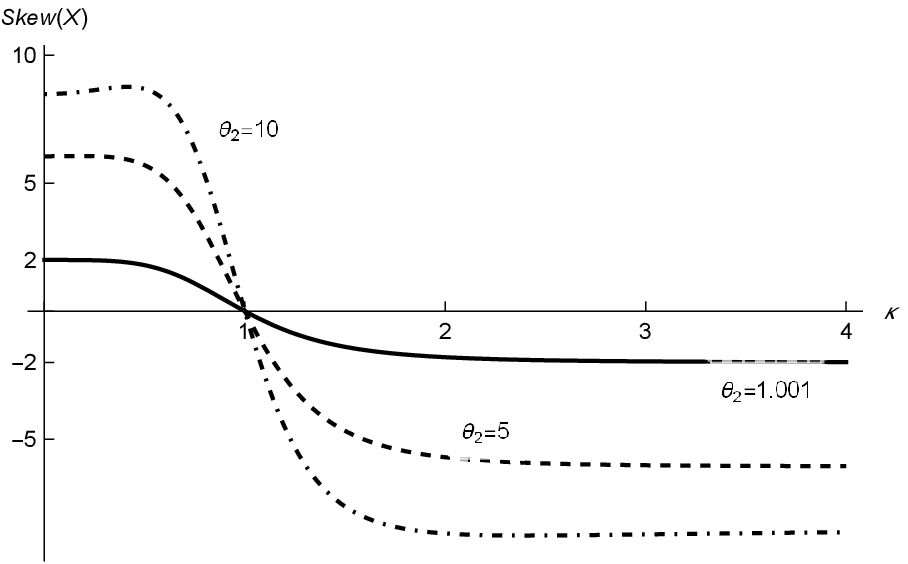}}}
\subfigure[Kurtosis\label{fig:TPALkurt2}]
{\resizebox{0.48\textwidth}{!}{\includegraphics{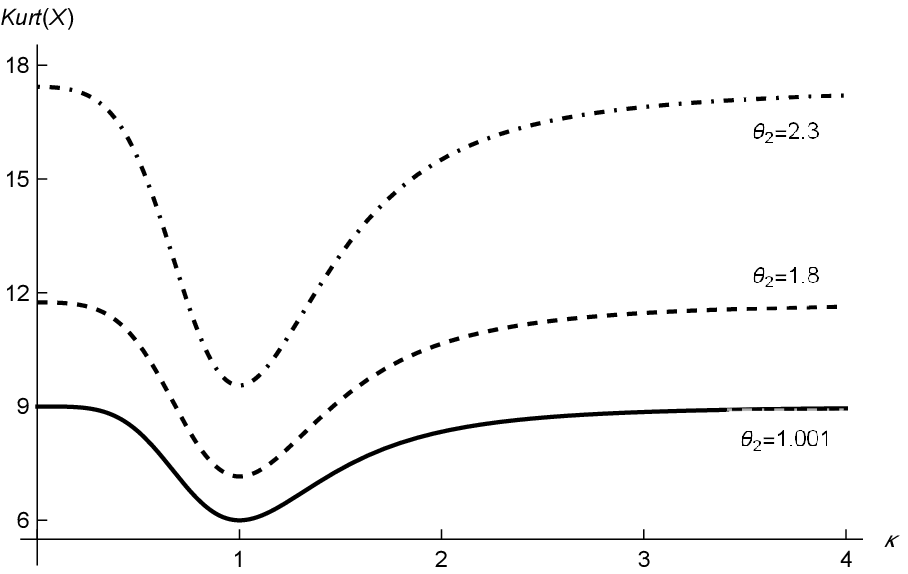}}}
\caption{
\footnotesize 
Examples of behavior of $\Skew\left(X\right)$ (on the left) and $\Kurt\left(X\right)$ (on the right), as functions of $\kappa$, at various levels of $\theta_2$ for the TP-AL distribution, with $\theta_1=0.95$.
\label{fig:TPAL skew and kurt2}
}
\end{figure}
From \figurename~\ref{fig:TPALskew2} we realize that: 1) values of $\theta_2$ close to 1 tend to produce the plot of $\Skew\left(X\right)$ we obtained for $\mathcal{AL}\left(\mu,\beta,\kappa\right)$ (refer to \figurename~\ref{fig:AL skewness and kurtosis}); and 2) as $\theta_2$ grows, the range of possible values of $\Skew\left(X\right)$ increases.
Similarly, from \figurename~\ref{fig:TPALkurt2} we realize that: 1) small values of $\theta$ (with a minimum of 1) tend to produce the plot of $\Kurt\left(X\right)$ we obtained for $\mathcal{AL}\left(\mu,\beta,\kappa\right)$ (refer to \figurename~\ref{fig:AL skewness and kurtosis}); 
and 2) $\kappa$ kept fixed, the higher the value $\theta_2$, the higher the kurtosis.
This means that the curve obtained for $\theta_2\rightarrow 1$ acts as a lower bound.

\subsection{Shifted exponential Laplace distribution}
\label{subsec:Shifted exponential mixing distribution}

Let
\begin{equation}
h\left(w;\theta\right)=\theta e^{-\theta\left(w-1\right)},\quad w>1,
\label{eq:pdf shifted exponential}
\end{equation}
with rate parameter $\theta>0$, be the pdf of the shifted exponential, in symbols $W\sim \mathcal{SE}_{\left(1,\infty\right)}\left(\theta\right)$.
When the pdf in \eqref{eq:pdf shifted exponential} is considered as mixing density in model~\eqref{eq:asymmetric Laplace scale mixture}, the pdf of the ALSM becomes
\begin{equation}
f_{\text{ALSM}}\left(x;\mu,\beta,\kappa,\theta\right) = \frac{\theta}{\beta}\frac{\kappa}{1+\kappa^2} 
\begin{cases}
    \displaystyle\frac{1+\theta+\frac{\kappa}{\beta}\left(x-\mu\right)}{\left[\theta+\frac{\kappa}{\beta}\left(x-\mu\right)\right]^2}e^{-\frac{\kappa}{\beta}\left(x-\mu\right)}  & \text{for $x\geq \mu$}, \\[6mm]
    \displaystyle\frac{1+\theta+\frac{1}{\beta\kappa}\left(\mu-x\right)}{\left[\theta+\frac{1}{\beta\kappa}\left(\mu-x\right)\right]^2}e^{-\frac{1}{\kappa\beta}\left(\mu-x\right)} & \text{for $x < \mu$}.
  \end{cases}
\label{eq:SE-AL distribution}
\end{equation}
We will refer to the resulting distribution as shifted exponential asymmetric Laplace (SE-AL).
Model~\eqref{eq:SE-AL distribution} reduces to $\mathcal{AL}\left(\mu,\beta,\kappa\right)$ when $\theta \rightarrow \infty$.
Moreover, if $\kappa=1$, then we obtain the SE-Laplace distribution, proposed by \citet{Punz:Bagn:PhyA:2021}, as a special case.

If $W\sim \mathcal{SE}_{\left(1,\infty\right)}\left(\theta\right)$, then the $r$th raw moment of $1/W$ is
\begin{equation}
\text{E}\left(\frac{1}{W^r}\right) =
\theta  e^{\theta} E_r(\theta),
\label{eq:SE moments 1/W}
\end{equation}
where $E_{n}\left(z\right)=\int_1^{\infty}e^{-zt}/t^ndt$ is the exponential integral function.
By substituting in \eqref{eq:ALSM mean}--\eqref{eq:ALSM kurt} the first four raw moments computed according to \eqref{eq:SE moments 1/W}, mean, variance, skewness and kurtosis of the SE-AL distribution become
\begin{equation}
\text{E}\left(X\right) = \mu + \beta \left(\frac{1}{\kappa}-\kappa \right) \theta e^{\theta} E_1(\theta),
\label{eq:SEAL mean}
\end{equation}
\begin{equation}
\var\left(X\right) 
=
  \frac{\beta^2\theta e^{\theta}}{\kappa^2\left(1+\kappa^2\right)}  \left[2 
\left(1+\kappa^6\right)                           E_2(\theta )
- \theta e^{\theta}\left(1-\kappa^4\right)\left(1-\kappa^2\right)  E_1(\theta )^2\right],
\label{eq:SEAL variance}
\end{equation}
\begin{equation}
\Skew\left(X\right) =   
\frac{
  3\left(1+\kappa^2\right)^2\left(1-\kappa^8\right)E_3(\theta ) 
  -3 \theta e^{\theta} \left(1+\kappa^2\right) \left(1+\kappa^6\right)\left(1-\kappa^4\right)E_2(\theta )E_1(\theta )
  + \left(1-\kappa^8\right)^3 \theta^2 e^{2\theta} E_1(\theta )^3}
{\sqrt{\theta e^{\theta}}\left[
  2\left(1+\kappa^2\right) \left(1+\kappa^6\right)E_2(\theta )
  - \theta e^{\theta} \left(1-\kappa^4\right)^2 E_1(\theta )^2 \right]^\frac{3}{2}},
\label{eq:SEAL skewness}
\end{equation}
and
\begin{equation}
\Kurt\left(X\right) = \frac{  12}{\theta e^\theta  }  \cdot
\frac{
2\left(1+\kappa^2\right) \left(1+\kappa^{10}\right) E_4(\theta ) 
-2  \theta e^{\theta }   \left(1-\kappa ^8\right)  \left(1-\kappa^4\right) E_1(\theta ) E_3(\theta )
e^{\theta } \theta  \left(1+\kappa^6 \right)^2 E_2(\theta )^2
} 
{ \left[e^{\theta } \theta \left( 1-\kappa^2 \right) \left(1-\kappa ^4\right)  E_1(\theta )^2-2 \left( 1+\kappa^6 \right) E_2(\theta )\right]^2}-3
\label{eq:SEAL kurtosis}
\end{equation}

\figurename~\ref{fig:SEAL skew and kurt} shows examples of behaviors of $\Skew\left(X\right)$ and $\Kurt\left(X\right)$, as functions of $\kappa$, at various levels of $\theta$.
\begin{figure}[!ht]
\centering
\subfigure[Skewness \label{fig:SEALskew}]
{\resizebox{0.48\textwidth}{!}{\includegraphics{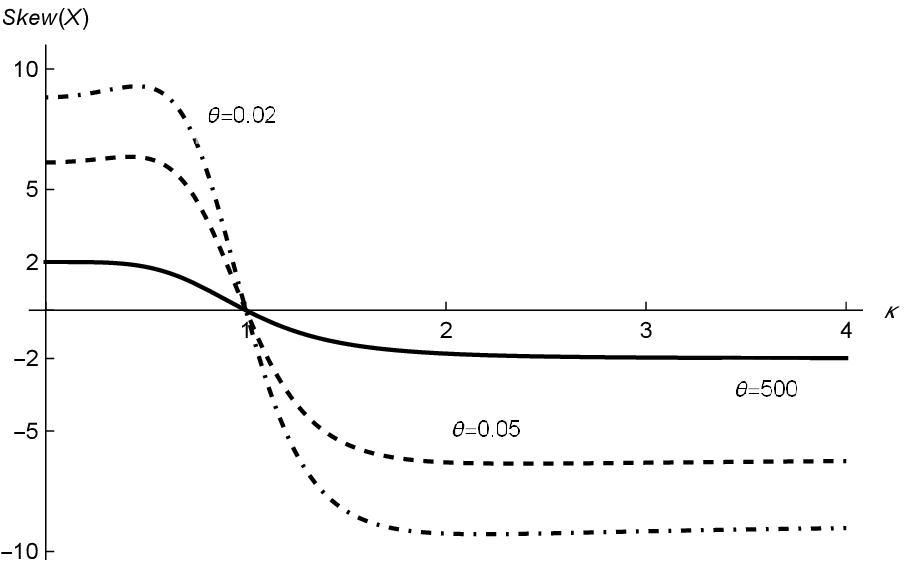}}}
\subfigure[Kurtosis\label{fig:SEALkurt}]
{\resizebox{0.48\textwidth}{!}{\includegraphics{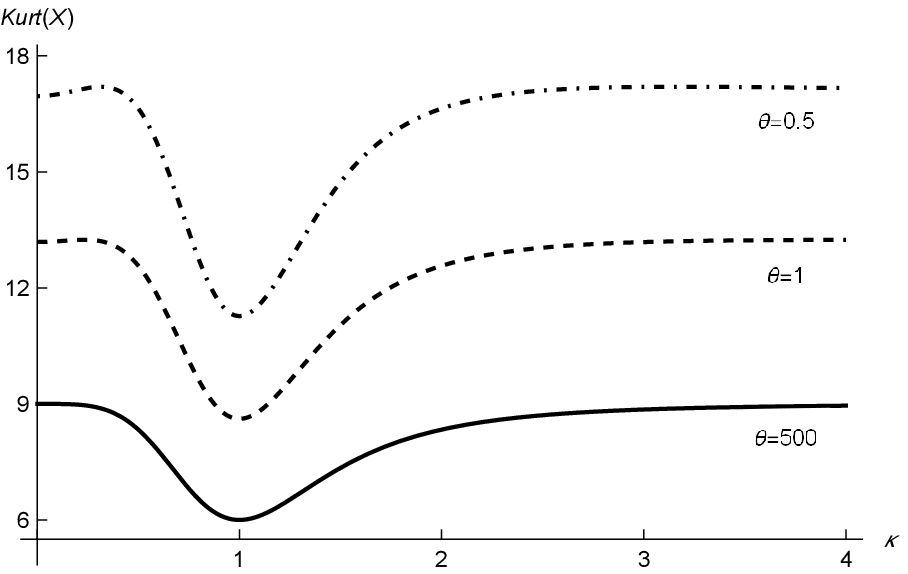}}}
\caption{
\footnotesize 
Examples of behavior of $\Skew\left(X\right)$ (on the left) and $\Kurt\left(X\right)$ (on the right), as functions of $\kappa$, at various levels of $\theta$ for the SE-AL distribution.
\label{fig:SEAL skew and kurt}
}
\end{figure}
From \figurename~\ref{fig:SEALskew} we realize that: 1) large values of $\theta$ tend to produce the plot of $\Skew\left(X\right)$ we obtained for $\mathcal{AL}\left(\mu,\beta,\kappa\right)$ (refer to \figurename~\ref{fig:AL skewness and kurtosis}); and 2) as $\theta$ decreases, the range of possible values of $\Skew\left(X\right)$ increases.
Moreover, to ``significantly'' modify the behavior of $\Skew\left(X\right)$ with respect to the $\mathcal{AL}\left(\mu,\beta,\kappa\right)$ case, we need low values of $\theta$ (in the examples, $\theta=0.05$ and $\theta=0.02$). 
Similarly, from \figurename~\ref{fig:SEALkurt} we realize that: 1) large values of $\theta$ tend to produce the plot of $\Kurt\left(X\right)$ we obtained for $\mathcal{AL}\left(\mu,\beta,\kappa\right)$ (refer to \figurename~\ref{fig:AL skewness and kurtosis}); 
and 2) $\kappa$ kept fixed, the lower the value $\theta$, the higher the kurtosis.
This means that the curve obtained for $\theta\rightarrow \infty$ acts as a lower bound.
Also in this case, to ``significantly'' modify the behavior of $\Kurt\left(X\right)$ with respect to the $\mathcal{AL}\left(\mu,\beta,\kappa\right)$ case, we need low values of $\theta$ (in the examples, $\theta=1$ and $\theta=0.5$). 

\subsection{Unimodal gamma asymmetric Laplace distribution}
\label{subsec:Unimodal gamma asymmetric Laplace distribution}

Let
\begin{equation}
h\left(w;\theta\right)=\frac{w^{\frac{1}{\theta}}e^{-\frac{w}{\theta}}}{\theta^{\frac{1}{\theta}+1}\Gamma\left(\frac{1}{\theta}+1\right)},\quad w>0,
\label{eq:Unimodal Gamma}
\end{equation}
with $\theta>0$, be the pdf of the unimodal gamma distribution \citep{Chen:prob:2000} with mode in $w=1$.
In symbols, $W\sim \mathcal{UG}\left(\theta\right)$; see also \citet{Bagn:Punz:Fine:2013}.
When the pdf in \eqref{eq:Unimodal Gamma} is considered as mixing density in model~\eqref{eq:asymmetric Laplace scale mixture}, the pdf of the ALSM becomes
\begin{equation}
f_{\text{ALSM}}\left(x;\mu,\beta,\kappa,\theta\right) = \frac{1+\theta}{\beta}\frac{\kappa}{1+\kappa^2} 
\begin{cases}
    \left[1+\frac{\kappa\theta}{\beta}\left(x-\mu\right)\right]^{-\frac{1}{\theta}- 2 } & \text{for $x\geq \mu$,} \\[3mm]
    \left[1+\frac{\theta}{\beta\kappa}\left(\mu-x\right)\right]^{-\frac{1}{\theta}- 2 } & \text{for $x < \mu$.}
  \end{cases}
\label{eq:UG-AL distribution}
\end{equation}
We refer to model~\eqref{eq:UG-AL distribution} as unimodal gamma asymmetric Laplace (UG-AL) distribution.
Model~\eqref{eq:UG-AL distribution} reduces to $\mathcal{AL}\left(\mu,\beta,\kappa\right)$ when $\theta \rightarrow 0$.
Moreover, if $\kappa=1$, then we obtain the UG-Laplace distribution, proposed by \citet{Punz:Bagn:PhyA:2021}, as a special case.


If $W\sim \mathcal{UG}\left(\theta\right)$, then the $r$th raw moment of $1/W$ is
\begin{align}
\text{E}\left(\frac{1}{W^r}\right) & = \frac{\Gamma\left(\frac{1}{\theta}+1-r\right)}{\theta^{r-1}\Gamma\left(\frac{1}{\theta}\right)}, \nonumber \\
  & = \frac{1}{\displaystyle \prod_{j=1}^{r}\left(1+\theta-j \theta  \right)}
\label{eq:UG moments 1/W}
\end{align}
which exists if $\theta < 1/(r-1)$.
By substituting in \eqref{eq:ALSM mean}--\eqref{eq:ALSM kurt} the first four raw moments in \eqref{eq:UG moments 1/W}, mean, variance, skewness and kurtosis of the UG-AL distribution become
\begin{equation}
\text{E}\left(X\right) = \mu + \beta \left(\frac{1}{\kappa}-\kappa\right),
\label{eq:UGAL mean}
\end{equation}
\begin{equation}
\var\left(X\right) = \frac{\beta^2 \left[\theta  \left(1-\kappa^2\right)^2+\kappa^4+1\right]}{\kappa^2(1-\theta) },
\label{eq:UGAL variance}
\end{equation}
which exists if $0<\theta<1$,
\begin{equation}
\Skew\left(X\right)=\frac{2 \sqrt{1-\theta } (2 \theta +1)  \left(1-\kappa^2\right) \left[\theta  \left( 1-\kappa^2\right)^2+\kappa^4+\kappa^2+1\right]}{(2 \theta -1) \left[\theta  \left(1-\kappa^2\right)^2+\kappa^4+1\right]^{3/2}}
\label{eq:UGAL skewness}
\end{equation}
which exists if $0<\theta<1/2$, and
\begin{equation}
\Kurt\left(X\right) = 
12 \frac{  (\theta +1) \kappa ^8-4 \theta  \kappa ^6+\left[\theta  (6 \theta -1)+1\right] \kappa ^4-4 \theta  \kappa ^2+\theta +1 }{\left[\theta  (6 \theta -5)+1 \right] \left[\theta  \left(1-\kappa ^2\right)^2+\kappa ^4+1\right]^2}-3
\label{eq:UGAL kurtosis}
\end{equation}
which exists if $0<\theta<1/3$.


\figurename~\ref{fig:UGAL skew and kurt} shows examples of behaviors of $\Skew\left(X\right)$ and $\Kurt\left(X\right)$, as functions of $\kappa$, at various levels of $\theta$.
\begin{figure}[!ht]
\centering
\subfigure[Skewness \label{fig:UGALskew}]
{\resizebox{0.48\textwidth}{!}{\includegraphics{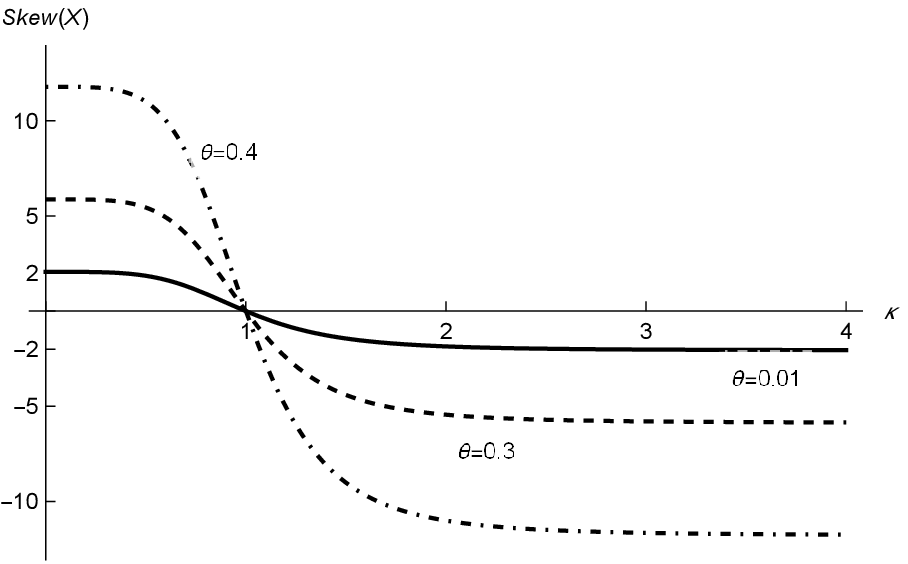}}}
\subfigure[Kurtosis\label{fig:UGALkurt}]
{\resizebox{0.48\textwidth}{!}{\includegraphics{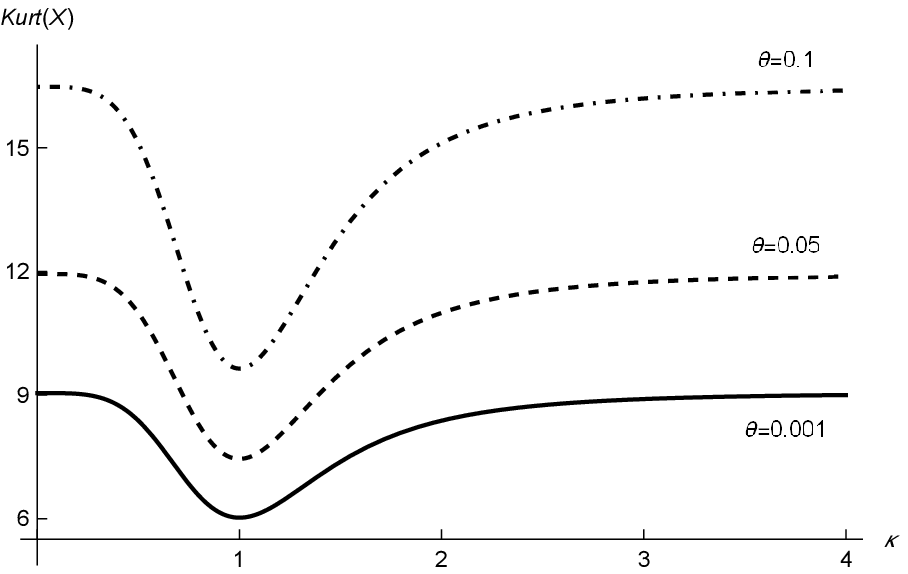}}}
\caption{
\footnotesize 
Examples of behavior of $\Skew\left(X\right)$ (on the left) and $\Kurt\left(X\right)$ (on the right), as functions of $\kappa$, at various levels of $\theta$ for the UG-AL distribution.
\label{fig:UGAL skew and kurt}
}
\end{figure}
From \figurename~\ref{fig:UGALskew} we realize that: 1) small values of $\theta$ tend to produce the plot of $\Skew\left(X\right)$ we obtained for $\mathcal{AL}\left(\mu,\beta,\kappa\right)$ (refer to \figurename~\ref{fig:AL skewness and kurtosis}); and 2) as $\theta$ grows from 0 to 1/2, the range of possible values of $\Skew\left(X\right)$ increases.
Similarly, from \figurename~\ref{fig:UGALkurt} we realize that: 1) small values of $\theta$ tend to produce the plot of $\Kurt\left(X\right)$ we obtained for $\mathcal{AL}\left(\mu,\beta,\kappa\right)$ (refer to \figurename~\ref{fig:AL skewness and kurtosis}); 
and 2) $\kappa$ kept fixed, as $\theta$ grows from 0 to 1/3, the kurtosis increases.
This means that the curve obtained for $\theta\rightarrow 0$ acts as a lower bound.

For completeness, and in analogy with the genesis of the $t$ distribution as a normal scale mixture, in \ref{app:Gamma asymmetric Laplace distribution} we give an alternative parameterization of the UG-AL distribution using a mixing gamma distribution with shape and rate equal to $\theta/2$.

\subsection{Inverse Gaussian asymmetric Laplace distribution}
\label{subsec:Inverse Gaussian asymmetric Laplace distribution}

Let
\begin{equation}
h\left(w;\theta\right)=\sqrt{\frac{3\theta+1}{2\pi\theta w^3}}e^{-\frac{\left(w-\sqrt{3\theta+1}\right)^2}{2\theta w}},\quad w>0,
\label{eq:Reparameterized Inverse Gaussian}
\end{equation}
with $\theta>0$, be the pdf of the reparameterized inverse Gaussian distribution, introduced by \citet{Punz:Anew:2019}, with mode in $w=1$; see also \citet{Punz:Bagn:Maru:Comp:2018}.
If $W$ has the pdf in \eqref{eq:Reparameterized Inverse Gaussian}, then we compactly write $W\sim \mathcal{IG}\left(\theta\right)$.  
When the pdf in \eqref{eq:Reparameterized Inverse Gaussian} is considered as mixing density in model~\eqref{eq:asymmetric Laplace scale mixture}, the pdf of the ALSM becomes
\begin{equation}
f_{\text{ALSM}}\left(x;\mu,\beta,\kappa,\theta\right) = \frac{\sqrt{1+3\theta}}{\beta} \frac{\kappa}{1+\kappa^2} 
\begin{cases}
\displaystyle\frac{e^{\frac{\sqrt{1+3\theta}}{\theta}\left[1-\sqrt{1+2\frac{\kappa\theta}{\beta}\left(x-\mu\right)}\right]}}{\sqrt{1+2\frac{\kappa\theta}{\beta}\left(x-\mu\right)}} & \text{for $x\geq \mu$,} \\[6mm]
    \displaystyle\frac{e^{\frac{\sqrt{1+3\theta}}{\theta}\left[1-\sqrt{1+2\frac{\theta}{\beta\kappa}\left(\mu-x\right)}\right]}}{\sqrt{1+2\frac{\theta}{\beta\kappa}\left(\mu-x\right)}}  & \text{for $x < \mu$.}
  \end{cases}
\label{eq:IG-AL distribution}
\end{equation}
We call inverse Gaussian asymmetric Laplace (IG-AL) the resulting distribution.
Model~\eqref{eq:IG-AL distribution} reduces to $\mathcal{AL}\left(\mu,\beta,\kappa\right)$ when $\theta \rightarrow 0$.
Moreover, if $\kappa=1$, then we obtain the IG-Laplace distribution, proposed by \citet{Punz:Bagn:PhyA:2021}, as a special case.

When $W\sim \mathcal{IG}\left(\theta\right)$, the $r$th raw moment of $1/W$ is
\begin{equation}
\text{E}\left(\frac{1}{W^r}\right) = 
\sqrt{\frac{2}{\theta \pi}} \left(3\theta+1\right)^{\frac{1-2r}{4}} 
e^{\frac{\sqrt{3\theta+1}}{\theta}}
K_{\frac{1}{2}+r}\left(\frac{\sqrt{3\theta+1}}{\theta}\right)
,
\label{eq:IG moments 1/W}
\end{equation}
where $K_{\lambda}\left(x\right)$ is the modified Bessel function of the third kind; for details, see \cite{Abra:Steg:Hand:1965} and \cite{watson1995treatise}.
By substituting in \eqref{eq:ALSM mean}--\eqref{eq:ALSM kurt} the first four raw moments in \eqref{eq:IG moments 1/W}, mean, variance, skewness and kurtosis of the IG-AL distribution become
\begin{equation}
\text{E}\left(X\right) = \mu + \beta \left(\frac{1}{\kappa}-\kappa\right) \frac{\theta +\sqrt{3\theta + 1}}{3\theta + 1},
\label{eq:IGAL mean}
\end{equation}
\begin{equation}
\var\left(X\right) = \frac{\beta^2}{\kappa^2 (3\theta +1) } \left[\frac{2 \left(\frac{3 \theta^2}{3 \theta +1}+\frac{3 \theta }{\sqrt{3 \theta +1}}+1\right) \left(\kappa^6+1\right)}{\left(\kappa^2+1\right)}-\frac{\left(\frac{\theta }{\sqrt{3 \theta +1}}+1\right)^2 \left(1-\kappa^4\right)^2}{\left(\kappa^2+1\right)^2}\right],
\label{eq:IGAL variance}
\end{equation}

\begin{align}
\Skew\left(X\right) = &
\frac{2 \left(1-\kappa ^2\right)}{\left(\theta ^2 \left(5 \kappa ^4-4 \kappa ^2+5\right)+\theta  \left(\left(4 \sqrt{3 \theta +1}+3\right) \kappa ^4-2 \sqrt{3 \theta +1} \kappa ^2+4 \sqrt{3 \theta +1}+3\right)+\kappa ^4+1\right)^{3/2}}  \nonumber\\
& \times  \bigg\{ \theta ^3 \left(37 \kappa ^4+7 \kappa ^2+37\right)+3 \theta ^2 \left[\left(10 \sqrt{3 \theta +1}+9\right) \kappa ^4+\left(4 \sqrt{3 \theta +1}+6\right) \kappa ^2+10 \sqrt{3 \theta +1}+9\right] +  \nonumber \\
&  +3 \theta  \left[\left(\sqrt{3 \theta +1}+3\right) \kappa ^4+\left(\sqrt{3 \theta +1}+2\right) \kappa ^2+\sqrt{3 \theta +1}+3\right]+\sqrt{3 \theta +1} \left(\kappa ^4+\kappa ^2+1\right) \bigg\}
\label{eq:IGAL skewness}
\end{align}
and
\begin{align}
\Kurt\left(X\right) = &
\frac{12}{(3 \theta +1)^4 \left(\kappa ^2+1\right)^2 \left[\frac{\left(\theta +\sqrt{3 \theta +1}\right)^2 \left(\kappa ^2-1\right)^2}{(3 \theta +1)^2}-\frac{2  \left\{\sqrt{3 \theta +1}+3 \theta  \left[\sqrt{3 \theta +1}+\theta  \left(\sqrt{3 \theta +1}+3\right)+1\right]  \right\} \left(\kappa ^4-\kappa ^2+1\right)}{(3 \theta +1)^{5/2}}\right]^2} \nonumber \\
& \times 
  \bigg[2 \left(105 \theta ^4+135 \theta ^3+105 \theta ^3 \sqrt{3 \theta +1}+54 \theta ^2+10 \theta  (3 \theta +1)^{3/2}+6 \theta +1\right) \left(\kappa ^2+1\right) \left(\kappa ^{10}+1\right)  + \nonumber \\
&   +\left[3 \theta  \left(\theta +\sqrt{3 \theta +1}+1\right)+1\right]^2 \left(\kappa ^6+1\right)^2+\nonumber \\
&  -2 \left(\theta +\sqrt{3 \theta +1}\right)  \left\{\sqrt{3 \theta +1}+3 \theta  \left[\sqrt{3 \theta +1}+\theta  \left(5 \theta +5 \sqrt{3 \theta +1}+6\right)+2\right] \right\} \left(\kappa ^4+1\right) \left(\kappa ^4-1\right)^2 \bigg]-3
\label{eq:IGAL kurtosis}
\end{align}
  
\figurename~\ref{fig:IGAL skew and kurt} shows examples of behaviors of $\Skew\left(X\right)$ and $\Kurt\left(X\right)$, as functions of $\kappa$, at various levels of $\theta$.
\begin{figure}[!ht]
\centering
\subfigure[Skewness \label{fig:IGALskew}]
{\resizebox{0.48\textwidth}{!}{\includegraphics{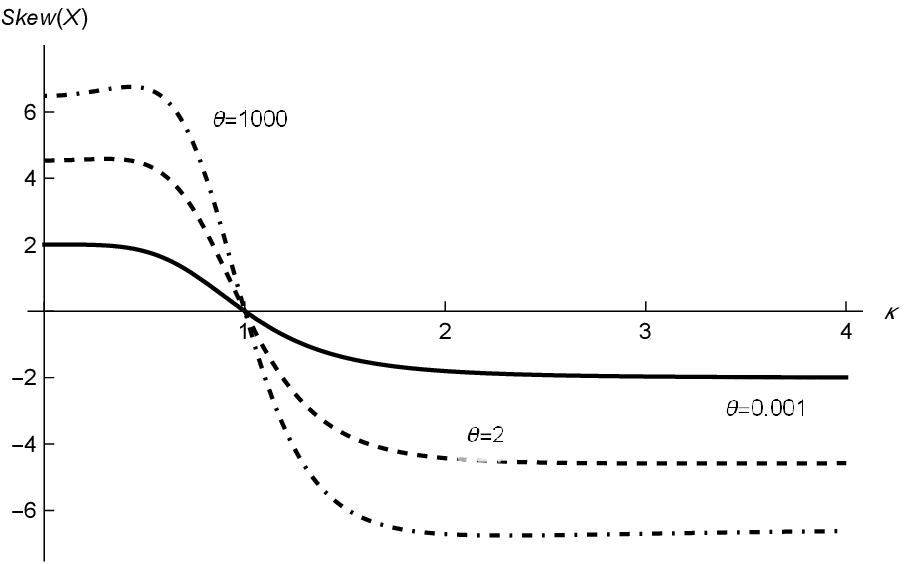}}}
\subfigure[Kurtosis\label{fig:IGALkurt}]
{\resizebox{0.48\textwidth}{!}{\includegraphics{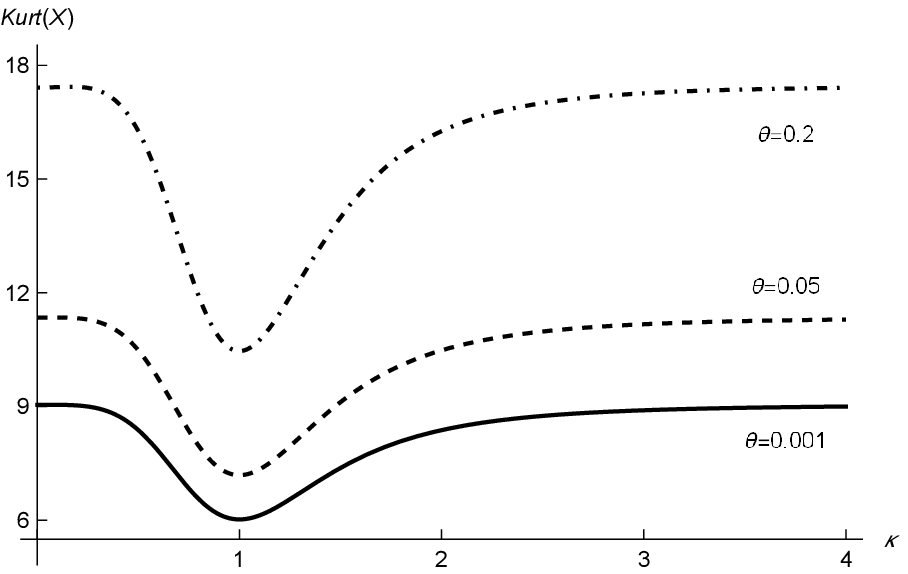}}}
\caption{
\footnotesize 
Examples of behavior of $\Skew\left(X\right)$ (on the left) and $\Kurt\left(X\right)$ (on the right), as functions of $\kappa$, at various levels of $\theta$ for the IG-AL distribution.
\label{fig:IGAL skew and kurt}
}
\end{figure}
From \figurename~\ref{fig:IGALskew} we realize that: 1) small values of $\theta$ tend to produce the plot of $\Skew\left(X\right)$ we obtained for $\mathcal{AL}\left(\mu,\beta,\kappa\right)$ (refer to \figurename~\ref{fig:AL skewness and kurtosis}); and 2) as $\theta$ grows, the range of possible values of $\Skew\left(X\right)$ increases.
Similarly, from \figurename~\ref{fig:IGALkurt} we realize that: 1) small values of $\theta$ tend to produce the plot of $\Kurt\left(X\right)$ we obtained for $\mathcal{AL}\left(\mu,\beta,\kappa\right)$ (refer to \figurename~\ref{fig:AL skewness and kurtosis}); 
and 2) $\kappa$ kept fixed, as $\theta$ grows, the kurtosis increases.
This means that the curve obtained for $\theta\rightarrow 0$ acts as a lower bound.


\subsection{Power-function asymmetric Laplace distribution}
\label{subsec:Power-function asymmetric Laplace distribution}

Let
\begin{equation}
h\left(w;\theta\right)=\theta w^{\theta-1},\quad 0 < w < 1,
\label{eq:Beta}
\end{equation}
with shape parameter $\theta>0$, be the pdf of the (standard) power-function distribution, in symbols $W\sim\mathcal{PF}\left(\theta\right)$, special case of the beta distribution $\mathcal{B}\left(\theta,1\right)$; see \citet[][Chapter~25]{John:Kotz:cont2:1970} and \citet{Ahsa:CJS:1974}. 
When the pdf in \eqref{eq:Beta} is considered as mixing density in model~\eqref{eq:asymmetric Laplace scale mixture}, the pdf of the ALSM simplifies as
\begin{equation}
f_{\text{ALSM}}\left(x;\mu,\beta,\kappa,\theta\right) = \frac{\theta}{\beta} \frac{\kappa}{1+\kappa^2} 
\begin{cases}
\left[\frac{\kappa}{\beta}\left(x-\mu\right)\right]^{-\left(\theta+1\right)}\gamma\left[\theta+1,\frac{\kappa}{\beta}\left(x-\mu\right)\right] 
& \text{for $x\geq \mu$}, \\[5mm]
\left[\frac{1}{\beta\kappa}\left(\mu-x\right)\right]^{-\left(\theta+1\right)}\gamma\left[\theta+1,\frac{1}{\beta\kappa}\left(\mu-x\right)\right]  & \text{for $x < \mu$},
  \end{cases}
\label{eq:PFAL distribution}
\end{equation}
where $\gamma\left(s,x\right)=\int_0^xt^{s-1}e^{-t}dt$ is the lower incomplete gamma function.
We refer to the resulting distribution as power-function asymmetric Laplace (PF-AL).
Model~\eqref{eq:PFAL distribution} reduces to $\mathcal{AL}\left(\mu,\beta,\kappa\right)$ when $\theta \rightarrow \infty$.
Moreover, if $\kappa=1$, then we obtain the PF-Laplace distribution, proposed by \citet{Punz:Bagn:PhyA:2021}, as a special case.

When $W\sim\mathcal{PF}\left(\theta\right)$, the $r$th raw moment of $1/W$ -- to be substituted in \eqref{eq:ALSM mean}--\eqref{eq:ALSM kurt} to obtain mean, variance, skewness and kurtosis of the PF-AL distribution -- is
\begin{equation}
\text{E}\left(\frac{1}{W^r}\right) = 
 \frac{\theta}{\theta - r}
, 
\label{eq:PF moments 1/W}
\end{equation}
which exists for $\theta>r$.
Thanks to \eqref{eq:PF moments 1/W}, we have
\begin{equation}
\text{E}\left(X\right) = \mu + \beta \left(\frac{1}{\kappa}-\kappa\right) \frac{\theta}{\theta - 1},
\label{eq:PFAL mean}
\end{equation}
which exists if $\theta>1$,
\begin{equation}
\var\left(X\right) = \frac{\beta^2 \theta  \left[(\theta -2) \theta  \left(\kappa^4+1\right)+2 \left(\kappa^4-\kappa^2+1\right)\right]}{(\theta -2) (\theta -1)^2 \kappa^2},
\label{eq:PFAL variance}
\end{equation}
which exists if $\theta>2$, 
\begin{equation}
\Skew\left(X\right) = 
-\frac{
2 \sqrt{\frac{\theta -2}{\theta }} \left(\kappa^2-1\right) \left\{\left[(\theta -2) \theta  ((\theta -3) \theta +6)+6\right] \kappa^4+(\theta -3) \theta  ((\theta -2) \theta +3) \kappa^2+(\theta -2) \theta  ((\theta -3) \theta +6)+6\right\}
}{
(\theta -3) \left[(\theta -2) \theta  \left(\kappa^4+1\right)+2 \left(\kappa^4-\kappa^2+1\right)\right]^{3/2}
}
\label{eq:PFAL skewness}
\end{equation}
which exists if $\theta>3$, and 
\begin{equation}
\Kurt\left(X\right) = 
\frac{12 (\theta -2)^2 (\theta -1)^4 \left[\frac{\theta  \left(\kappa ^6+1\right)^2}{(\theta -2)^2}-\frac{2 \theta  \left(\kappa ^4+1\right) \left(\kappa ^4-1\right)^2}{(\theta -3) (\theta -1)}+\frac{2 \left(\kappa ^2+1\right) \left(\kappa ^{10}+1\right)}{\theta -4}\right]}{\theta  \left(\kappa ^2+1\right)^2 \left[(\theta -2) \theta  \left(\kappa ^4+1\right)+2 \left(\kappa ^4-\kappa ^2+1\right)\right]^2}-3
,
\label{eq:PFAL kurtosis}
\end{equation}
which exists if $\theta>4$.
 
\figurename~\ref{fig:PFAL skew and kurt} shows examples of behaviors of $\Skew\left(X\right)$ and $\Kurt\left(X\right)$, as functions of $\kappa$, at various levels of $\theta$.
\begin{figure}[!ht]
\centering
\subfigure[Skewness \label{fig:PFALskew}]
{\resizebox{0.48\textwidth}{!}{\includegraphics{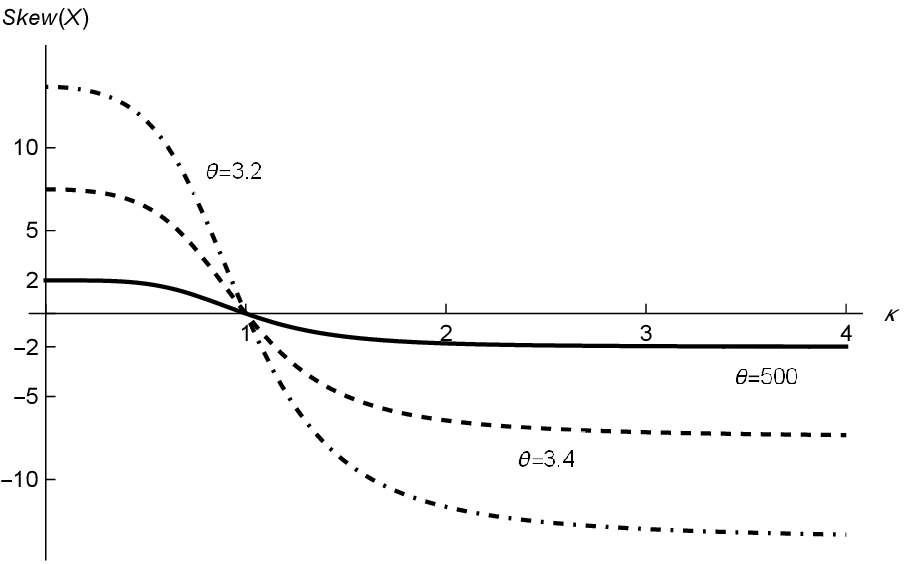}}}
\subfigure[Kurtosis\label{fig:PFALkurt}]
{\resizebox{0.48\textwidth}{!}{\includegraphics{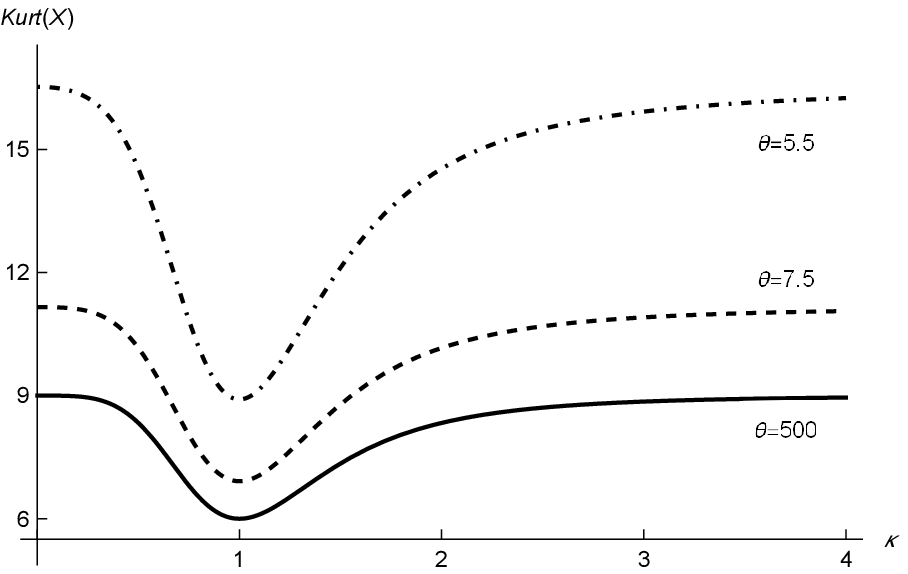}}}
\caption{
\footnotesize 
Examples of behavior of $\Skew\left(X\right)$ (on the left) and $\Kurt\left(X\right)$ (on the right), as functions of $\kappa$, at various levels of $\theta$ for the PF-AL distribution.
\label{fig:PFAL skew and kurt}
}
\end{figure}
From \figurename~\ref{fig:PFALskew} we realize that: 1) large values of $\theta$ tend to produce the plot of $\Skew\left(X\right)$ we obtained for $\mathcal{AL}\left(\mu,\beta,\kappa\right)$ (refer to \figurename~\ref{fig:AL skewness and kurtosis}); and 2) as $\theta$ decreases (to a minimum value of 3), the range of possible values of $\Skew\left(X\right)$ increases.
Moreover, to ``significantly'' modify the behavior of $\Skew\left(X\right)$ with respect to the $\mathcal{AL}\left(\mu,\beta,\kappa\right)$ case, we need low values of $\theta$ (in the examples, $\theta=3.4$ and $\theta=3.2$). 
Similarly, from \figurename~\ref{fig:PFALkurt} we realize that: 1) large values of $\theta$ tend to produce the plot of $\Kurt\left(X\right)$ we obtained for $\mathcal{AL}\left(\mu,\beta,\kappa\right)$ (refer to \figurename~\ref{fig:AL skewness and kurtosis}); 
and 2) $\kappa$ kept fixed, the lower the value $\theta$ (with a minimum of 4), the higher the kurtosis.
This means that the curve obtained for $\theta\rightarrow \infty$ acts as a lower bound.
Also in this case, to ``significantly'' modify the behavior of $\Kurt\left(X\right)$ with respect to the $\mathcal{AL}\left(\mu,\beta,\kappa\right)$ case, we need low values of $\theta$ (in the examples, $\theta=7.5$ and $\theta=5.5$). 

\subsection{Pareto asymmetric Laplace distribution}
\label{subsec:Pareto asymmetric Laplace distribution}

Let
\begin{equation}
h\left(w;\theta\right)=\frac{\theta}{x^{\theta+1}},\quad w > 1,
\label{eq:Pareto}
\end{equation}
with $\theta>0$, be the pdf of a Pareto distribution with unitary scale and shape $\theta$; in symbols, $W\sim\mathcal{P}\left(\theta\right)$. 
When the pdf in \eqref{eq:Pareto} is considered as mixing density in model~\eqref{eq:asymmetric Laplace scale mixture}, the pdf of the ALSM simplifies as
\begin{equation}
f_{\text{ALSM}}\left(x;\mu,\beta,\kappa,\theta\right) = \frac{\theta}{\beta} \frac{\kappa}{1+\kappa^2} 
\begin{cases}
 E_{\theta}\left[\frac{\kappa}{\beta}\left(x-\mu\right)\right]
& \text{for $x\geq \mu$,} \\[5mm]
E_{\theta}\left[\frac{1}{\beta\kappa}\left(\mu-x\right)\right]  & \text{for $x < \mu$.}
  \end{cases}
\label{eq:PAL distribution}
\end{equation}
We refer to the resulting distribution as Pareto asymmetric Laplace (P-AL).
Model~\eqref{eq:PAL distribution} reduces to $\mathcal{AL}\left(\mu,\beta,\kappa\right)$ when $\theta \rightarrow \infty$.

When $W\sim\mathcal{P}\left(\theta\right)$, the $r$th raw moment of $1/W$ is
\begin{equation}
\text{E}\left(\frac{1}{W^r}\right) = 
 \frac{\theta}{\theta+r}. 
\label{eq:PAL moments 1/W}
\end{equation}
By substituting in \eqref{eq:ALSM mean}--\eqref{eq:ALSM kurt} the first four raw moments in \eqref{eq:PAL moments 1/W}, mean, variance, skewness and kurtosis of the P-AL distribution become
\begin{equation}
\text{E}\left(X\right) = \mu + \beta \left(\frac{1}{\kappa}-\kappa\right) \frac{\theta}{1+\theta},
\label{eq:PAL mean}
\end{equation}
\begin{equation}
\var\left(X\right) = \frac{\beta^2 \theta  \left[\left(\theta + 2\right) \left(1+\kappa^4\right)  -2 \kappa^2  \right]}{(\theta+1)^2 \kappa^2},
\label{eq:PAL variance}
\end{equation}
\begin{align}
\Skew\left(X\right)  =& 
\frac{1}{(\theta +3) \left[\theta ^2 \left(\kappa ^4+1\right)+2 \theta  \left(\kappa ^4+1\right)+2 \left(\kappa ^4-\kappa ^2+1\right)\right]^{3/2}}
 \Bigg\{ 2 \sqrt{\frac{\theta +2}{\theta }} \left(1-\kappa^2\right) \bigg[\theta ^4 \left(\kappa ^4+\kappa ^2+1\right) +   \nonumber \\
&   +5 \theta ^3 \left(\kappa ^4+\kappa ^2+1\right) +3 \theta ^2 \left(4 \kappa ^4+3 \kappa ^2+4\right)+3 \theta  \left(4 \kappa ^4+3 \kappa ^2+4\right)+6 \left(\kappa ^4+1\right) \bigg] \Bigg\}
\label{eq:PAL skewness}
\end{align}
and
\begin{equation}
\Kurt\left(X\right) = 
\frac{12 (\theta +1)^4 (\theta +2)^2 \left[\frac{\theta  \left(\kappa ^6+1\right)^2}{(\theta +2)^2}-\frac{2 \theta  \left(\kappa ^4+1\right) \left(\kappa ^4-1\right)^2}{(\theta +1) (\theta +3)}+\frac{2 \left(\kappa ^2+1\right) \left(\kappa ^{10}+1\right)}{\theta +4}\right]}{\theta  \left[\theta ^2 \left(\kappa ^2+1\right) \left(\kappa ^4+1\right)+2 \theta  \left(\kappa ^2+1\right) \left(\kappa ^4+1\right)+2 \kappa ^6+2\right]^2}-3
\label{eq:PAL kurtosis}
\end{equation}

\figurename~\ref{fig:PAL skew and kurt} shows examples of behaviors of $\Skew\left(X\right)$ and $\Kurt\left(X\right)$, as functions of $\kappa$, at various levels of $\theta$.
\begin{figure}[!ht]
\centering
\subfigure[Skewness \label{fig:PALskew}]
{\resizebox{0.48\textwidth}{!}{\includegraphics{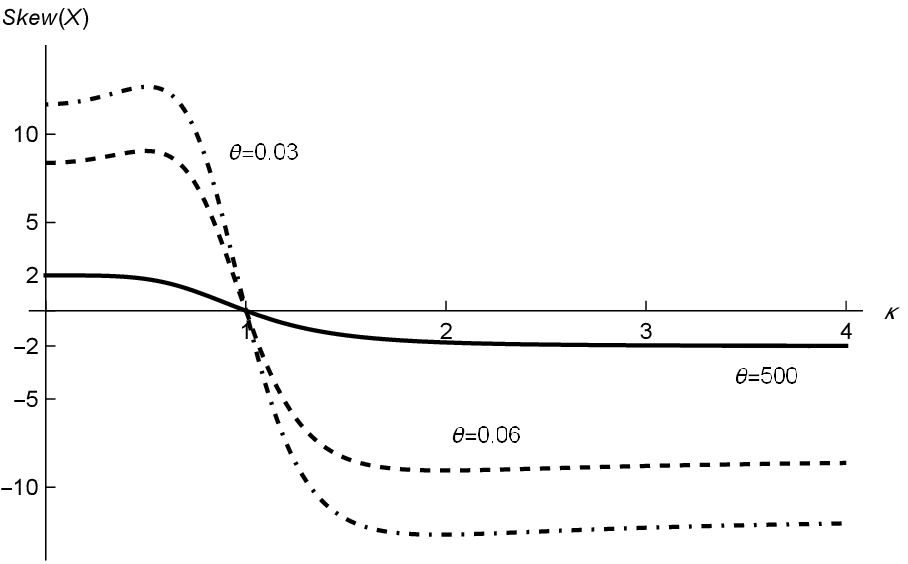}}}
\subfigure[Kurtosis\label{fig:PALkurt}]
{\resizebox{0.48\textwidth}{!}{\includegraphics{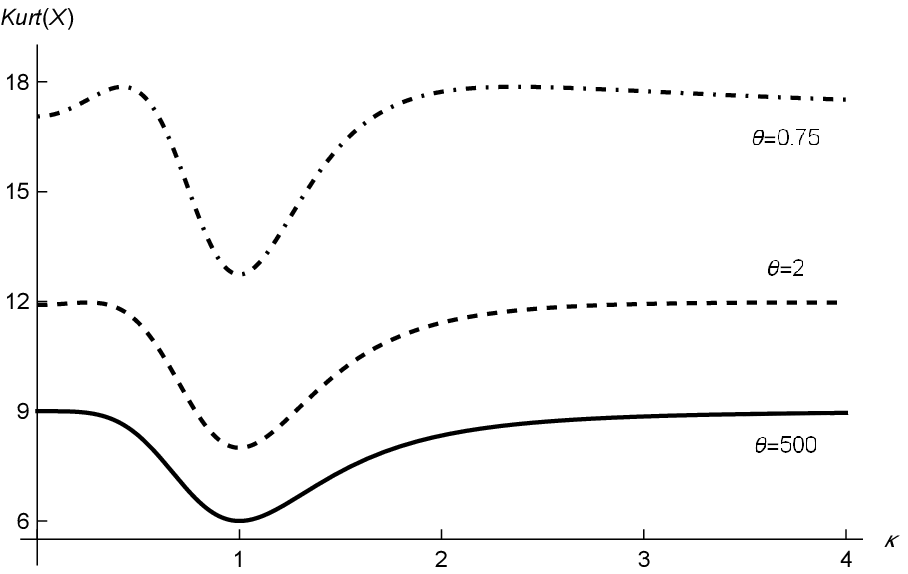}}}
\caption{
\footnotesize 
Examples of behavior of $\Skew\left(X\right)$ (on the left) and $\Kurt\left(X\right)$ (on the right), as functions of $\kappa$, at various levels of $\theta$ for the P-AL distribution.
\label{fig:PAL skew and kurt}
}
\end{figure}
From \figurename~\ref{fig:PALskew} we realize that: 1) a large value of $\theta$ produces, as a limiting case, the plot of $\Skew\left(X\right)$ we obtained for $\mathcal{AL}\left(\mu,\beta,\kappa\right)$ (refer to \figurename~\ref{fig:AL skewness and kurtosis}); and 2) as $\theta$ decreases, the range of possible values of $\Skew\left(X\right)$ increases.
Moreover, to ``significantly'' modify the behavior of $\Skew\left(X\right)$ with respect to the $\mathcal{AL}\left(\mu,\beta,\kappa\right)$ case, we need low values of $\theta$ (in the examples, $\theta=0.06$ and $\theta=0.03$). 
Similarly, from \figurename~\ref{fig:PALkurt} we realize that: 1) a large value of $\theta$ produces, as a limiting case, the plot of $\Kurt\left(X\right)$ we obtained for $\mathcal{AL}\left(\mu,\beta,\kappa\right)$ (refer to \figurename~\ref{fig:AL skewness and kurtosis}); 
and 2) $\kappa$ kept fixed, the lower the value $\theta$, the higher the kurtosis.
This means that the curve obtained for $\theta\rightarrow \infty$ acts as a lower bound.
Also in this case, to ``significantly'' modify the behavior of $\Kurt\left(X\right)$ with respect to the $\mathcal{AL}\left(\mu,\beta,\kappa\right)$ case, we need low values of $\theta$ (in the examples, $\theta=2$ and $\theta=0.75$).

\subsection{Uniform asymmetric Laplace distribution}
\label{subsec:Uniform asymmetric Laplace distribution}

When a uniform distribution on $\left(1-\theta,1\right)$, $0<\theta<1$, is chosen as mixing distribution, in symbols $W\sim \mathcal{U}_{\left(1-\theta,1\right)}\left(\theta\right)$, the pdf of the ALSM becomes
\begin{equation}
f_{\text{ALSM}}\left(x;\mu,\beta,\kappa,\theta\right) = 
\frac{\kappa}{\theta\left(1+\kappa^2\right)\left(x-\mu\right)^2}
\begin{cases}
 \frac{1}{\kappa} e^{\frac{\kappa (\theta -1) ( x - \mu)}{\beta  }}  \left\{\left(\frac{\beta}{\kappa}+x-\mu\right)\left[1- e^{-\frac{\kappa \theta (x-\mu )}{\beta  }}  \right] -\theta\left(x-\mu\right)\right\} & \text{for $x\geq \mu$,}
 \\[5mm]
  \kappa   e^{\frac{ (\theta -1) (\mu -x)}{\beta \kappa }}  \left\{\left( \beta \kappa +\mu -x\right)\left[1- e^{-\frac{ \theta (\mu - x )}{\beta \kappa }}  \right] -\theta\left(\mu - x \right)\right\} 
  & \text{for $x < \mu$.}
  \end{cases}
\label{eq:U-AL distribution}
\end{equation}
We will refer to the resulting distribution as uniform asymmetric Laplace (U-AL) or, in analogy with \citet{Punz:Bagn:JSCS:2021}, as tail-inflated asymmetric Laplace distribution.
Model~\eqref{eq:U-AL distribution} reduces to $\mathcal{AL}\left(\mu,\beta,\kappa\right)$ when $\theta \rightarrow 0$.

When $W\sim \mathcal{U}_{\left(1-\theta,1\right)}\left(\theta\right)$, the $r$th raw moment of $1/W$ is
\begin{equation}
\text{E}\left(\frac{1}{W^r}\right) =
\begin{cases}
 -\displaystyle\frac{\log\left(1-\theta\right)}{\theta}
	& \text{for $r = 1$,} \\[2mm]
	\displaystyle\frac{\left(1-\theta\right)^{1-r}-1}{\theta\left(r-1\right)}
  & \text{for $r \in \{2,3,\ldots\}$.}
  \end{cases}
\label{eq:U moments 1/W}
\end{equation}
By substituting in \eqref{eq:ALSM mean}--\eqref{eq:ALSM kurt} the first four raw moments in \eqref{eq:U moments 1/W}, mean, variance, skewness and kurtosis of the U-AL distribution become
\begin{equation}
\text{E}\left(X\right) = \mu - \beta\left(\frac{1}{\kappa}-\kappa\right)\frac{\log(1-\theta )}{\theta},
\label{eq:UAL mean}
\end{equation}
\begin{equation}
\var\left(X\right) = -\frac{\beta^2}{\theta^2 \kappa^2} \log (1-\theta ) \left[\left(\kappa^2-1\right)^2 \log (1-\theta )+2 \theta  \left(\kappa^4-\kappa^2+1\right)\right] ,
\label{eq:UAL variance}
\end{equation}
\begin{equation}
\Skew\left(X\right) = 
\frac{3 \left[\frac{1}{(\theta -1)^2}-1\right] \theta ^2 \left(\kappa ^2+1\right)^2 \left(1-\kappa ^8\right)+\frac{6 \theta ^2 \left(\kappa ^2+1\right) \left(1-\kappa ^4\right) \left(\kappa ^6+1\right) \log (1-\theta )}{1-\theta }+2 \left(\kappa ^4-1\right)^3 \log ^3(1-\theta )}{\theta ^3 \left(\kappa ^2+1\right)^3 \left[-\frac{\left(\kappa ^2-1\right)^2 \log ^2(1-\theta )}{\theta ^2}-\frac{2 \left(\kappa ^4-\kappa ^2+1\right)}{\theta -1}\right]^{3/2}}
\label{eq:UAL skewness}
\end{equation}
and
\begin{equation}
\Kurt\left(X\right) = 
\frac{4  \left\{\frac{3 \theta  \left(\kappa ^6+1\right)^2}{(\theta -1)^2}-\frac{3 (\theta -2) \left(\kappa ^4+1\right) \left(\kappa ^4-1\right)^2 \log (1-\theta )}{(\theta -1)^2}+2 \left[\frac{1}{(1-\theta )^3}-1\right] \left(\kappa ^2+1\right) \left(\kappa ^{10}+1\right) \right\}}{\theta  \left(\kappa ^2+1\right)^2 \left[\frac{\left(\kappa ^2-1\right)^2 \log ^2(1-\theta )}{\theta ^2}+\frac{2 \left(\kappa ^4-\kappa ^2+1\right)}{\theta -1}\right]^2}-3
\label{eq:UAL kurtosis}
\end{equation}
 
\figurename~\ref{fig:UAL skew and kurt} shows examples of behaviors of $\Skew\left(X\right)$ and $\Kurt\left(X\right)$, as functions of $\kappa$, at various levels of $\theta$.
\begin{figure}[!ht]
\centering
\subfigure[Skewness \label{fig:UALskew}]
{\resizebox{0.48\textwidth}{!}{\includegraphics{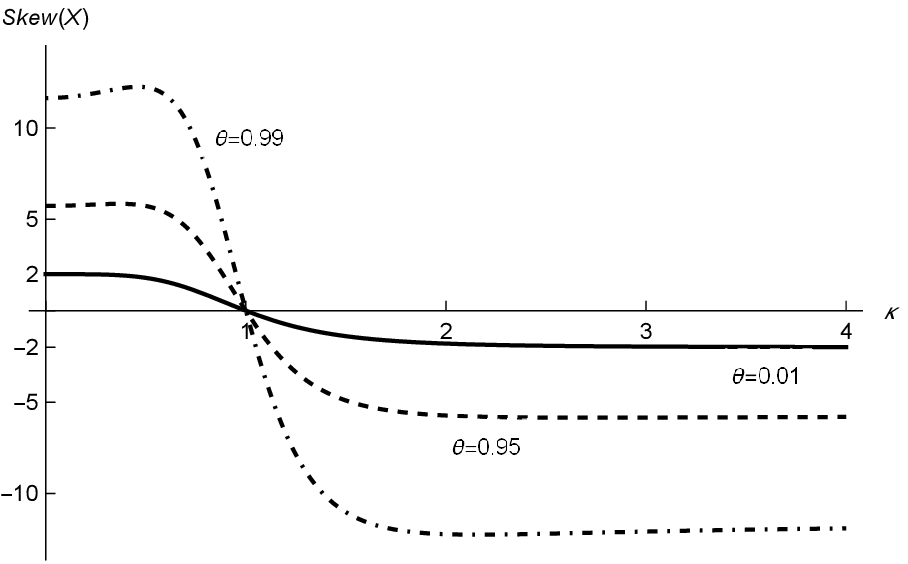}}}
\subfigure[Kurtosis\label{fig:UALkurt}]
{\resizebox{0.48\textwidth}{!}{\includegraphics{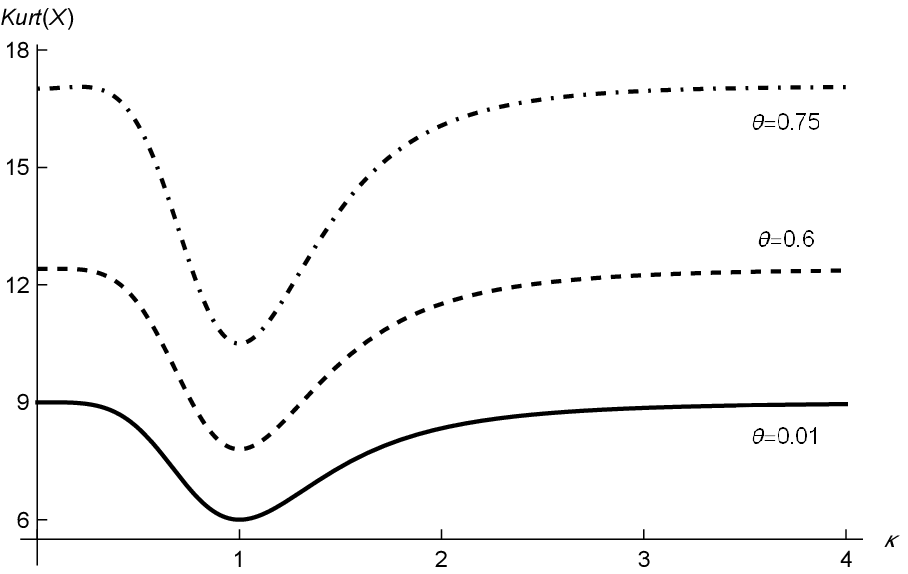}}}
\caption{
\footnotesize 
Examples of behavior of $\Skew\left(X\right)$ (on the left) and $\Kurt\left(X\right)$ (on the right), as functions of $\kappa$, at various levels of $\theta$ for the U-AL distribution.
\label{fig:UAL skew and kurt}
}
\end{figure}
From \figurename~\ref{fig:UALskew} we realize that: 1) small values of $\theta$ tend to produce the plot of $\Skew\left(X\right)$ we obtained for $\mathcal{AL}\left(\mu,\beta,\kappa\right)$ (refer to \figurename~\ref{fig:AL skewness and kurtosis}); and 2) as $\theta$ grows, the range of possible values of $\Skew\left(X\right)$ increases.
Moreover, to ``significantly'' modify the behavior of $\Skew\left(X\right)$ with respect to the $\mathcal{AL}\left(\mu,\beta,\kappa\right)$ case, we need values of $\theta$ close to 1 (in the examples, $\theta=0.95$ and $\theta=0.99$). 
Similarly, from \figurename~\ref{fig:UALkurt} we realize that: 1) small values of $\theta$ tend to produce the plot of $\Kurt\left(X\right)$ we obtained for $\mathcal{AL}\left(\mu,\beta,\kappa\right)$ (refer to \figurename~\ref{fig:AL skewness and kurtosis}); 
and 2) $\kappa$ kept fixed, as $\theta$ grows, the kurtosis increases.
This means that the curve obtained for $\theta\rightarrow 0$ acts as a lower bound.

\section{Maximum likelihood estimation}
\label{sec:Maximum likelihood estimation}

Several estimators of the parameters of the ALSMs may be considered.
Among them, maximum likelihood (ML) estimators are most attractive because of their large sample properties.

Given a sample $x_1,\ldots,x_n$ from the pdf in \eqref{eq:asymmetric Laplace scale mixture}, the log-likelihood function 
of the ALSM is
\begin{equation}
l\left(\mu,\beta,\kappa,\btheta\right) = \sum_{i=1}^n \ln\left[f_{\text{ALSM}}\left(x_i;\mu,\beta,\kappa,\btheta\right)\right].
\label{eq:compound loglik}
\end{equation}
The first order partial derivatives of \eqref{eq:compound loglik}, with respect to $\left(\mu,\beta,\kappa,\btheta'\right)'$, are
\begin{equation*}
l'\left(\mu,\beta,\kappa,\btheta\right)=\sum_{i=1}^n\frac{\partial}{\partial \left(\mu,\beta,\kappa,\btheta'\right)'} \ln\left[f_{\text{ALSM}}\left(x_i;\mu,\beta,\kappa,\btheta\right)\right]. 
\label{eq:loglik first derivatives}
\end{equation*}
The values of $\mu$, $\beta$, $\kappa$, and $\btheta$ that maximize $l\left(\mu,\beta,\kappa,\btheta\right)$ are the ML estimates $\widehat{\mu}$, $\widehat{\beta}$, $\widehat{\kappa}$, and $\widehat{\btheta}$ and satisfy the condition $l'\left(\mu,\beta,\kappa,\btheta\right) = \bzero$.
 
The ML fitting of most of the models in our family is simplified considerably by the application of the expectation-maximization (EM) algorithm \citep{Demp:Lair:Rubi:Maxi:1977} which is, indeed, the classical approach to find ML estimates for parameters of distributions which are defined as a mixture.
The TP-AL and the U-AL distributions are the exceptions; indeed, they require convenient variants of the EM algorithm.
For the TP-AL distribution, in analogy with other contaminated models available in the literature (\citealp{Punz:McNi:Robu:2016}, \citealp{Punz:Anew:2019}, \citealp{Morr:Punz:McNi:Brow:Asym:2019}, \citealp{Toma:Punz:Heav:2020} and \citealp{Punz:Tort:SM:2021}), we illustrate in \ref{app:ECM algorithm for the two-point asymmetric Laplace distribution} the application of the expectation-conditional maximization (ECM) algorithm \citep{Meng:Rubin:Maxi:1993}, a well-known extension of the EM algorithm \cite[see][Chapter~5, for details]{McLa:Kris:TheE:2007}.
The ECM algorithm replaces the M-step of the EM algorithm by a number of computationally simpler conditional maximization (CM) steps.
The expectation-conditional maximization either (ECME) algorithm \citep{Liu:Rubi:TheE:1994} generalizes the ECM algorithm by conditionally maximizing on some or all of the CM-steps the incomplete-data log-likelihood.
We illustrate this algorithm, to fit the U-AL distribution, in \ref{app:ECME algorithm for the uniform asymmetric Laplace distribution}.

For the application of the EM algorithm and, more in general, of any of its variant, it is convenient to view the observed data as incomplete.
The complete data are $\left(x_1,w_1\right),\ldots,\left(x_n,w_n\right)$, where the missing variables $w_1,\ldots,w_n$ are defined so that 
\begin{equation}
X_i | W_i=w_i \sim \mathcal{AL}\left(\mu, \beta/w_i, \kappa\right),
\nonumber \label{eq:complete data X}
\end{equation}
independently for $i = 1, \ldots , n$, and
\begin{equation}
W_i \stackrel{\text{i.i.d.}}{\sim} \mathcal{H}\left(\btheta\right).
\nonumber \label{eq:complete data W}
\end{equation} 
Because of this conditional structure, the complete-data likelihood function $L_c\left(\mu,\beta,\kappa,\btheta\right)$ can be factored into the product of the conditional densities of $X_i$ given the $w_i$ and the joint marginal densities of $W_i$, i.e.
\begin{equation}
L_c\left(\mu,\beta,\kappa,\btheta\right) = \prod_{i=1}^n f_{\text{AL}}\left(x_i;\mu, \beta/w_i, \kappa \right) h\left(w_i;\btheta\right). 
\nonumber \label{eq:ALSM complete-data likelihood}
\end{equation}
Accordingly, the complete-data log-likelihood function can be written as
\begin{equation}
l_c\left(\mu,\beta,\kappa,\btheta\right) = 
l_{1c}\left(\mu,\beta,\kappa\right) + l_{2c}\left(\btheta\right) ,
\label{eq:ALSM complete-data log-likelihood}
\end{equation}
where 
\begin{equation}
l_{1c}\left(\mu,\beta,\kappa\right) = 
\sum_{i=1}^n \log\left(w_i\right)
-n \log\left(\beta\right)
+ n \log\left(\frac{\kappa}{1+\kappa^2}\right)
-\frac{1}{\beta}\left[\kappa \sum_{i=1}^n \left(x_i-\mu\right)^+w_i+\frac{1}{\kappa}\sum_{i=1}^n \left(x_i-\mu\right)^-w_i\right]
\label{eq:ALSM complete-data log-likelihood mu, Gamma and Sigma}
\end{equation}
and
\begin{equation} 
l_{2c}\left(\btheta\right) =   \sum_{i=1}^n \log\left[h\left(w_{i};\btheta\right)\right], 
\label{eq:ALSM complete-data log-likelihood vartheta}
\end{equation}
with
\begin{equation*}
  \left(x_i-\mu\right)^+ =  
	\begin{cases}
      x_i-\mu & x_i \geq \mu\\
      0 & x_i < \mu
    \end{cases} \, \,
\end{equation*}
and
\begin{equation*}
  \left(x_i-\mu\right)^- =  
	\begin{cases}
      \mu - x_i & x_i \leq \mu\\
      0 & x_i > \mu
    \end{cases} \, .
\end{equation*}
So, while $l_{1c}\left(\mu,\beta,\kappa\right)$ is shared by all the ALSMs, $l_{2c}\left(\btheta\right)$ is distribution-dependent and -- by only focusing on the distributions for which the EM algorithm is considered -- is given by  
\begin{equation} 
	l_{2c}\left(\theta\right) = n \log(\theta) - \theta \sum_{i=1}^n (w_i-1)
	\label{eq:complete-data log-likelihood shifted}
\end{equation}
for the SE-AL distribution,
\begin{equation} 
l_{2c}\left(\theta\right) = \frac{1}{\theta} \sum_{i=1}^n \log w_i - \frac{1}{\theta}  \sum_{i=1}^n w_i 
-n\left(\frac{1}{\theta}+1\right)\log \theta - n\log \left[\Gamma\left(\frac{1}{\theta}+1\right)\right]
	\label{eq:complete-data log-likelihood unimodal gamma}
\end{equation}
for the UG-AL distribution,
\begin{equation} 
	l_{2c}\left(\theta\right) = -\frac{n}{2} \log\left(\frac{3\theta + 1}{2 \pi \theta}\right) - 
	\frac{3}{2} \sum_{i=1}^n \log w_i 
	- \frac{1}{2\theta} \left[ \sum_{i=1}^n w_i
	- 2n\sqrt{3\theta+1}
	+ (3\theta+1) \sum_{i=1}^n \frac{1}{w_i} \right]
	\label{eq:complete-data log-likelihood inverse Gaussian}
\end{equation}
for the IG-AL distribution, 
\begin{equation} 
l_{2c}\left(\theta\right) = n\log \theta + (\theta-1)  \sum_{i=1}^n \log w_i
	\label{eq:complete-data log-likelihood beta}
\end{equation}
for the PF-AL distribution, and 
\begin{equation} 
l_{2c}\left(\theta\right) = n\log \theta 
-  \left(\theta + 1 \right)\sum_{i=1}^n \log w_i
	\label{eq:complete-data log-likelihood pareto}
\end{equation}
for the PAR-AL distribution.

Working on $l_c\left(\mu,\beta,\kappa,\btheta\right)$, the EM algorithm iterates between two steps, one E-step and one M-step, until convergence.
We detail these steps below for a generic iteration of the algorithm, as well as for all the considered ALSMs apart from the TP-AL and U-AL distributions (see \ref{app:ECM algorithm for the two-point asymmetric Laplace distribution} and \ref{app:ECME algorithm for the uniform asymmetric Laplace distribution}).
As in \citet{melnykov2018model,melnykov2019studying}, quantities/parameters marked with one dot will correspond to the previous iteration and those marked with
two dots will represent the estimates at the current iteration.


\subsection{E-step}
\label{subsec:E-step}

The E-step requires the calculation of 
\begin{equation}
Q\left(\mu,\beta,\kappa,\btheta|\dot{\mu},\dot{\beta},\dot{\kappa},\dot{\btheta}\right) = Q_1\left(\mu,\beta,\kappa|\dot{\mu},\dot{\beta},\dot{\kappa},\dot{\btheta}\right) + Q_2\left(\btheta|\dot{\mu},\dot{\beta},\dot{\kappa},\dot{\btheta}\right),
\label{eq:Q}
\end{equation}
the conditional expectation of $l_c\left(\mu,\beta,\kappa,\btheta\right)$ given the observed data $x_1,\ldots,x_n$, using $\left\{\dot{\mu},\dot{\beta},\dot{\kappa},\dot{\btheta}\right\}$ for $\left\{\mu,\beta,\kappa,\btheta\right\}$.
In \eqref{eq:Q} the two terms on the right-hand side are ordered as the two terms on the right-hand side of \eqref{eq:ALSM complete-data log-likelihood}.
As well-explained in \citet[][p.~82]{McNe:Frey:Embr:Quan:2005}, in order to compute $Q\left(\mu,\beta,\kappa,\btheta|\dot{\mu},\dot{\beta},\dot{\kappa},\dot{\btheta}\right)$ we need to replace any function $g\left(W_{i}\right)$ of the latent mixing variables which arise in \eqref{eq:ALSM complete-data log-likelihood mu, Gamma and Sigma} and \eqref{eq:ALSM complete-data log-likelihood vartheta} by the quantities $\dot{\text{E}}
\left[g\left(W_{i}\right)|X_i=x_i\right]$, where the expectation, as highlighted by the superscript, is taken using $\left\{\dot{\mu},\dot{\beta},\dot{\kappa},\dot{\btheta}\right\}$ for $\left\{\mu,\beta,\kappa,\btheta\right\}$, $i=1,\ldots,n$.
To calculate these expectations we can observe that the conditional pdf of $W_{i}|X_i=x_i$ satisfies $f\left(w_{i}|x_i;\mu,\beta,\kappa,\btheta\right) \propto f\left(w_{i},x_i;\mu,\beta,\kappa,\btheta\right)$, up to some constant of proportionality.

Below we detail the E-step for the considered ALSMs.

\subsubsection{Shifted exponential asymmetric Laplace distribution}
\label{subsub:shifted}

For the SE-AL distribution we have 
\begin{align}
f\left(w_{i}|x_i;\mu,\beta,\kappa,\theta\right) 
& \propto  f\left(w_{i},x_i;\mu,\beta,\kappa,\theta\right) \nonumber \\
& \propto \frac{1}{\Gamma\left(2,\delta\left(x_i;\mu,\beta,\kappa\right)+\theta \right)} f_{\text{G}}\left(w_i;2,\delta\left(x_i;\mu,\beta,\kappa\right)+\theta\right) ,
\label{eq:posterior_shifted}
\end{align} 
where  
\begin{equation*}
  \delta\left(x_i;\mu,\beta,\kappa\right) =  
	\begin{cases}
      \frac{\kappa}{\beta} \left(x_i - \mu\right)  & x_i \geq \mu\\
      \frac{1}{\kappa\beta} \left(\mu -x_i  \right) & x_i < \mu
    \end{cases} \, \,
		\label{eq:delta}
\end{equation*}
and
$$
f_{\text{G}}\left(w;\alpha,\beta\right)=\frac{\beta^{\alpha}}{\Gamma\left(\alpha\right)}w^{\alpha-1}\exp\left(-\beta w\right)
$$ 
denotes the pdf of a gamma distribution with shape $\alpha>0$ and rate $\beta>0$.
This means that $W_{i}|X_i=x_i$ has a left-truncated gamma distribution (see, e.g., \citealp{coffey2000properties}), on the interval $\left(1,\infty\right)$, with shape $2$ and rate $\delta\left(x_i;\mu,\beta,\kappa\right)+\theta$, whose pdf is given in \eqref{eq:posterior_shifted}; in symbols   
\begin{equation}
W_{i}|X_i=x_i \sim \mathcal{LTG}_{\left(1,\infty\right)}\left(2,\delta\left(x_i;\mu,\beta,\kappa\right)+\theta \right).
\label{eq:Truncated Gamma}
\end{equation}
The function $g(W_i)$ arising from \eqref{eq:ALSM complete-data log-likelihood mu, Gamma and Sigma} and \eqref{eq:complete-data log-likelihood shifted} is $g_1(w)=w$.
Thanks to \eqref{eq:Truncated Gamma} we obtain
\begin{align}
\dot{\mbox{E}} \left(W_{i}|X_i=x_i\right) & =   
\frac{\varphi_{2}\left[\delta\left(x_i;\dot{\mu},\dot{\beta},\dot{\kappa}\right)+\dot{\theta}\right]}
{\varphi_{1}\left[\delta\left(x_i;\dot{\mu},\dot{\beta},\dot{\kappa}\right)+\dot{\theta}\right]}, \nonumber \\
&\eqqcolon \dot{w}_{i}, \nonumber
\end{align}
where
\begin{eqnarray}
\varphi_{m}(z) &=& E_{-m}\left(z\right) \nonumber \\
 &=& \int_{1}^{\infty} t^{m} e^{-z t} dt  \nonumber \\
&=& z^{-(m+1)}\Gamma\left(m+1,z \right) \nonumber
\label{eq:expint}
\end{eqnarray}
is the Misra function \citep{Misr:1940}, further generalization of the generalized exponential integral function $E_n\left(z\right) = \int_{1}^\infty t^{-n} e^{-zt} dt$ 
\citep{Abra:Steg:Hand:1965}.

\subsubsection{Unimodal gamma asymmetric Laplace distribution}
\label{subsub:gamma}

For the UG-AL distribution we have 
\begin{align}
f\left(w_{i}|x_i;\mu,\beta,\kappa,\theta\right) 
& \propto  f\left(w_{i},x_i;\mu,\beta,\kappa,\theta\right) \nonumber \\
& \propto  f_{\text{G}}\left(w_i;\frac{1}{\theta}+2,\delta\left(x_i;\mu,\beta,\kappa\right)+\frac{1}{\theta}\right) 
\label{eq:posterior_gamma}.
\end{align} 
This means that $W_{i}|X_i=x_i$ has a gamma distribution with shape $1/\theta+2$ and rate $\delta\left(x_i;\mu,\beta,\kappa\right)+1/\theta$, whose pdf is given in \eqref{eq:posterior_gamma}; in symbols   
\begin{equation}
W_{i}|X_i=x_i \sim \mathcal{G}\left(\frac{1}{\theta}+2,\delta\left(x_i;\mu,\beta,\kappa\right)+\frac{1}{\theta} \right).
\label{eq:W_Gamma}
\end{equation}
The functions $g(W_i)$ arising from \eqref{eq:ALSM complete-data log-likelihood mu, Gamma and Sigma} and \eqref{eq:complete-data log-likelihood unimodal gamma} are $g_1\left(w\right)=w$ and $g_2\left(w\right)=\log\left(w\right)$. 
Thanks to \eqref{eq:W_Gamma} we obtain
\begin{align}
\dot{\mbox{E}} \left(W_{i}|X_i=x_i\right) & =   
\frac{\frac{1}{\dot{\theta}}+2}{\delta\left(x_i;\dot{\mu},\dot{\beta},\dot{\kappa}\right) +\frac{1}{\dot{\theta}}} \nonumber \\
&\eqqcolon \dot{w}_{i}, \nonumber 
\end{align}
and 
\begin{align}
\dot{\mbox{E}}\left(\log W_{i}|X_i=x_i\right) & =\psi\left(\frac{1}{\dot{\theta}}+2\right) - \log\left[   \delta\left(x_i;\dot{\mu},\dot{\beta},\dot{\kappa}\right) +\frac{1}{\dot{\theta}} \right] 
 \nonumber \\
&\eqqcolon \dot{\log}\:w_{i} \nonumber
\end{align}
where $\psi(\cdot)$ denotes the digamma function.

In \ref{app:E-step G-AL} we report the analogous material for the reparameterized UG-AL distribution obtained considering the classical gamma as mixing distribution.

\subsubsection{Inverse Gaussian asymmetric Laplace distribution}
\label{subsub:inverse}

For the IG-AL distribution we have 
\begin{align}
f\left(w_{i}|x_i;\mu,\beta,\kappa,\theta\right) 
& \propto  f\left(w_{i},x_i;\mu,\beta,\kappa,\theta\right) \nonumber \\
& \propto \frac{1}{\eta \left(x_i,\mu,\beta,\kappa,\theta\right)}   w_i^{-\frac{1}{2}} e^{-w_i \left[\delta\left(x_i;\mu,\beta,\kappa\right) + \frac{\left(w_i-\sqrt{3\theta+1}\right)^2}{2\theta w_i^2}\right]},
\label{eq:posterior_inverse} 
\end{align} 
where 
$$
\eta \left(x_i;\mu,\beta,\kappa,\theta\right)= \sqrt{\frac{2\pi\theta}{1+2\theta \delta\left(x_i;\mu,\beta,\kappa\right) }} e^{\frac{\sqrt{1+3\theta} \left(1- \sqrt{ 1+2\theta \delta\left(x_i;\mu,\beta,\kappa\right)}\right)}{\theta}}
$$
allows \eqref{eq:posterior_inverse} to be a well-defined pdf.
The functions $g(W_i)$ arising from \eqref{eq:ALSM complete-data log-likelihood mu, Gamma and Sigma} and \eqref{eq:complete-data log-likelihood inverse Gaussian} are $g_1\left(w\right)=w$ and $g_2\left(w\right)=1/w$. 
Thanks to \eqref{eq:posterior_inverse} we obtain
\begin{align}
\dot{\mbox{E}}\left(W_{i}|X_i=x_i\right) & =   
\frac{\sqrt{\dot{\theta}\left(1+3\dot{\theta}\right)} + \sqrt{\frac{\dot{\theta}^3}{1+2\dot{\theta} \delta\left(x_i;\dot{\mu},\dot{\beta},\dot{\kappa}\right)}}}
{\sqrt{ \dot{\theta} \left[1+2\dot{\theta} \delta\left(x_i;\dot{\mu},\dot{\beta},\dot{\kappa}\right)\right]}} \nonumber \\
&\eqqcolon \dot{w}_{i}, \nonumber
\end{align}
and 
\begin{align}
\dot{\mbox{E}}\left(\displaystyle \left.\frac{1}{W_{i}} \right|X_i=x_i\right) & =  
\sqrt{\frac{1+2\dot{\theta} \delta\left(x_i;\dot{\mu},\dot{\beta},\dot{\kappa}\right)}{1+3\dot{\theta}}}
 \nonumber \\
&\eqqcolon \dot{\bar{w}}_{i}. \nonumber
\end{align}

\subsubsection{Power-function asymmetric Laplace distribution}
\label{subsub:beta}

For the PF-AL distribution we have 
\begin{align}
f\left(w_{i}|x_i;\mu,\beta,\kappa,\theta\right) 
& \propto  f\left(w_{i},x_i;\mu,\beta,\kappa,\theta\right) \nonumber \\
& \propto \frac{1}{\eta \left(x_i,\mu,\beta,\kappa,\theta\right)}  
f_{\text{G}}\left[w_i;\theta+1,\delta\left(x_i;\mu,\beta,\kappa\right) \right] ,   
\label{eq:posterior_beta} 
\end{align} 
where 
$$
\eta \left(x_i;\mu,\beta,\kappa,\theta\right)= 1- \frac{ \Gamma\left[\theta+1, \delta\left(x_i;\mu,\beta,\kappa\right) \right]}{\Gamma(\theta+1)}.
$$
This means that $W_{i}|X_i=x_i$ has a right-truncated gamma distribution (see, e.g., \citealp{coffey2000properties}), on the interval $\left(0,1\right)$, with shape $\theta+1$ and rate $\delta\left(x_i;\mu,\beta,\kappa\right)$, whose pdf is given in \eqref{eq:posterior_beta}; in symbols 
\begin{equation}
W_{i}|X_i=x_i \sim \mathcal{RTG}_{\left(0,1\right)}\left(\theta+1,\delta\left(x_i;\mu,\beta,\kappa\right)  \right).
\label{eq:Doubly Truncated Gamma}
\end{equation}
The functions $g(W_i)$ arising from \eqref{eq:ALSM complete-data log-likelihood mu, Gamma and Sigma} and \eqref{eq:complete-data log-likelihood beta} are $g_1\left(w\right)=w$ and $g_2\left(w\right)=\log\left(w\right)$.
Thanks to \eqref{eq:Doubly Truncated Gamma} we obtain
\begin{align}
\dot{\mbox{E}}\left(W_{i}|X_i=x_i\right) & =
\frac{1}{\delta\left(x_i;\dot{\mu},\dot{\beta},\dot{\kappa}\right)}
\frac{\Gamma\left(\dot{\theta}+2\right)-\Gamma\left[\dot{\theta}+2, \delta\left(x_i;\dot{\mu},\dot{\beta},\dot{\kappa}\right) \right] }{\Gamma\left(\dot{\theta}+1\right)-\Gamma\left[\dot{\theta}+1, \delta\left(x_i;\dot{\mu},\dot{\beta},\dot{\kappa}\right) \right]}   ,  \nonumber \\
&\eqqcolon \dot{w}_{i}  \nonumber
\end{align}
and
\begin{align}
\dot{\mbox{E}}\left(\log W_{i}|X_i=x_i\right) & =
- \frac{_2F_2\left[\dot{\theta}+1,\dot{\theta}+1;\dot{\theta}+2,\dot{\theta}+2;\delta\left(x_i;\dot{\mu},\dot{\beta},\dot{\kappa}\right)\right]}{(\dot{\theta}+1)^2}
  \nonumber \\
&\eqqcolon \dot{\log}\: w_{i}, \nonumber 
\end{align}
where $_pF_q(a_1,\ldots,a_p;b_1,\ldots,b_q;x)$ is the generalized hypergeometric function \citep{dwork1990generalized}. 

\subsubsection{Pareto asymmetric Laplace distribution}
\label{subsub:pareto}

For the PAR-AL distribution we have 
\begin{align}
f\left(w_{i}|x_i;\mu,\beta,\kappa,\theta\right) 
& \propto  f\left(w_{i},x_i;\mu,\beta,\kappa,\theta\right) \nonumber \\
& \propto \frac{1}{\varphi_{-\theta}(\delta\left(x_i;\mu,\beta,\kappa\right)) }w_i^{-\theta} e^{-w_i \delta\left(x_i;\mu,\beta,\kappa\right) } 
\label{eq:posterior_pareto} 
\end{align} 
which, for $\theta<1$, reduces to a left-truncated gamma distribution on the interval $(1,\infty)$. 
The function $g(W_i)$ arising from \eqref{eq:ALSM complete-data log-likelihood mu, Gamma and Sigma} and \eqref{eq:complete-data log-likelihood pareto} is $g\left(w\right)=\log\left(w\right)$.
Thanks to the pdf in \eqref{eq:posterior_pareto} we obtain
\begin{align}
\dot{\mbox{E}}\left(\log W_{i}|X_i=x_i\right) & =
G_{2,3}^{3,0}\left(\delta\left(x_i;\mu,\beta,\kappa\right);
\begin{matrix} 
 \theta,\theta    \\
 0,\theta-1, \theta-1 
\end{matrix}
\right)
\nonumber \\
&\eqqcolon \dot{\log}\: w_{i}, \nonumber
\end{align}
where $G$ is the Meijer G-function \citep[see Chapter 16 in][]{askey2010generalized}.


\subsection{M-step}
\label{subsec:M-step}

The M-step requires the calculation of $\left\{\ddot{\mu},\ddot{\beta},\ddot{\kappa},\ddot{\btheta}\right\}$ as the value of $\left\{\mu,\beta,\kappa,\btheta\right\}$ that maximizes $Q\left(\mu,\beta,\kappa,\btheta|\dot{\mu},\dot{\beta},\dot{\kappa},\dot{\btheta}\right)$.
According to the right-hand side of \eqref{eq:Q}, $Q_1\left(\mu,\beta,\kappa|\dot{\mu},\dot{\beta},\dot{\kappa},\dot{\btheta}\right)$ and $Q_2\left(\btheta|\dot{\mu},\dot{\beta},\dot{\kappa},\dot{\btheta}\right)$ can be maximized separately with respect to the parameters they involve, with the maximization of $Q_1\left(\mu,\beta,\kappa|\dot{\mu},\dot{\beta},\dot{\kappa},\dot{\btheta}\right)$ being independent of the ALSM considered.
$Q_1$ is also the function to be maximized in the first CM-step of the ECME algorithm related to the U-AL distribution to obtain the updates of $\mu$, $\beta$ and $\kappa$.

$Q_1\left(\mu,\beta,\kappa|\dot{\mu},\dot{\beta},\dot{\kappa},\dot{\btheta}\right)$ can be maximized following an approach similar to the one used to find ML estimators for the parameters of the AL distribution \citep[see, e.g.,][]{kotz2001laplace}.
To find $\mu$, $\beta$, and $\kappa$ we can proceed as follows.
\begin{description}
\item[Step 1:] Consider the functions 
 \begin{equation*}
a(\mu) = \sum_{i=1}^n \dot{w}_i \left(x_i-\mu \right)^+, \quad b(\mu) = \sum_{i=1}^n \dot{w}_i \left(x_i-\mu \right)^-, 
\label{eq:a e b}
\end{equation*}
and
 \begin{equation*}
g(\mu) =    2 \log \left(\sqrt{a(\mu)} +  \sqrt{b(\mu)} \right) + \sqrt{a(\mu)}   \sqrt{b(\mu)},
\label{eq:h}
\end{equation*}
as defined in \eqref{eq:aeb} and \eqref{eq:functiong}, respectively.
Evaluate the $n$ values $g(x_{(j)})$, $j=1,\ldots,n$, and choose a positive integer $r\leq n$ such that
$$
g(x_{(r)}) \leq g(x_{(j)}), \quad  j=1,\ldots,n.
$$
\item[Step 2:] Set $\ddot{\mu}=x_{(r)}$ as defined in \ref{app:Case 2} and use the results in \ref{app:Case 1} to find solutions for $\beta$ and $\kappa$.
In particular, there are three scenarios in Step 2.
\begin{description}
	\item[--]  If $r=1$, then as in the situation i) in \ref{app:Case 1} the solutions of $\beta$ and $\kappa$ do not exist. 
	In this case the ML estimates for the ALSM can not be found by using the EM algorithm;
	\item[--]  If $r=n$, then as in the situation ii) in \ref{app:Case 1} the solutions of $\beta$ and $\kappa$ do not exist. 
	As before the ML estimates for the ALSM can not be found by using the EM algorithm.
	\item[--]  If $1<r<n$, then the parameters maximizing $Q_1\left(\mu,\beta,\kappa|\dot{\mu},\dot{\beta},\dot{\kappa},\dot{\btheta}\right)$ are
	\begin{equation} 
\begin{array}{l} 
\ddot{\mu}     = x_{(r)}             \\
\ddot{\beta}    = \sqrt[4]{a\left(\ddot{\mu}\right) b\left(\ddot{\mu}\right)}\left[\sqrt[4]{a\left(\ddot{\mu}\right)} + \sqrt[4]{b\left(\ddot{\mu}\right)}\right]            \\
\ddot{\kappa}   = \sqrt[4]{\displaystyle\frac{b\left(\ddot{\mu}\right)}{a\left(\ddot{\mu}\right)}} ;
  \end{array}
\label{eq:estimates}
\end{equation}
see also the results in \ref{app:Case 1}.
\end{description}
\end{description}
We provide below details about the maximization of $Q_2\left(\btheta|\dot{\mu},\dot{\beta},\dot{\btheta}\right)$ for each ALSM.

\subsubsection{Shifted exponential asymmetric Laplace distribution}
\label{subsub:Mstepshifted}

For the SE-AL,
$$
Q_2\left(\theta|\dot{\mu},\dot{\beta},\dot{\kappa},\dot{\theta}\right) = n\log \theta  - \theta \sum_{i=1}^n\left(\dot{w}_i-1\right)
$$
is the log-likelihood function of $n$ independent observations $\dot{w}_1,\ldots,\dot{w}_n$ from a shifted exponential distribution with parameter $\theta$.
Therefore, from the standard theory about the exponential distribution \citep[see, e.g.,][Chapter~19]{John:Kotz:cont1:1970}, the update for $\theta$ is
\begin{equation*}
\ddot{\theta} = \frac{n}{\displaystyle\sum_{i=1}^n\left(\dot{w}_{i} - 1\right)}.
\label{eq:M-step shifted}
\end{equation*}

\subsubsection{Unimodal gamma asymmetric Laplace distribution}
\label{subsub:Mstepgamma}

For the UG-AL distribution, a closed-form update for $\theta$ does not exist.
The function to maximize is
\begin{equation}
Q_2\left(\theta|\dot{\mu},\dot{\beta},\dot{\kappa},\dot{\theta}\right) = 
\frac{1}{\theta} \left(\sum_{i=1}^n \dot{\log}\: w_i 
-  \sum_{i=1}^n \dot{w}_i \right)
- n \left(\frac{1}{\theta}+1\right) \log \theta
- n \log\left[\Gamma\left(\frac{1}{\theta}+1\right)\right]
\label{eq:theta UG-L loglik}
\end{equation}
whose derivative, with respect to $\theta$, is
\begin{equation}
\frac{\partial Q_2\left(\theta|\dot{\mu},\dot{\beta},\dot{\kappa},\dot{\theta}\right)}
{\partial \theta} = -\frac{1}{\theta^2} \left( \sum_{i=1}^n \dot{\log}\: w_i + \sum_{i=1}^n \dot{w}_i \right)
+ \frac{n}{\theta^2}\left[\psi\left(\frac{1}{\theta}+1\right) +\log \theta - 1\right] -\frac{n}{\theta}.
\label{eq:theta UG-L loglik derivative}
\end{equation}
The update for $\theta$ can be obtained numerically either by maximizing \eqref{eq:theta UG-L loglik} over $\theta$ or by finding the root (over $\theta$) of the equation obtained equating \eqref{eq:theta UG-L loglik derivative} to zero.

In \ref{app:M-step G-AL} we provide the analogous quantities for the reparameterized UG-AL distribution which uses the classical gamma as mixing distribution.

\subsubsection{Inverse Gaussian asymmetric Laplace distribution}
\label{subsub:Mstepinverse}

For the IG-AL distribution, a closed-form update for $\theta$ does not exist.
The function to maximize is
\begin{equation}
Q_2\left(\theta|\dot{\mu},\dot{\beta},\dot{\kappa},\dot{\theta}\right) = 
\frac{n}{2} \log \left(\frac{1+ 3\theta}{2\pi\theta}\right)
- \frac{3}{2}  \sum_{i=1}^n  \dot{\log}\:w_i 
-\frac{1}{2\theta} 
\left[
\sum_{i=1}^n \dot{w}_i - 2\sqrt{1+3\theta} +(1+3\theta)\sum_{i=1}^n \dot{\bar{w}}_i
\right]
\label{eq:theta IG-L loglik}
\end{equation}
where $\dot{\log}\:w_i$ does not need to be calculated since no parameters are related to it.
The derivative, with respect to $\theta$, is
\begin{equation}
\frac{\partial Q_2\left(\theta|\dot{\mu},\dot{\beta},\dot{\kappa},\dot{\theta}\right)}
{\partial \theta} = -\frac{n}{2\theta(1+3\theta)}
+ \frac{1}{2\theta^2} \sum_{i=1}^n \left(\dot{\bar{w}}_i - \dot{w}_i\right) + \frac{n\left(1-\frac{5}{2}\theta\right)}{\theta^2\sqrt{1+3\theta}}.
\label{eq:theta IG-L loglik derivative}
\end{equation}
As for the UG-AL distribution, the update for $\theta$ can be obtained numerically either by maximizing \eqref{eq:theta IG-L loglik} over $\theta$ or by finding the root (over $\theta$) of the equation obtained equating \eqref{eq:theta IG-L loglik derivative} to zero.

\subsubsection{Power-function asymmetric Laplace distribution}
\label{subsubsec:Mstep beta}

For the PF-AL distribution, the function to be maximized to obtain the update for $\theta$ is
$$
Q_2\left(\theta|\dot{\mu},\dot{\beta},\dot{\kappa},\dot{\theta}\right) = n\log \theta  - \left(\theta -1\right) \sum_{i=1}^n\dot{\log}\:w_i. 
$$
The closed-form update for $\theta$ is
\begin{equation*}
\ddot{\theta} = - \frac{n}{\displaystyle\sum_{i=1}^n \dot{\log}\:w_i}.
\label{eq:M-step beta}
\end{equation*}

\subsubsection{Pareto asymmetric Laplace distribution}
\label{subsubsec:Mstep Pareto}

For the PAR-AL,
$$
Q_2\left(\theta|\dot{\mu},\dot{\beta},\dot{\kappa},\dot{\theta}\right) = 
n\log \theta  
-\left(\theta+1\right) \sum_{i=1}^n  \dot{\log}\: w_{i}  
$$
is the log-likelihood function of $n$ independent observations $\dot{w}_1,\ldots,\dot{w}_n$ from a Pareto distribution.
Therefore, from the standard theory about the Pareto distribution \citep[see, e.g.,][]{rytgaard1990estimation}, the update for $\theta$ is
\begin{equation*}
\ddot{\theta} = \frac{n}{\sum_{i=1}^n  \dot{\log}\: w_{i}}.
\label{eq:M-step pareto}
\end{equation*}

\subsubsection{Uniform asymmetric Laplace distribution}
\label{subsubsec:Mstep Uniform}

For the U-AL distribution the update $\ddot{\theta}$ of $\theta$ is computed at the second CM-step of the ECME algorithm by maximizing the observed-data log-likelihood function $l\left(\ddot{\mu},\ddot{\beta},\ddot{\kappa},\theta\right)$ over $\theta$.

\section{Data and results}
\label{sec:Real data analysis}

\subsection{Data}
\label{subsec:Data}

We downloaded the daily adjusted close prices of the following two cryptocurrencies: Bitcoin EUR (BTC-EUR) and TRON EUR (TRX-EUR).
All prices are in Euro. 
The data were downloaded from \url{https://finance.yahoo.com/cryptocurrencies}.
The period under investigation goes from 1 April 2018 to 31 May 2022. 
Returns are estimated by taking logarithmic differences.

We conduct the whole analysis in \textsf{R} \citep{R:2018}, and the codes needed to replicate the analysis are available upon request.
\tablename~\ref{tab:desriptives} shows some descriptive statistics of the considered cryptocurrency returns.
\tablename~\ref{tab:desriptives} also reports the $p$-values of the D'Agostino test of skewness and of three commonly employed tests of normality: Anderson-Darling, Jarque-Bera and Shapiro-Wilk \citep[see][for a comparison]{Yap:Sim:JSCS:2011}.
In the following we will compare the $p$-values with the classical 0.05 significance level. 
Regardless of the considered cryptocurrency, mean and median are very close to zero.
TRX-EUR has a slightly larger variability, as measured by the standard deviation and range.
All the return distributions are skewed (according to the D'Agostino test, regardless of the considered significance level) and, in agreement with the literature, they are clearly non-normal regardless of both the considered normality test and significance level considered ($p$-values very close to 0).
Finally, the return distributions are highly leptokurtic with an excess kurtosis of 15.966 for BTC-EUR and of 8.249 for TRX-EUR.
\begin{table}[!ht]  
\centering      
 \begin{tabular}{lrr}        
\toprule        
	&	BTC-EUR	&	 TRX-EUR		\\
	\midrule
$n$	&	1521	&	1521		\\
Mean	    &	0.001	    &	 0.001		\\
Median	  &	0.001	    &	 0.001		\\
St.~Dev.	&	0.038	    &	 0.057		\\
Skewness	&	-1.079	  &	 -0.434 	\\
Kurtosis	&	18.966    &	 11.249	\\
Minimum	  &	-0.458  	&	 -0.516	\\
Maximum	  &	0.195	    &	 0.334	\\
D'Agostino test ($p$-value)	      &	0.000	&	0.000 	\\
Anderson-Darling test ($p$-value) & 0.000 & 0.000   \\
Jarque-Bera test ($p$-value)      & 0.000 & 0.000   \\
Shapiro-Wilk test ($p$-value)     & 0.000 & 0.000   \\
 \bottomrule        
\end{tabular}
\caption{Summary statistics of the considered cryptocurrency daily returns.}        
\label{tab:desriptives}        
\end{table}
  
The histograms of the cryptocurrency returns appear in \figurename~\ref{fig:histograms}; to make the comparison easier, the histograms share the same axes range.
Regardless of the considered cryptocurrency, we observe the shape typical of the asymmetric Laplace distribution.
Indeed, the distributions have high peaks near zero and appear to have tails heavier than those of the normal distribution.
However, the empirical kurtoses in \tablename~\ref{tab:desriptives} are not included in $\left[6,9\right)$, the interval of kurtoses allowed by the asymmetric Laplace distribution (refer to \figurename~\ref{fig:AL skewness and kurtosis}).
In particular, they are larger. 
\begin{figure}[!ht]
\centering
\subfigure[BTC-EUR\label{fig:BTC}]
{\resizebox{0.482\textwidth}{!}{\includegraphics{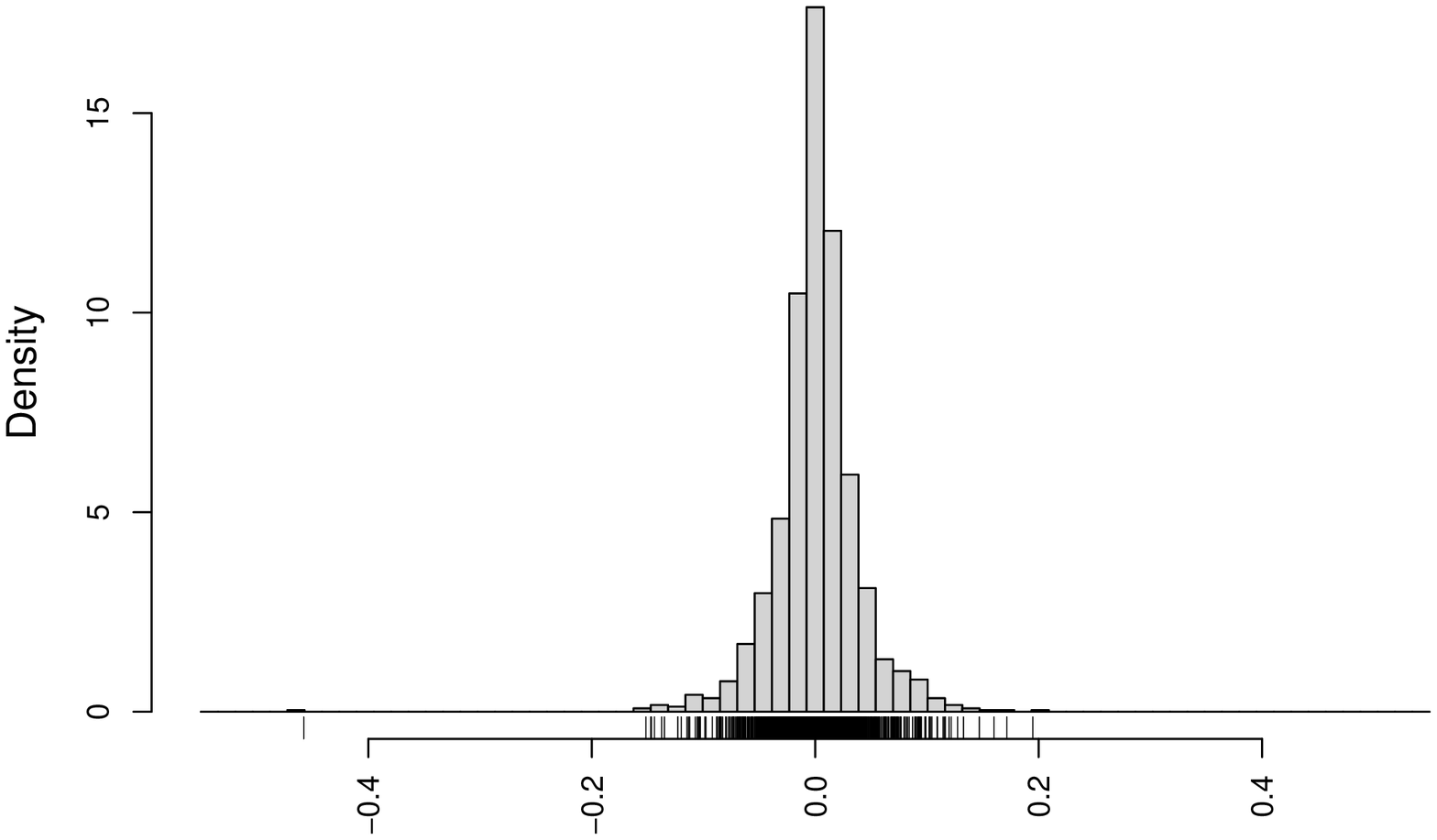}}}
\subfigure[TRX-EUR\label{fig:TRX}]
{\resizebox{0.482\textwidth}{!}{\includegraphics{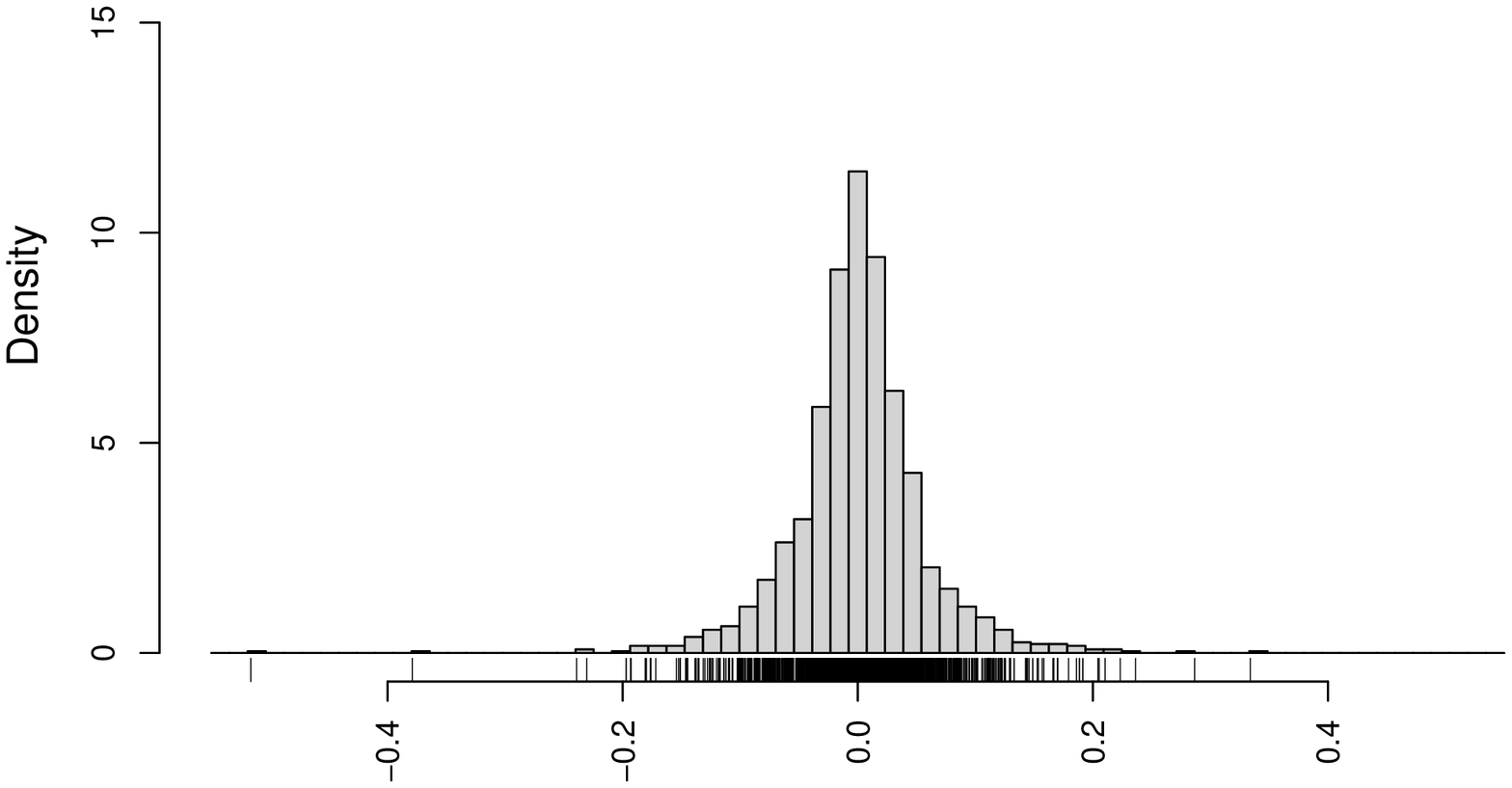}}}
\caption{
\footnotesize 
Histograms of the considered cryptocurrency daily returns.
\label{fig:histograms}
}
\end{figure}	
	
\subsection{Results}
\label{subsec:Results}

Motivated by the findings above, especially those related to shape, skewness and kurtosis of the empirical distributions, we fit our ALSMs to the considered cryptocurrency returns.
For the sake of comparison, we also fit several parametric distributions defined on the whole real line.
We take the competing models among those implemented in the \textbf{gamlss} package \citep{rigby2005generalized} for \textsf{R}.
In detail, we consider the following distributions, given in alphabetic order using the nomenclature by \citet{rigby2019distributions}: exponential Gaussian (exGAUS), generalized $t$ (GT), Gumbel (GU), Johnson's SU (JSU), logistic (LO), NET (NET), normal (NO), normal family (NOF), power exponential (PE), reverse Gumbel (RG), sinh-arcsinh (SHASH), sinh-arcsinh original (SHASHo), skew normal type 1 (SN1), skew normal type 2 (SN2), skew power exponential type 1 (SEP1), skew power exponential type 2 (SEP2), skew power exponential type 3 (SEP3), skew power exponential type 4 (SEP4), skew Student $t$ (SST), skew $t$ type 1 (ST1), skew $t$ type 2 (ST2), skew $t$ type 4 (ST4), skew $t$ type 5 (ST5), and $t$ family (TF).
Moreover, for the sake of completeness, and motivated by the results in \citet{chu2015statistical} and \citet{Chan:2017}, we also fit the generalyzed hyperbolic (GH) and some of its special cases which are not implemented in the \textbf{gamlss} package, namely hyperbolic (H), variance gamma (VG), and normal inverse Gaussian (NIG) distributions.
These models are fitted via the \texttt{fit.ghypuv()} function of the \textbf{ghyp} package.
This yields a total of 36 competing models.  
We estimate the parameters of all the models under consideration via the ML approach.

To compare models with the same number of parameters, in terms of goodness-of-fit, we use the log-likelihood (in addition to the criteria described below).
We accomplish the comparison of models with differing number of parameters, as usual, via the Akaike information criterion (AIC; \citealp{AIC}) and the Bayesian information criterion (BIC; \citealp{Schw:Esti:1978}) that, in our formulation, need to be maximized because they are multiplied by $-1$.
Moreover, we use the likelihood-ratio (LR) test to compare each Laplace-based model (alternative model) with the nested asymmetric Laplace distribution (null model).
In particular, the LR test can be used to determine whether the alternative model is a significant improvement over the asymmetric Laplace distribution. 
The test statistic is
$$
\text{LR} = -2\left[l(\widehat{\mu},\widehat{\beta},\widehat{\kappa})-l(\widehat{\mu},\widehat{\beta},\widehat{\kappa},\widehat{\btheta})\right],
$$
where $\widehat{\mu}$, $\widehat{\beta}$, $\widehat{\kappa}$ and $\widehat{\btheta}$ are the ML estimates of $\mu$, $\beta$, $\kappa$ and $\btheta$, respectively, and where $l(\widehat{\mu},\widehat{\beta},\widehat{\kappa})$ and $l(\widehat{\mu},\widehat{\beta},\widehat{\kappa},\widehat{\btheta})$ are the maximized log-likelihood values under the null and alternative models, respectively. 
Under the null hypothesis of no improvement, using Wilks' theorem, LR can be approximated by a $\chi^2$ random variable with number of degrees of freedom given by the difference in the number of estimated parameters between the alternative and the null model, and this allows us to compute a $p$-value.

\tablename s~\ref{tab:BTC}--\ref{tab:TRX} present the model comparison separately for each cryptocurrency.
To easy the reader in comparing the performance of the considered parametric models, the tables also give rankings induced by AIC, BIC, and LR tests (limited to the AL-based models).
\begin{table}[!ht]																					
 \resizebox{\textwidth}{!}{
 \begin{tabular}{lc@{\hskip 0.3in}r@{\hskip 0.3in}rr@{\hskip 0.3in}rr@{\hskip 0.3in}rr}	
\toprule													
Model	&	$\#$ par	&	Log-lik	&	AIC	&	Rank	&	BIC	&	Rank	&	LR $p$-value	&	Rank		\\
\midrule																					
AL	&	3	&	3020.921	&	6035.842	&	18	&	6019.861	&	15	&		&		\\
TP-AL	&	5	&	3027.713	&	6045.426	&	7	&	6018.790	&	17	&	0.001	&	6	\\
SE-AL	&	4	&	3026.311	&	6044.621	&	8	&	6023.313	&	7	&	0.001	&	5	\\
UG-AL	&	4	&	3027.701	&	6047.402	&	4	&	6026.094	&	5	&	0.000	&	3	\\
IG-AL	&	4	&	3027.388	&	6046.776	&	5	&	6025.467	&	6	&	0.000	&	4	\\
PF-AL	&	4	&	3028.167	&	6048.334	&	2	&	6027.026	&	3	&	0.000	&	2	\\
P-AL	&	4	&	3025.807	&	6043.614	&	13	&	6022.306	&	12	&	0.002	&	7	\\
U-AL	&	4	&	3047.685	&	6087.370	&	1	&	6066.062	&	1	&	0.000	&	1	\\
exGAUS	&	3	&	2818.553	&	5631.105	&	31	&	5615.124	&	32	&		&		\\
GH	&	5	&	3025.958	&	6041.917	&	16	&	6015.281	&	18	&		&		\\
GT	&	4	&	3027.728	&	6047.455	&	3	&	6026.147	&	4	&		&		\\
GU	&	2	&	2593.517	&	5183.034	&	35	&	5172.380	&	35	&		&		\\
H	&	4	&	3020.920	&	6033.841	&	20	&	6012.532	&	20	&		&		\\
JSU	&	4	&	3018.886	&	6029.773	&	21	&	6008.464	&	21	&		&		\\
LO	&	2	&	2962.209	&	5920.419	&	29	&	5909.765	&	29	&		&		\\
NET	&	2	&	2988.639	&	5969.277	&	28	&	5947.969	&	28	&		&		\\
NIG	&	4	&	3021.583	&	6035.167	&	19	&	6013.858	&	19	&		&		\\
NO	&	2	&	2816.573	&	5629.146	&	32	&	5618.491	&	31	&		&		\\
NOF	&	3	&	2816.573	&	5627.146	&	33	&	5611.164	&	33	&		&		\\
PE	&	3	&	3026.119	&	6046.239	&	6	&	6030.257	&	2	&		&		\\
RG	&	2	&	1909.950	&	3815.899	&	36	&	3805.245	&	36	&		&		\\
SHASH	&	4	&	3025.711	&	6043.422	&	14	&	6022.113	&	13	&		&		\\
SHASHo	&	4	&	3025.522	&	6043.044	&	15	&	6021.735	&	14	&		&		\\
SN1	&	3	&	2816.573	&	5627.146	&	34	&	5611.164	&	34	&		&		\\
SN2	&	3	&	2820.347	&	5634.693	&	30	&	5618.712	&	30	&		&		\\
SEP1	&	4	&	3026.126	&	6044.253	&	11	&	6022.944	&	10	&		&		\\
SEP2	&	4	&	3026.146	&	6044.292	&	10	&	6022.984	&	9	&		&		\\
SEP3	&	4	&	3026.120	&	6044.241	&	12	&	6022.932	&	11	&		&		\\
SEP4	&	4	&	3026.208	&	6044.416	&	9	&	6023.108	&	8	&		&		\\
SST	&	4	&	3012.351	&	6016.702	&	27	&	5995.394	&	27	&		&		\\
ST1	&	4	&	3012.427	&	6016.853	&	23	&	5995.545	&	23	&		&		\\
ST2	&	4	&	3012.360	&	6016.719	&	26	&	5995.411	&	26	&		&		\\
ST4	&	4	&	3012.363	&	6016.725	&	25	&	5995.417	&	25	&		&		\\
ST5	&	4	&	3012.363	&	6016.726	&	24	&	5995.418	&	24	&		&		\\
TF	&	3	&	3012.320	&	6018.641	&	22	&	6002.659	&	22	&		&		\\
VG	&	4	&	3024.470	&	6040.940	&	17	&	6019.631	&	16	&		&		\\
\bottomrule
\end{tabular}			
}																		
\caption{BTC-EUR: number of parameters ($\#$ par), log-likelihood (Log-lik), AIC, and BIC for the competing models, along with rankings from these criteria.
The $p$-values from the LR tests for the asymmetric Laplace-based models are given along with rankings.
}		
\label{tab:BTC}																			
\end{table}	
\begin{table}[!ht]																					
 \resizebox{\textwidth}{!}{
 \begin{tabular}{lc@{\hskip 0.3in}r@{\hskip 0.3in}rr@{\hskip 0.3in}rr@{\hskip 0.3in}rr}	
\toprule													
Model	&	$\#$ par	&	Log-lik	&	AIC	&	Rank	&	BIC	&	Rank	&	LR $p$-value	&	Rank		\\
\midrule	
AL	&	3	&	2383.212	&	4760.423	&	20	&	4744.442	&	8	&		&		\\
TP-AL	&	5	&	2388.481	&	4766.963	&	5	&	4740.327	&	18	&	0.005	&	5	\\
SE-AL	&	4	&	2386.362	&	4764.723	&	11	&	4743.415	&	11	&	0.012	&	6	\\
UG-AL	&	4	&	2387.646	&	4767.293	&	4	&	4745.984	&	5	&	0.003	&	3	\\
IG-AL	&	4	&	2387.405	&	4766.810	&	6	&	4745.501	&	6	&	0.004	&	4	\\
PF-AL	&	4	&	2388.365	&	4768.731	&	3	&	4747.422	&	4	&	0.001	&	2	\\
P-AL	&	4	&	2385.815	&	4763.630	&	14	&	4742.322	&	13	&	0.022	&	7	\\
U-AL	&	4	&	2409.962	&	4811.925	&	1	&	4790.617	&	1	&	0.000	&	1	\\
exGAUS	&	3	&	2213.975	&	4421.950	&	30	&	4405.969	&	30	&		&		\\
GH	&	5	&	2388.037	&	4766.073	&	7	&	4739.438	&	20	&		&		\\
GT	&	4	&	2388.430	&	4768.861	&	2	&	4747.552	&	3	&		&		\\
GU	&	2	&	1902.159	&	3800.318	&	35	&	3789.663	&	35	&		&		\\
H	&	4	&	2383.211	&	4758.423	&	21	&	4737.114	&	22	&		&		\\
JSU	&	4	&	2384.976	&	4761.952	&	18	&	4740.644	&	17	&		&		\\
LO	&	2	&	2333.869	&	4663.738	&	29	&	4653.084	&	29	&		&		\\
NET	&	2	&	2360.386	&	4712.772	&	28	&	4691.463	&	28	&		&		\\
NIG	&	4	&	2386.959	&	4765.918	&	8	&	4744.610	&	7	&		&		\\
NO	&	2	&	2202.754	&	4401.507	&	31	&	4390.853	&	31	&		&		\\
NOF	&	3	&	2202.626	&	4399.253	&	34	&	4383.271	&	34	&		&		\\
PE	&	3	&	2385.124	&	4764.247	&	12	&	4748.266	&	2	&		&		\\
RG	&	2	&	1587.234	&	3170.468	&	36	&	3159.814	&	36	&		&		\\
SHASH	&	4	&	2386.467	&	4764.933	&	10	&	4743.625	&	10	&		&		\\
SHASHo	&	4	&	2386.468	&	4764.935	&	9	&	4743.627	&	9	&		&		\\
SN1	&	3	&	2202.754	&	4399.507	&	33	&	4383.526	&	33	&		&		\\
SN2	&	3	&	2203.538	&	4401.075	&	32	&	4385.094	&	32	&		&		\\
SEP1	&	4	&	2385.688	&	4763.376	&	16	&	4742.068	&	15	&		&		\\
SEP2	&	4	&	2385.830	&	4763.660	&	13	&	4742.351	&	12	&		&		\\
SEP3	&	4	&	2385.763	&	4763.526	&	15	&	4742.218	&	14	&		&		\\
SEP4	&	4	&	2385.322	&	4762.644	&	17	&	4741.335	&	16	&		&		\\
SST	&	4	&	2380.241	&	4752.482	&	23	&	4731.174	&	23	&		&		\\
ST1	&	4	&	2380.229	&	4752.458	&	24	&	4731.149	&	24	&		&		\\
ST2	&	4	&	2380.222	&	4752.444	&	25	&	4731.136	&	25	&		&		\\
ST4	&	4	&	2380.207	&	4752.414	&	27	&	4731.106	&	27	&		&		\\
ST5	&	4	&	2380.218	&	4752.435	&	26	&	4731.127	&	26	&		&		\\
TF	&	3	&	2380.206	&	4754.412	&	22	&	4738.430	&	21	&		&		\\
VG	&	4	&	2384.530	&	4761.061	&	19	&	4739.752	&	19	&		&		\\
\bottomrule
\end{tabular}			
}																		
\caption{TRX-EUR: number of parameters ($\#$ par), log-likelihood (Log-lik), AIC, and BIC for the competing models, along with rankings from these criteria.
The $p$-values from the LR tests for the asymmetric Laplace-based models are given along with rankings.
}		
\label{tab:TRX}																			
\end{table}	
For the BTC-EUR, \tablename~\ref{tab:BTC} shows that both AIC and BIC indicate one of our models, the U-AL, as the best one.
Furthermore, six out of the seven models we propose (all except P-AL) are in the first eight positions according to the AIC (GT and PE occupy the third and sixth positions, respectively) while five out of our seven models (all except P-AL and TP-AL) are in the first six positions according to the BIC (PE and GT occupy the second and fourth positions, respectively).
These findings are further corroborated by the ranking induced by the LR-test, with the U-AL in the first position and the TP-AL and P-AL in the penultimate and last position, respectively.
Moreover, at the 5\% significance level, all our models represent a significant improvement over the asymmetric Laplace.

For the TRX-EUR (refer to \tablename~\ref{tab:TRX}) we observe similar results.
The U-AL is still the best model for AIC, BIC and LR test.
Moreover, based on the $p$-values from the LR test, all our models still represent a significant improvement over the asymmetric Laplace at the 5\% significance level.


A final interesting aspect to be noted is that, regardless of the considered cryptocurrency, the Gumbel-based models (RG and GU) are the worst models according to AIC and BIC.

\section{Conclusions and future works}
\label{sec:Conclusions and future works}

In this paper we further corroborate some recent findings about the cryptocurrency return distribution; in particular, we confirm it can be highly-peaked, skewed, and heavy-tailed, with a large excess kurtosis.
To account for all these peculiarities, in a modeling perspective, we propose the asymmetric Laplace scale mixture (ALSM) family of distributions and, for illustrative purposes, we consider seven different members of this family.
Advantageously with respect to the nested asymmetric Laplace distribution, the members of our family allow for a wider range of skewness values and for a larger excess kurtosis.
These improvements allows the ALSM to be a convenient model for the distribution of cryptocurrency returns. 

However, the novelty of the paper is not limited to the five illustrated models and the usefulness of our family is not restricted to cryptocurrencies. 
Indeed, other members of the family may be easily defined and other financial contexts, like those outlined in \citet[][Chapter~8]{Kotz:Kozu:Podg:TheL:2012}, may benefit from the use of ALSMs.


\begin{thebibliography}{49}
\expandafter\ifx\csname natexlab\endcsname\relax\def\natexlab#1{#1}\fi
\providecommand{\url}[1]{\texttt{#1}}
\providecommand{\href}[2]{#2}
\providecommand{\path}[1]{#1}
\providecommand{\DOIprefix}{doi:}
\providecommand{\ArXivprefix}{arXiv:}
\providecommand{\URLprefix}{URL: }
\providecommand{\Pubmedprefix}{pmid:}
\providecommand{\doi}[1]{\href{http://dx.doi.org/#1}{\path{#1}}}
\providecommand{\Pubmed}[1]{\href{pmid:#1}{\path{#1}}}
\providecommand{\bibinfo}[2]{#2}
\ifx\xfnm\relax \def\xfnm[#1]{\unskip,\space#1}\fi
\bibitem[{Abramowitz \& Stegun(1965)}]{Abra:Steg:Hand:1965}
\bibinfo{author}{Abramowitz, M.}, \& \bibinfo{author}{Stegun, I.~A.}
  (\bibinfo{year}{1965}).
\newblock {\it \bibinfo{title}{Handbook of Mathematical Functions: With
  Formulas, Graphs, and Mathematical Tables}\/} volume~\bibinfo{volume}{55} of
  {\it \bibinfo{series}{Applied Mathematics Series}\/}.
\newblock \bibinfo{address}{New York}: \bibinfo{publisher}{Dover Publications}.
\bibitem[{Ahsanullah \& Kabir(1974)}]{Ahsa:CJS:1974}
\bibinfo{author}{Ahsanullah, M.}, \& \bibinfo{author}{Kabir, A.~B.~M.~L.}
  (\bibinfo{year}{1974}).
\newblock \bibinfo{title}{A characterization of the power function
  distribution}.
\newblock {\it \bibinfo{journal}{The Canadian Journal of Statistics}\/},  {\it
  \bibinfo{volume}{2}\/}, \bibinfo{pages}{95--98}.
\bibitem[{Akaike(1974)}]{AIC}
\bibinfo{author}{Akaike, H.} (\bibinfo{year}{1974}).
\newblock \bibinfo{title}{A new look at the statistical model identification}.
\newblock {\it \bibinfo{journal}{IEEE Transactions on Automatic Control}\/},
  {\it \bibinfo{volume}{19}\/}, \bibinfo{pages}{716--723}.
\bibitem[{Askey \& Daalhuis(2010)}]{askey2010generalized}
\bibinfo{author}{Askey, R.~A.}, \& \bibinfo{author}{Daalhuis, A. B.~O.}
  (\bibinfo{year}{2010}).
\newblock \bibinfo{title}{Generalized hypergeometric functions and {M}eijer
  ${G}$-function.}
\bibitem[{Bagnato \& Punzo(2013)}]{Bagn:Punz:Fine:2013}
\bibinfo{author}{Bagnato, L.}, \& \bibinfo{author}{Punzo, A.}
  (\bibinfo{year}{2013}).
\newblock \bibinfo{title}{Finite mixtures of unimodal beta and gamma densities
  and the $k$-bumps algorithm}.
\newblock {\it \bibinfo{journal}{Computational Statistics}\/},  {\it
  \bibinfo{volume}{28}\/}, \bibinfo{pages}{1571--1597}.
\bibitem[{Bariviera et~al.(2017)Bariviera, Basgall, Hasperu{\'e} \&
  Naiouf}]{Bariviera:2017}
\bibinfo{author}{Bariviera, A.~F.}, \bibinfo{author}{Basgall, M.~J.},
  \bibinfo{author}{Hasperu{\'e}, W.}, \& \bibinfo{author}{Naiouf, M.}
  (\bibinfo{year}{2017}).
\newblock \bibinfo{title}{Some stylized facts of the {B}itcoin market}.
\newblock {\it \bibinfo{journal}{Physica A: Statistical Mechanics and its
  Applications}\/},  {\it \bibinfo{volume}{484}\/}, \bibinfo{pages}{82--90}.
\bibitem[{Chan et~al.(2017)Chan, Chu, Nadarajah \& Osterrieder}]{Chan:2017}
\bibinfo{author}{Chan, S.}, \bibinfo{author}{Chu, J.},
  \bibinfo{author}{Nadarajah, S.}, \& \bibinfo{author}{Osterrieder, J.}
  (\bibinfo{year}{2017}).
\newblock \bibinfo{title}{A statistical analysis of cryptocurrencies}.
\newblock {\it \bibinfo{journal}{Journal of Risk and Financial Management}\/},
  {\it \bibinfo{volume}{10}\/}, \bibinfo{pages}{1--24}.
\bibitem[{Chen(2000)}]{Chen:prob:2000}
\bibinfo{author}{Chen, S.~X.} (\bibinfo{year}{2000}).
\newblock \bibinfo{title}{Probability density function estimation using gamma
  kernels}.
\newblock {\it \bibinfo{journal}{Annals of the Institute of Statistical
  Mathematics}\/},  {\it \bibinfo{volume}{52}\/}, \bibinfo{pages}{471--480}.
\bibitem[{Chu et~al.(2015)Chu, Nadarajah \& Chan}]{chu2015statistical}
\bibinfo{author}{Chu, J.}, \bibinfo{author}{Nadarajah, S.}, \&
  \bibinfo{author}{Chan, S.} (\bibinfo{year}{2015}).
\newblock \bibinfo{title}{Statistical analysis of the exchange rate of
  bitcoin}.
\newblock {\it \bibinfo{journal}{PLOS One}\/},  {\it \bibinfo{volume}{10}\/},
  \bibinfo{pages}{e0133678}.
\bibitem[{Coffey \& Muller(2000)}]{coffey2000properties}
\bibinfo{author}{Coffey, C.~S.}, \& \bibinfo{author}{Muller, K.~E.}
  (\bibinfo{year}{2000}).
\newblock \bibinfo{title}{Properties of doubly-truncated gamma variables}.
\newblock {\it \bibinfo{journal}{Communications in Statistics-Theory and
  Methods}\/},  {\it \bibinfo{volume}{29}\/}, \bibinfo{pages}{851--857}.
\bibitem[{Davies \& Gather(1993)}]{Davi:Gath:Thei:1993}
\bibinfo{author}{Davies, L.}, \& \bibinfo{author}{Gather, U.}
  (\bibinfo{year}{1993}).
\newblock \bibinfo{title}{The identification of multiple outliers}.
\newblock {\it \bibinfo{journal}{Journal of the American Statistical
  Association}\/},  {\it \bibinfo{volume}{88}\/}, \bibinfo{pages}{782--792}.
\bibitem[{Dempster et~al.(1977)Dempster, Laird \&
  Rubin}]{Demp:Lair:Rubi:Maxi:1977}
\bibinfo{author}{Dempster, A.~P.}, \bibinfo{author}{Laird, N.~M.}, \&
  \bibinfo{author}{Rubin, D.~B.} (\bibinfo{year}{1977}).
\newblock \bibinfo{title}{{Maximum likelihood from incomplete data via the EM
  algorithm}}.
\newblock {\it \bibinfo{journal}{Journal of the Royal Statistical Society:
  Series B (Statistical Methodology)}\/},  {\it \bibinfo{volume}{39}\/},
  \bibinfo{pages}{1--38}.
\bibitem[{Dwork(1990)}]{dwork1990generalized}
\bibinfo{author}{Dwork, B.} (\bibinfo{year}{1990}).
\newblock {\it \bibinfo{title}{Generalized hypergeometric functions}\/}.
\newblock \bibinfo{publisher}{Clarendon Press}.
\bibitem[{Hennig(2002)}]{Henn:Fixe:2002}
\bibinfo{author}{Hennig, C.} (\bibinfo{year}{2002}).
\newblock \bibinfo{title}{Fixed point clusters for linear regression:
  computation and comparison}.
\newblock {\it \bibinfo{journal}{Journal of Classification}\/},  {\it
  \bibinfo{volume}{19}\/}, \bibinfo{pages}{249--276}.
\bibitem[{Johnson \& Kotz(1970{\natexlab{a}})}]{John:Kotz:cont2:1970}
\bibinfo{author}{Johnson, N.~L.}, \& \bibinfo{author}{Kotz, S.}
  (\bibinfo{year}{1970}{\natexlab{a}}).
\newblock {\it \bibinfo{title}{Continuous Univariate Distributions}\/}
  volume~\bibinfo{volume}{2}.
\newblock \bibinfo{address}{New York}: \bibinfo{publisher}{John Wiley \& Sons}.
\bibitem[{Johnson \& Kotz(1970{\natexlab{b}})}]{John:Kotz:cont1:1970}
\bibinfo{author}{Johnson, N.~L.}, \& \bibinfo{author}{Kotz, S.}
  (\bibinfo{year}{1970}{\natexlab{b}}).
\newblock {\it \bibinfo{title}{Continuous Univariate Distributions}\/}
  volume~\bibinfo{volume}{1}.
\newblock \bibinfo{address}{New York}: \bibinfo{publisher}{John Wiley \& Sons}.
\bibitem[{Kotz et~al.(2001)Kotz, Kozubowski \& Podg{\'o}rski}]{kotz2001laplace}
\bibinfo{author}{Kotz, S.}, \bibinfo{author}{Kozubowski, T.}, \&
  \bibinfo{author}{Podg{\'o}rski, K.} (\bibinfo{year}{2001}).
\newblock {\it \bibinfo{title}{The Laplace distribution and generalizations: a
  revisit with applications to communications, economics, engineering, and
  finance}\/}.
\newblock \bibinfo{number}{183}.
\newblock \bibinfo{publisher}{Springer Science \& Business Media}.
\bibitem[{Kotz et~al.(2012)Kotz, Kozubowski \&
  Podgorski}]{Kotz:Kozu:Podg:TheL:2012}
\bibinfo{author}{Kotz, S.}, \bibinfo{author}{Kozubowski, T.~J.}, \&
  \bibinfo{author}{Podgorski, K.} (\bibinfo{year}{2012}).
\newblock {\it \bibinfo{title}{The Laplace Distribution and Generalizations: A
  Revisit with Applications to Communications, Economics, Engineering, and
  Finance}\/}.
\newblock SpringerLink : B{\"u}cher.
\newblock \bibinfo{publisher}{Birkh{\"a}user Boston}.
\bibitem[{Kozubowski \& Podgorski(2000)}]{Kozu:Podg:Asym:2000}
\bibinfo{author}{Kozubowski, T.~J.}, \& \bibinfo{author}{Podgorski, K.}
  (\bibinfo{year}{2000}).
\newblock \bibinfo{title}{Asymmetric laplace distributions}.
\newblock {\it \bibinfo{journal}{Mathematical Scientist}\/},  {\it
  \bibinfo{volume}{25}\/}, \bibinfo{pages}{37--46}.
\bibitem[{Liu \& Rubin(1994)}]{Liu:Rubi:TheE:1994}
\bibinfo{author}{Liu, C.}, \& \bibinfo{author}{Rubin, D.~B.}
  (\bibinfo{year}{1994}).
\newblock \bibinfo{title}{The {ECME} algorithm: a simple extension of {EM} and
  {ECM} with faster monotone convergence}.
\newblock {\it \bibinfo{journal}{Biometrika}\/},  {\it \bibinfo{volume}{81}\/},
  \bibinfo{pages}{633--648}.
\bibitem[{Mazza \& Punzo(2019)}]{Mazz:Punz:Mode:2019}
\bibinfo{author}{Mazza, A.}, \& \bibinfo{author}{Punzo, A.}
  (\bibinfo{year}{2019}).
\newblock \bibinfo{title}{Modeling household income with contaminated unimodal
  distributions}.
\newblock In \bibinfo{editor}{A.~Petrucci}, \bibinfo{editor}{F.~Racioppi}, \&
  \bibinfo{editor}{R.~Verde} (Eds.), {\it \bibinfo{booktitle}{New Statistical
  Developments in Data Science}\/} (pp. \bibinfo{pages}{373--391}).
\newblock \bibinfo{address}{Cham, Switzerland}: \bibinfo{publisher}{Springer}
  volume~\bibinfo{volume}{88} of {\it \bibinfo{series}{Springer Proceedings in
  Mathematics \& Statistics}\/}.
\bibitem[{McLachlan \& Krishnan(2007)}]{McLa:Kris:TheE:2007}
\bibinfo{author}{McLachlan, G.}, \& \bibinfo{author}{Krishnan, T.}
  (\bibinfo{year}{2007}).
\newblock {\it \bibinfo{title}{The {EM} algorithm and extensions}\/} volume
  \bibinfo{volume}{382} of {\it \bibinfo{series}{Wiley Series in Probability
  and Statistics}\/}.
\newblock (\bibinfo{edition}{2nd} ed.).
\newblock \bibinfo{address}{New York}: \bibinfo{publisher}{John Wiley \& Sons}.
\bibitem[{McNeil et~al.(2005)McNeil, Frey \&
  Embrechts}]{McNe:Frey:Embr:Quan:2005}
\bibinfo{author}{McNeil, A.}, \bibinfo{author}{Frey, R.}, \&
  \bibinfo{author}{Embrechts, P.} (\bibinfo{year}{2005}).
\newblock {\it \bibinfo{title}{Quantitative Risk Management: Concepts,
  Techniques and Tools}\/}.
\newblock Princeton Series in Finance.
\newblock \bibinfo{publisher}{Princeton University Press}.
\bibitem[{Melnykov \& Zhu(2018)}]{melnykov2018model}
\bibinfo{author}{Melnykov, V.}, \& \bibinfo{author}{Zhu, X.}
  (\bibinfo{year}{2018}).
\newblock \bibinfo{title}{On model-based clustering of skewed matrix data}.
\newblock {\it \bibinfo{journal}{Journal of Multivariate Analysis}\/},  {\it
  \bibinfo{volume}{167}\/}, \bibinfo{pages}{181--194}.
\bibitem[{Melnykov \& Zhu(2019)}]{melnykov2019studying}
\bibinfo{author}{Melnykov, V.}, \& \bibinfo{author}{Zhu, X.}
  (\bibinfo{year}{2019}).
\newblock \bibinfo{title}{Studying crime trends in the {USA} over the years
  2000--2012}.
\newblock {\it \bibinfo{journal}{Advances in Data Analysis and
  Classification}\/},  {\it \bibinfo{volume}{13}\/}, \bibinfo{pages}{325--341}.
\bibitem[{Meng \& Rubin(1993)}]{Meng:Rubin:Maxi:1993}
\bibinfo{author}{Meng, X.-L.}, \& \bibinfo{author}{Rubin, D.~B.}
  (\bibinfo{year}{1993}).
\newblock \bibinfo{title}{Maximum likelihood estimation via the {ECM}
  algorithm: A general framework}.
\newblock {\it \bibinfo{journal}{Biometrika}\/},  {\it \bibinfo{volume}{80}\/},
  \bibinfo{pages}{267--278}.
\bibitem[{Misra(1940)}]{Misr:1940}
\bibinfo{author}{Misra, R.~D.} (\bibinfo{year}{1940}).
\newblock \bibinfo{title}{On the stability of crystal lattices. {II}}.
\newblock {\it \bibinfo{journal}{Mathematical Proceedings of the Cambridge
  Philosophical Society}\/},  {\it \bibinfo{volume}{36}\/},
  \bibinfo{pages}{173--182}.
\bibitem[{Morris et~al.(2019)Morris, Punzo, McNicholas \&
  Browne}]{Morr:Punz:McNi:Brow:Asym:2019}
\bibinfo{author}{Morris, K.}, \bibinfo{author}{Punzo, A.},
  \bibinfo{author}{McNicholas, P.~D.}, \& \bibinfo{author}{Browne, R.~P.}
  (\bibinfo{year}{2019}).
\newblock \bibinfo{title}{Asymmetric clusters and outliers: Mixtures of
  multivariate contaminated shifted asymmetric laplace distributions}.
\newblock {\it \bibinfo{journal}{Computational Statistics \& Data Analysis}\/},
   {\it \bibinfo{volume}{132}\/}, \bibinfo{pages}{145--166}.
\bibitem[{Osterrieder(2017)}]{Oste:2017}
\bibinfo{author}{Osterrieder, J.} (\bibinfo{year}{2017}).
\newblock \bibinfo{title}{The statistics of bitcoin and cryptocurrencies}.
\newblock In {\it \bibinfo{booktitle}{Proceedings of the 2017 International
  Conference on Economics, Finance and Statistics (ICEFS 2017)}\/}.
\newblock \bibinfo{publisher}{Atlantis Press} volume~\bibinfo{volume}{26} of
  {\it \bibinfo{series}{Advances in Economics, Business and Management
  Research}\/}.
\bibitem[{Phillip et~al.(2018)Phillip, Chan \& Peiris}]{Phillip:2018}
\bibinfo{author}{Phillip, A.}, \bibinfo{author}{Chan, J. S.~K.}, \&
  \bibinfo{author}{Peiris, S.} (\bibinfo{year}{2018}).
\newblock \bibinfo{title}{A new look at cryptocurrencies}.
\newblock {\it \bibinfo{journal}{Economics Letters}\/},  {\it
  \bibinfo{volume}{163}\/}, \bibinfo{pages}{6--9}.
\bibitem[{Punzo(2019)}]{Punz:Anew:2019}
\bibinfo{author}{Punzo, A.} (\bibinfo{year}{2019}).
\newblock \bibinfo{title}{A new look at the inverse {G}aussian distribution
  with applications to insurance and economic data}.
\newblock {\it \bibinfo{journal}{Journal of Applied Statistics}\/},  {\it
  \bibinfo{volume}{46}\/}, \bibinfo{pages}{1260--1287}.
\bibitem[{Punzo \& Bagnato(2021{\natexlab{a}})}]{Punz:Bagn:PhyA:2021}
\bibinfo{author}{Punzo, A.}, \& \bibinfo{author}{Bagnato, L.}
  (\bibinfo{year}{2021}{\natexlab{a}}).
\newblock \bibinfo{title}{Modeling the cryptocurrency return distribution via
  {L}aplace scale mixtures}.
\newblock {\it \bibinfo{journal}{Physica A: Statistical Mechanics and its
  Applications}\/},  {\it \bibinfo{volume}{563}\/}.
  \DOIprefix\doi{10.1016/j.physa.2020.125354}.
\bibitem[{Punzo \& Bagnato(2021{\natexlab{b}})}]{Punz:Bagn:JSCS:2021}
\bibinfo{author}{Punzo, A.}, \& \bibinfo{author}{Bagnato, L.}
  (\bibinfo{year}{2021}{\natexlab{b}}).
\newblock \bibinfo{title}{The multivariate tail-inflated normal distribution
  and its application in finance}.
\newblock {\it \bibinfo{journal}{Journal of Statistical Computation and
  Simulation}\/},  {\it \bibinfo{volume}{91}\/}, \bibinfo{pages}{1--36}.
\bibitem[{Punzo et~al.(2018)Punzo, Bagnato \&
  Maruotti}]{Punz:Bagn:Maru:Comp:2018}
\bibinfo{author}{Punzo, A.}, \bibinfo{author}{Bagnato, L.}, \&
  \bibinfo{author}{Maruotti, A.} (\bibinfo{year}{2018}).
\newblock \bibinfo{title}{Compound unimodal distributions for insurance
  losses}.
\newblock {\it \bibinfo{journal}{Insurance: Mathematics and Economics}\/},
  {\it \bibinfo{volume}{81}\/}, \bibinfo{pages}{95--107}.
\bibitem[{Punzo \& McNicholas(2016)}]{Punz:McNi:Robu:2016}
\bibinfo{author}{Punzo, A.}, \& \bibinfo{author}{McNicholas, P.~D.}
  (\bibinfo{year}{2016}).
\newblock \bibinfo{title}{Parsimonious mixtures of multivariate contaminated
  normal distributions}.
\newblock {\it \bibinfo{journal}{Biometrical Journal}\/},  {\it
  \bibinfo{volume}{58}\/}, \bibinfo{pages}{1506--1537}.
\bibitem[{Punzo \& Tortora(2021)}]{Punz:Tort:SM:2021}
\bibinfo{author}{Punzo, A.}, \& \bibinfo{author}{Tortora, C.}
  (\bibinfo{year}{2021}).
\newblock \bibinfo{title}{Multiple scaled contaminated normal distribution and
  its application in clustering}.
\newblock {\it \bibinfo{journal}{Statistical Modelling}\/},  {\it
  \bibinfo{volume}{21}\/}, \bibinfo{pages}{332--358}.
\bibitem[{{\textsf{R} Core Team}(2018)}]{R:2018}
\bibinfo{author}{{\textsf{R} Core Team}} (\bibinfo{year}{2018}).
\newblock {\it \bibinfo{title}{\textsf{R}: A Language and Environment for
  Statistical Computing}\/}.
\newblock \bibinfo{organization}{\textsf{R} Foundation for Statistical
  Computing} \bibinfo{address}{Vienna, Austria}.
\newblock \URLprefix \url{https://www.R-project.org/}.
\bibitem[{Rigby \& Stasinopoulos(2005)}]{rigby2005generalized}
\bibinfo{author}{Rigby, R.~A.}, \& \bibinfo{author}{Stasinopoulos, D.~M.}
  (\bibinfo{year}{2005}).
\newblock \bibinfo{title}{Generalized additive models for location, scale and
  shape}.
\newblock {\it \bibinfo{journal}{Journal of the Royal Statistical Society:
  Series C (Applied Statistics)}\/},  {\it \bibinfo{volume}{54}\/},
  \bibinfo{pages}{507--554}.
\bibitem[{Rigby et~al.(2019)Rigby, Stasinopoulos, Heller \&
  De~Bastiani}]{rigby2019distributions}
\bibinfo{author}{Rigby, R.~A.}, \bibinfo{author}{Stasinopoulos, M.~D.},
  \bibinfo{author}{Heller, G.~Z.}, \& \bibinfo{author}{De~Bastiani, F.}
  (\bibinfo{year}{2019}).
\newblock {\it \bibinfo{title}{Distributions for Modeling Location, Scale, and
  Shape: Using {GAMLSS} in \textsf{R}}\/}.
\newblock Chapman \& Hall/CRC The \textsf{R} Series.
\newblock \bibinfo{publisher}{CRC Press}.
\newblock \URLprefix \url{https://books.google.it/books?id=Tu-yDwAAQBAJ}.
\bibitem[{Rytgaard(1990)}]{rytgaard1990estimation}
\bibinfo{author}{Rytgaard, M.} (\bibinfo{year}{1990}).
\newblock \bibinfo{title}{Estimation in the pareto distribution}.
\newblock {\it \bibinfo{journal}{ASTIN Bulletin: The Journal of the IAA}\/},
  {\it \bibinfo{volume}{20}\/}, \bibinfo{pages}{201--216}.
\bibitem[{Schwarz(1978)}]{Schw:Esti:1978}
\bibinfo{author}{Schwarz, G.} (\bibinfo{year}{1978}).
\newblock \bibinfo{title}{Estimating the dimension of a model}.
\newblock {\it \bibinfo{journal}{The Annals of Statistics}\/},  {\it
  \bibinfo{volume}{6}\/}, \bibinfo{pages}{461--464}.
\bibitem[{Shevchenko(2010)}]{Shev:Calc:2010}
\bibinfo{author}{Shevchenko, P.~V.} (\bibinfo{year}{2010}).
\newblock \bibinfo{title}{Calculation of aggregate loss distributions}.
\newblock {\it \bibinfo{journal}{The Journal of Operational Risk}\/},  {\it
  \bibinfo{volume}{5}\/}, \bibinfo{pages}{3--40}.
\bibitem[{Szczygielski et~al.(2020)Szczygielski, Karathanasopoulos \&
  Zaremba}]{Szcz:Kara:Zare:Appl:2019}
\bibinfo{author}{Szczygielski, J.~J.}, \bibinfo{author}{Karathanasopoulos, A.},
  \& \bibinfo{author}{Zaremba, A.} (\bibinfo{year}{2020}).
\newblock \bibinfo{title}{One shape fits all? {A} comprehensive examination of
  cryptocurrency return distributions}.
\newblock {\it \bibinfo{journal}{Applied Economics Letters}\/},  {\it
  \bibinfo{volume}{27}\/}, \bibinfo{pages}{1567--1573}.
\bibitem[{Takaishi(2018)}]{Taka:PhyA:2018}
\bibinfo{author}{Takaishi, T.} (\bibinfo{year}{2018}).
\newblock \bibinfo{title}{Statistical properties and multifractality of
  {B}itcoin}.
\newblock {\it \bibinfo{journal}{Physica A: Statistical Mechanics and its
  Applications}\/},  {\it \bibinfo{volume}{506}\/}, \bibinfo{pages}{507--519}.
\bibitem[{Templ et~al.(2019)Templ, Gussenbauer \&
  Filzmoser}]{templ2019evaluation}
\bibinfo{author}{Templ, M.}, \bibinfo{author}{Gussenbauer, J.}, \&
  \bibinfo{author}{Filzmoser, P.} (\bibinfo{year}{2019}).
\newblock \bibinfo{title}{Evaluation of robust outlier detection methods for
  zero-inflated complex data}.
\newblock {\it \bibinfo{journal}{Journal of Applied Statistics}\/},  {\it
  \bibinfo{volume}{47}\/}, \bibinfo{pages}{1144--1167}.
\bibitem[{Tomarchio \& Punzo(2020)}]{Toma:Punz:Heav:2020}
\bibinfo{author}{Tomarchio, S.~D.}, \& \bibinfo{author}{Punzo, A.}
  (\bibinfo{year}{2020}).
\newblock \bibinfo{title}{Dichotomous unimodal compound models: Application to
  the distribution of insurance losses}.
\newblock {\it \bibinfo{journal}{Journal of Applied Statistics}\/},  {\it
  \bibinfo{volume}{47}\/}, \bibinfo{pages}{2328--2353}.
\bibitem[{Watson(1995)}]{watson1995treatise}
\bibinfo{author}{Watson, G.~N.} (\bibinfo{year}{1995}).
\newblock {\it \bibinfo{title}{A treatise on the theory of Bessel
  functions}\/}.
\newblock \bibinfo{publisher}{Cambridge university press}.
\bibitem[{Yap \& Sim(2011)}]{Yap:Sim:JSCS:2011}
\bibinfo{author}{Yap, B.~W.}, \& \bibinfo{author}{Sim, C.~H.}
  (\bibinfo{year}{2011}).
\newblock \bibinfo{title}{Comparisons of various types of normality tests}.
\newblock {\it \bibinfo{journal}{Journal of Statistical Computation and
  Simulation}\/},  {\it \bibinfo{volume}{81}\/}, \bibinfo{pages}{2141--2155}.
\bibitem[{Zhang et~al.(2018)Zhang, Wang, Li \& Shen}]{Zhang:2018}
\bibinfo{author}{Zhang, W.}, \bibinfo{author}{Wang, P.}, \bibinfo{author}{Li,
  X.}, \& \bibinfo{author}{Shen, D.} (\bibinfo{year}{2018}).
\newblock \bibinfo{title}{Some stylized facts of the cryptocurrency market}.
\newblock {\it \bibinfo{journal}{Applied Economics}\/},  {\it
  \bibinfo{volume}{50}\/}, \bibinfo{pages}{5950--5965}.

\end{thebibliography}

\setlist[description]{font=\normalfont\itshape\textbullet\space}

\appendix

\section{Gamma asymmetric Laplace distribution}
\label{app:Gamma asymmetric Laplace distribution}

Let
\begin{equation}
h\left(w;\theta\right)=\displaystyle\frac{\left(\frac{\theta}{2}\right)^{\frac{\theta}{2}} e^{-\frac{\theta}{2}w} w^{\frac{\theta}{2}-1}}{\Gamma \left(\frac{\theta}{2}\right)},\quad w>0,
\label{eq:Gamma}
\end{equation}
with $\theta>0$, be the pdf of the gamma distribution with shape and rate equal to $\theta/2$.
In symbols, $W\sim \mathcal{G}\left(\theta/2,\theta/2\right)$. 
When the pdf in \eqref{eq:Gamma} is considered as mixing density in model~\eqref{eq:asymmetric Laplace scale mixture}, the pdf of the ALSM becomes
\begin{equation}
f_{\text{ALSM}}\left(x;\mu,\beta,\kappa,\theta\right) = 
\frac{1}{\beta}\frac{\kappa}{1+\kappa^2} 
\begin{cases}
    \left[1+\frac{2\kappa}{\beta\theta}\left(x-\mu\right)\right]^{-\frac{\theta}{2}-1} & \text{for $x\geq \mu$,} \\[3mm]
    \left[1+\frac{2}{\beta\kappa\theta}\left(\mu-x\right)\right]^{-\frac{\theta}{2}-1} & \text{for $x < \mu$.}
  \end{cases}
\label{eq:G-AL distribution}
\end{equation}
Although the pdf in \eqref{eq:G-AL distribution} is a simple reparameterization of the pdf of the UG-AL model given in \eqref{eq:UG-AL distribution}, for the sake of clarity/confusion we prefer referring to the resulting model as gamma asymmetric Laplace (G-AL).
Model~\eqref{eq:G-AL distribution} reduces to $\mathcal{AL}\left(\mu,\beta,\kappa\right)$ when $\theta \rightarrow \infty$.

When $W\sim \mathcal{G}\left(\theta/2,\theta/2\right)$, the $r$th raw moment of $1/W$ is
\begin{align}
\text{E}\left(\frac{1}{W^r}\right) & =
\frac{\theta^{r}\Gamma\left(\frac{\theta}{2}-r\right)}{2^r\Gamma\left(\frac{\theta}{2}\right)}, \nonumber \\
& =  \frac{\theta^r}{\displaystyle \prod_{j=1}^{r}\left( \theta- 2 j \right)}
\label{eq:G moments 1/W}
\end{align}
which exists when $\theta > 2r$.
By substituting in \eqref{eq:ALSM mean}--\eqref{eq:ALSM kurt} the first four raw moments in \eqref{eq:G moments 1/W}, mean, variance, skewness and kurtosis of the G-AL distribution become
\begin{equation}
\text{E}\left(X\right) = \mu + \beta \left(\frac{1}{\kappa}-\kappa \right)\frac{\theta}{\theta-2},
\label{eq:GAL mean}
\end{equation}
which exists if $\theta>2$, 
\begin{equation}
\var\left(X\right) = 
\frac{\beta ^2 \theta ^2 \left(\theta  \kappa ^4+\theta -4 \kappa ^2\right)}{(\theta -4) (\theta -2)^2 \kappa ^2}
\label{eq:GAL variance}
\end{equation}
which exists if $\theta>4$, 
\begin{equation}
\Skew\left(X\right) = 
-\frac{2 \sqrt{\theta -4} (\theta +2) \left(\kappa ^2-1\right) \left[\theta  \kappa ^4+(\theta -6) \kappa ^2+\theta \right]}{(\theta -6) \left(\theta  \kappa ^4+\theta -4 \kappa ^2\right)^{3/2}},
\label{eq:GAL skewness}
\end{equation}
which exists if $\theta>6$, and 
\begin{align}
\Kurt\left(X\right) = &
\frac{1}{(\theta -8) (\theta -6) \left(\theta  \kappa ^4+\theta -4 \kappa ^2\right)^2}
\bigg[\theta ^4 \left(9 \kappa ^8+6 \kappa ^4+9\right)-6 \theta ^3 \left(5 \kappa ^8+12 \kappa ^6+6 \kappa ^4+12 \kappa ^2+5\right) +  \nonumber \\
&  +48 \theta ^2 \kappa ^2 \left(5 \kappa ^4+8 \kappa ^2+5\right) -96 \theta  \left(\kappa ^8+12 \kappa ^4+1\right)+768 \left(\kappa ^6-\kappa ^4+\kappa ^2\right)\bigg],
 \label{eq:GAL kurtosis}
\end{align}
which exists if $\theta>8$.
 
\figurename~\ref{fig:GAL skew and kurt} shows examples of behaviors of $\Skew\left(X\right)$ and $\Kurt\left(X\right)$, as functions of $\kappa$, at various levels of $\theta$.
\begin{figure}[!ht]
\centering
\subfigure[Skewness \label{fig:GALskew}]
{\resizebox{0.48\textwidth}{!}{\includegraphics{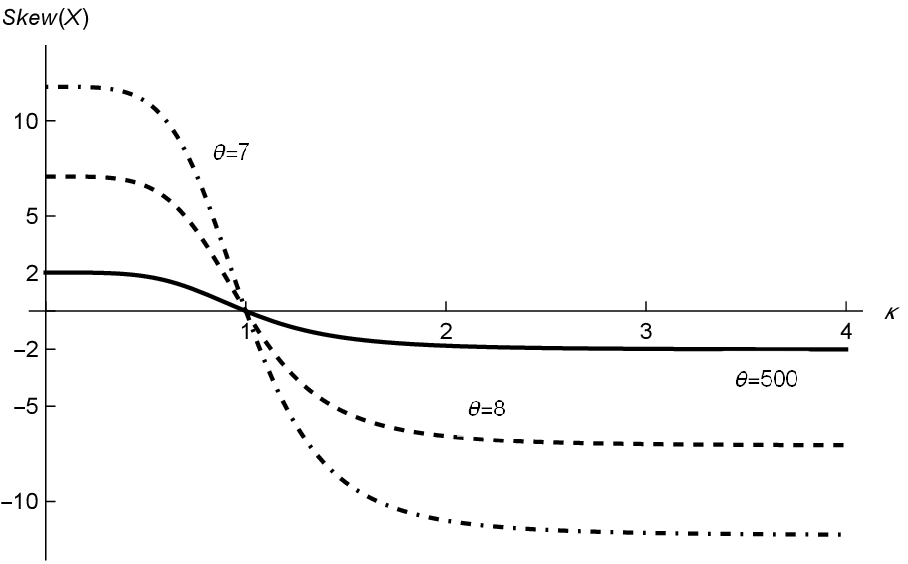}}}
\subfigure[Kurtosis\label{fig:GALkurt}]
{\resizebox{0.48\textwidth}{!}{\includegraphics{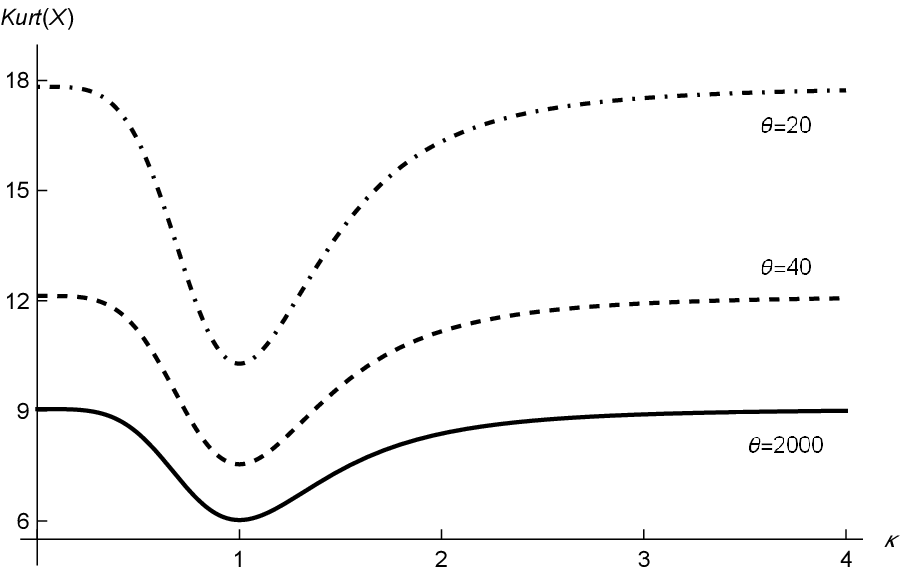}}}
\caption{
\footnotesize 
Examples of behavior of $\Skew\left(X\right)$ (on the left) and $\Kurt\left(X\right)$ (on the right), as functions of $\kappa$, at various levels of $\theta$ for the G-AL distribution.
\label{fig:GAL skew and kurt}
}
\end{figure}
From \figurename~\ref{fig:GALskew} we realize that: 1) large values of $\theta$ tend to produce the plot of $\Skew\left(X\right)$ we obtained for $\mathcal{AL}\left(\mu,\beta,\kappa\right)$ (refer to \figurename~\ref{fig:AL skewness and kurtosis}); and 2) as $\theta$ decreases (to a minimum value of 6), the range of possible values of $\Skew\left(X\right)$ increases.
Moreover, to ``significantly'' modify the behavior of $\Skew\left(X\right)$ with respect to the $\mathcal{AL}\left(\mu,\beta,\kappa\right)$ case, we need low values of $\theta$ (in the examples, $\theta=8$ and $\theta=7$). 
Similarly, from \figurename~\ref{fig:GALkurt} we realize that: 1) large values of $\theta$ tend to produce the plot of $\Kurt\left(X\right)$ we obtained for $\mathcal{AL}\left(\mu,\beta,\kappa\right)$ (refer to \figurename~\ref{fig:AL skewness and kurtosis}); 
and 2) $\kappa$ kept fixed, the lower the value $\theta$ (with a minimum of 8), the higher the kurtosis.
This means that the curve obtained for $\theta\rightarrow \infty$ acts as a lower bound.
Also in this case, to ``significantly'' modify the behavior of $\Kurt\left(X\right)$ with respect to the $\mathcal{AL}\left(\mu,\beta,\kappa\right)$ case, we need low values of $\theta$ (in the examples, $\theta=40$ and $\theta=20$).

\section{E- and M-steps for the G-AL distribution}

For the G-AL distribution we have
\begin{equation} 
l_{2c}\left(\theta\right) = \frac{n\theta}{2} \log \left(\frac{\theta}{2}\right)  
-n \log \left[\Gamma\left(\frac{\theta}{2}\right)\right] 
+\left(\frac{\theta}{2}-1\right) \sum_{i=1}^n \log w_i
- \frac{\theta}{2} \sum_{i=1}^n  w_i.
	\label{eq:complete-data log-likelihood t}
\end{equation}

\subsection{E-step}
\label{app:E-step G-AL}

As for the E-step,
\begin{align}
f\left(w_{i}|x_i;\mu,\beta,\kappa,\theta\right) 
& \propto  f\left(w_{i},x_i;\mu,\beta,\kappa,\theta\right) \nonumber \\
& \propto  f_{\text{G}}\left(w_i;\frac{\theta}{2}+1,\delta\left(x_i;\mu,\beta,\kappa\right)+\frac{\theta}{2}\right) 
\label{eq:posterior_t}.
\end{align} 
This means that $W_{i}|X_i=x_i$ has a gamma distribution with shape $\theta/2+1$ and rate $\delta\left(x_i;\mu,\beta,\kappa\right)+\theta/2$, whose pdf is given in \eqref{eq:posterior_t}; in symbols   
\begin{equation*}
W_{i}|X_i=x_i \sim \mathcal{G}\left(\frac{\theta}{2}+1,\delta\left(x_i;\mu,\beta,\kappa\right)+\frac{\theta}{2}\right).
\label{eq:W_t}
\end{equation*}
The functions $g(W_i)$ arising from \eqref{eq:ALSM complete-data log-likelihood mu, Gamma and Sigma} and \eqref{eq:complete-data log-likelihood t} are $g_1\left(w\right)=w$ and $g_2\left(w\right)=\log\left(w\right)$. 
Thanks to \eqref{eq:posterior_t} we obtain
\begin{align}
\dot{\mbox{E}} \left(W_{i}|X_i=x_i\right) & =   
\frac{\frac{\dot{\theta}}{2}+1}{\delta\left(x_i;\dot{\mu},\dot{\beta},\dot{\kappa}\right) +\frac{\dot{\theta}}{2}} \nonumber \\
&\eqqcolon \dot{w}_{i} \nonumber 
\end{align}
and 
\begin{align}
\dot{\mbox{E}}\left(\log W_{i}|X_i=x_i\right) & =\psi\left(\frac{\dot{\theta}}{2}+1\right) - \log\left[   \delta\left(x_i;\dot{\mu},\dot{\beta},\dot{\kappa}\right) +\frac{\dot{\theta}}{2} \right] 
 \nonumber \\
&\eqqcolon \dot{\log}\:w_{i}. \nonumber
\end{align}

\subsection{M-step}
\label{app:M-step G-AL}

For the G-AL distribution, a closed-form update for $\theta$ does not exist.
The function to maximize is
\begin{equation}
Q_2\left(\theta|\dot{\mu},\dot{\beta},\dot{\kappa},\dot{\theta}\right) = 
\frac{\theta}{2} \left(\sum_{i=1}^n \dot{\log}\: w_i 
-  \sum_{i=1}^n \dot{w}_i \right)
-   \frac{n\theta}{2}  \log \left(\frac{ \theta}{2}\right)
- n \log\left[\Gamma\left(\frac{\theta}{2} \right)\right],
\label{eq:theta t-L loglik}
\end{equation}
whose derivative, with respect to $\theta$, is
\begin{equation}
\frac{\partial Q_2\left(\theta|\dot{\mu},\dot{\beta},\dot{\kappa},\dot{\theta}\right)}
{\partial \theta} =
\frac{1}{2} \left(\sum_{i=1}^n \dot{\log}\: w_i 
-  \sum_{i=1}^n \dot{w}_i \right)
-\frac{n}{2} \left[\psi\left(\frac{\theta}{2} \right) +  \log \left(\frac{\theta}{2} \right) + 1
\right] .
\label{eq:theta t-L loglik derivative}
\end{equation}
The update for $\theta$ can be obtained numerically either by maximizing \eqref{eq:theta t-L loglik} over $\theta$ or by finding the root (over $\theta$) of the equation obtained equating \eqref{eq:theta t-L loglik derivative} to zero.

\section{Method of moments}
\label{app:Method of moments}

In the method of moments (MM) applied to the estimation of the parameters of the ALSM distribution, we relate the (unknown) population moments in \eqref{eq:ALSM mean}--\eqref{eq:ALSM kurt} to their sample counterparts 
\begin{equation*}
\overline{x} = \frac{1}{n} \sum_{i=1}^n x_i,
\quad 
s^2 =  \frac{1}{n-1}  \sum_{i=1}^n \left(x_i - \overline{x}\right)^2,
\end{equation*}
\begin{equation*} 
\widehat{\text{Skew}}(X) = \frac{1}{n s^3} \sum_{i=1}^n 
\left(x_i-\overline{x} \right)^3
,
\quad 
\text{and} \quad 
\widehat{\text{Kurt}}(X) = \frac{1}{n s^4} \sum_{i=1}^n 
\left(x_i-\overline{x} \right)^4.
\end{equation*}

Apart from the TP-AL distribution, where $\btheta$ is bidimensional, for all the other models it is unidimensional.
For these models, with the aim to find the MM estimates of $\mu$, $\beta$, $\kappa$, and $\theta$, we have to solve a system of four equations.
In doing that, it is important to remember that the four equations involve a different number of parameters (refer to Section~\ref{subsec:Moments}).
So, it is convenient to first find the estimates $\widehat{\kappa}$ and $\widehat{\theta}$ by solving the system of the two equations $\text{Skew}(X) = \widehat{\text{Skew}}(X)$ and $\text{Kurt}(X) = \widehat{\text{Kurt}}(X)$.
To search for the roots of the system of the two nonlinear equations, we use the \texttt{nleqslv()} function included in the \textbf{nleqslv} package.  
Moreover, to ensure the existence of the moments we have to impose (if necessary) suitable constraints for $\theta$ (refer to Section~\ref{sec:Examples of AL scale mixtures}).
Then, based on \eqref{eq:ALSM var}, we solve the equation $\text{Var}(X) = s^2$ as a function of $\beta$ with $\kappa$ and $\theta$ replaced by $\widehat{\kappa}$ and $\widehat{\theta}$, respectively; this leads to
\begin{equation}
\widehat{\beta} = 
\sqrt{\frac{s^{2}\widehat{\kappa}^2\left(1+\widehat{\kappa}^2\right)}{  2  \left(1+\widehat{\kappa}^6\right) \widehat{\text{E}}\left(\frac{1}{W^2}\right)- \left(1-\widehat{\kappa}^4\right)\left(1-\widehat{\kappa}^2\right)  \widehat{\text{E}}\left(\frac{1}{W}\right)^2 }},
\label{eq:mm variance}
\end{equation}
where the quantities $\widehat{\text{E}}\left(\frac{1}{W^r}\right)$, with $r=1,2$, are estimates of $ \text{E}\left(\frac{1}{W^r}\right)$ obtained using $\widehat{\theta}$.
Finally, based on \eqref{eq:ALSM mean}, we solve the equation $\text{E}(X) = \overline{x}$ as a function of $\mu$ with $\beta$, $\kappa$ and $\theta$ replaced by $\widehat{\beta}$, $\widehat{\kappa}$ and $\widehat{\theta}$, respectively; this leads to 
\begin{equation}
\widehat{\mu} =  \overline{x} - \widehat{\beta} \left(\frac{1}{\widehat{\kappa}}-\widehat{\kappa}\right) \widehat{\text{E}}\left( \frac{1}{W}  \right) .
\label{eq:mm mean}
\end{equation}

\section{M-Step details}
\label{app:M-Step details}

With the aim to maximize $Q_1\left(\mu,\beta,\kappa|\dot{\mu},\dot{\beta},\dot{\kappa},\dot{\btheta}\right)$ we consider the following two cases: 
\begin{description}
	\item[Case 1,] only $\mu$ is known;
	\item[Case 2,] all parameters are unknown.
\end{description}

\subsection{Case 1}
\label{app:Case 1}

Maximizing $Q_1\left(\mu,\beta,\kappa|\dot{\mu},\dot{\beta},\dot{\kappa},\dot{\btheta}\right)$ is the same as maximizing $n^{-1}\: Q_1\left(\mu,\beta,\kappa|\dot{\mu},\dot{\beta},\dot{\kappa},\dot{\btheta}\right)$ which, when the value of $\mu$ is known, reduces to 
\begin{equation}
Q\left(\beta,\kappa\right) =
\log \kappa -
\log \left(1+\kappa^2\right) -
\log \beta -
\frac{1}{\beta} \left(\kappa, 1/k \right) \bar{\bm{Z}}^{(n)},
\label{eq:Q1_1}
\end{equation} 
where $\bar{\bm{Z}}^{(n)}$ is defined as
\begin{equation}
 \bar{\bm{Z}}^{(n)} = \left(\bar{Z}_1^{(n)}  ,\bar{Z}_2^{(n)} \right)',   
\label{eq:Zn}
\end{equation}
with 
 \begin{equation}
\bar{Z}_1^{(n)}  =   \frac{1}{n}\sum_{i=1}^{n} \dot{w}_i \left(x_i-\mu \right)^+  \quad\text{and} \quad \bar{Z}_2^{(n)} =  \frac{1}{n}\sum_{i=1}^{n} \dot{w}_i \left(x_i-\mu \right)^-. 
\label{eq:Zi}
\end{equation}
The parameter $\mu$ could be: i) $\mu \leq x_{(1)}$, ii) $\mu \geq x_{(n)}$, or iii) $ x_{(1)} < \mu < x_{(n)}$.
If i) holds, then all sample values are greater or equal to $\mu$; this means that 
 \begin{equation*}
\left(x_i-\mu \right)^+ = \left(x_i-\mu \right)  \quad \text{and} \quad   \left(x_i-\mu \right)^- =0 \quad \text{for all $i=1,\ldots,n$} .
\end{equation*}
Thus, 
\begin{equation}
 \bar{\bm{Z}}^{(n)} = \left(\bar{x}^{(n)}_w  , 0 \right)',   
\label{eq:Zni)}
\end{equation}
where 
\begin{equation}
\bar{x}^{(n)}_w = \frac{1}{n}\sum_{i=1}^{n} \dot{w}_i \left(x_i-\mu \right).
\label{eq:xnw)}
\end{equation}
Therefore, \eqref{eq:Q1_1} takes the form
\begin{equation}
Q\left(\beta,\kappa\right) =
\log \kappa -
\log \left(1+\kappa^2\right) -
\log \beta -
\frac{\kappa}{\beta}  \bar{x}^{(n)}_w.
\label{eq:Q1_1i)}
\end{equation} 
Fix $\kappa>0$ and differentiate \eqref{eq:Q1_1i)} with respect to $\beta$ to obtain
\begin{equation}
\frac{\partial Q\left(\beta,\kappa\right)}{\partial \beta} =
-\frac{1}{\beta}
+\frac{\kappa}{\beta^2}  \bar{x}^{(n)}_w .
\label{eq:derivQ1_1i)}
\end{equation} 
The derivative \eqref{eq:derivQ1_1i)} is positive for $\beta<\beta(\kappa)$ and negative for $\beta>\beta(\kappa)$, where
$$
\beta(\kappa) = \kappa \: \bar{x}^{(n)}_w.
$$
As a consequence, for any fixed $\kappa>0$, \eqref{eq:Q1_1i)} is maximized by $\beta(\kappa)$. 
Then, for all $\sigma$ and $\kappa$, we have
\begin{equation}
Q\left(\beta,\kappa\right) \leq Q\left(\beta(\kappa),\kappa\right) = 
- \log \left(1+\kappa^2\right) 
- \log \bar{x}^{(n)}_w - 1,
\label{eq:maxQ1_1i)}
\end{equation} 
where $Q\left(\beta(\kappa),\kappa\right)$ is a function of $\kappa$ only and it is decreasing on $(0,\infty)$, with the least upper bound being equal to 
$$
\lim_{k \rightarrow 0 } Q\left(\beta(\kappa),\kappa\right) = - \log \bar{x}^{(n)}_w - 1 .
$$
Since these values are not admissible, formally the M-step solutions of $\kappa$ and $\beta$ do not exist in this case. 
However, as $\kappa \rightarrow 0^+$ and $\beta(\kappa) \rightarrow 0^+ $, the conditional $\mathcal{AL}\left(\mu,\beta/w,\kappa\right)$ in \eqref{eq:asymmetric Laplace scale mixture} tends to
$\mathcal{SE}_{\left(\mu,\infty\right)}\left(w\right)$.


If we are in the case ii), then all sample values are lower or equal to $\mu$.
This means that
 \begin{equation*}
\left(x_i-\mu \right)^+ = 0  \quad \text{and} \quad   \left(x_i-\mu \right)^- =\left(\mu - x_i \right) \quad \text{for all $i=1,\ldots,n$} .
\end{equation*}
Thus, 
\begin{equation}
 \bar{\bm{Z}}^{(n)} = \left(   0 ,-\bar{x}^{(n)}_w\right)' .  
\label{eq:Znii)}
\end{equation}
Then, \eqref{eq:Q1_1} takes the form
\begin{equation}
Q\left(\beta,\kappa\right) =
\log \kappa -
\log \left(1+\kappa^2\right) -
\log \beta +
\frac{1}{\beta\kappa}  \bar{x}^{(n)}_w .
\label{eq:Q1_1ii)}
\end{equation} 
Fix $\kappa>0$ and differentiate \eqref{eq:Q1_1ii)} with respect to $\beta$ to obtain
\begin{equation}
\frac{\partial Q\left(\beta,\kappa\right)}{\partial \beta} =
-\frac{1}{\beta}
-\frac{1}{\beta^2 \kappa }  \bar{x}^{(n)}_w .
\label{eq:derivQ1_1ii)}
\end{equation} 
The derivative \eqref{eq:derivQ1_1ii)} is positive for $\beta<\beta(\kappa)$ and negative for $\beta>\beta(\kappa)$, where
$$
\beta(\kappa) = -\frac{\bar{x}^{(n)}_w}{ \kappa }   .
$$
Therefore, for any fixed $\kappa>0$, \eqref{eq:Q1_1ii)} is maximized by $\beta(\kappa)$. 
Hence, for all $\sigma$ and $\kappa$, we have
\begin{equation}
Q\left(\beta,\kappa\right) \leq Q\left(\beta(\kappa),\kappa\right) = 
- \log \left(\frac{\kappa^2}{1+\kappa^2}\right) 
- \log \left(-\bar{x}^{(n)}_w\right) - 1,
\label{eq:maxQ1_1ii)}
\end{equation} 
where $Q\left(\beta(\kappa),\kappa\right)$ is a function of $\kappa$ only and it is increasing on $(0,\infty)$, with the limit being equal to 
$$
\lim_{k \rightarrow \infty } Q\left(\beta(\kappa),\kappa\right) =- \log \left(-\bar{x}^{(n)}_w\right) - 1.
$$
Since these values are not admissible, formally the M-step solutions of $\kappa$ and $\beta$ do not exist in this case. 
However, as $\kappa \rightarrow \infty$ and $\beta(\kappa) \rightarrow 0^+ $ we have that the conditional $\mathcal{AL}\left(\mu,\beta/w,\kappa\right)$ in \eqref{eq:asymmetric Laplace scale mixture} tends to a 
``reversed'' shifted exponential distribution with pdf
\begin{equation}
f_{\text{SE}}\left(x;w\right)=  
\begin{cases}
    0  & \text{for $x\geq \mu$}, \\
    w e^{-w\left(\mu-x\right)} & \text{for $x < \mu$}.
  \end{cases}
\label{eq:reverted exponential}
\end{equation}


Under case iii), that is when $ x_{(1)} < \mu < x_{(n)}$, we can find the solutions for $\beta$ and $\kappa$ through the following equations of derivatives of $Q\left(\beta,\kappa\right)$:
\begin{equation} 
\begin{array}{ll}
    \frac{\partial Q\left(\beta,\kappa\right)}{\partial \beta}  = -\frac{1}{\beta}
+\frac{1}{\beta^2} \left(\kappa, 1/k \right) \bar{\bm{Z}}^{(n)}=0 \\
    \frac{\partial Q\left(\beta,\kappa\right)}{\partial \kappa} = 
		 \frac{1}{\kappa}
-\frac{2\kappa}{1+\kappa^2} 

-\frac{1}{\beta} \left(1, -1/k^2 \right) \bar{\bm{Z}}^{(n)}=0
		.
  \end{array}
\label{eq:derivatives}
\end{equation}
These equations are equivalent to
\begin{equation} 
\begin{array}{ll}
\left(\kappa, 1/k \right) \bar{\bm{Z}}^{(n)} =\beta \\
\left(-\kappa^2, 1/k^2 \right) \bar{\bm{Z}}^{(n)} =0
		,
  \end{array}
\label{eq:derivatives2}
\end{equation}
which lead to the following unique and explicit solutions for $\beta$ and $\kappa$:
\begin{equation} 
\widehat{\kappa} =   \sqrt[4]{\frac{\bar{Z}_2^{(n)}}{\bar{Z}_1^{(n)}} }  \quad \text{and} \quad 
\widehat{\beta}  =   \left(\sqrt[4]{\frac{\bar{Z}_2^{(n)}}{\bar{Z}_1^{(n)}} }, \sqrt[4]{\frac{\bar{Z}_1^{(n)}}{\bar{Z}_2^{(n)}} } \right) \bar{\bm{Z}}^{(n)}.
\label{eq:solutions}
\end{equation}

\subsection{Case 2}
\label{app:Case 2}

The function to maximize, when all the parameters are unknown, is
\begin{equation}
Q\left(\mu,\beta,\kappa\right) =
-\log \beta +
\log \left(\frac{\kappa}{1+\kappa^2}\right) -
\frac{1}{\beta} \left[ \kappa a(\mu) + \frac{1}{\kappa}b(\mu)\right],
\label{eq:Q1_12}
\end{equation} 
where
\begin{equation}
a(\mu) = \frac{1}{n}\sum_{i=1}^n \dot{w}_i \left(x_i-\mu \right)^+ \quad \text{and} \quad b(\mu) = \frac{1}{n} \sum_{i=1}^n \dot{w}_i \left(x_i-\mu \right)^-. 
\label{eq:aeb}
\end{equation}
We proceed by first fixing the value of $\mu$ and then applying the same results obtained under the Case 1 in \ref{app:Case 1}.

If $\mu\leq x_{(1)}$, thanks to \eqref{eq:maxQ1_1i)}, we conclude that for any $\beta>0$ and $\kappa>0$ 
$$
Q\left(\mu,\beta,\kappa\right)  \leq  - \log \left(1+\kappa^2\right) 
- \log \bar{x}^{(n)}_w - 1.
$$
Similarly, when $\mu\geq x_{(n)}$, we can use \eqref{eq:maxQ1_1ii)} and conclude that
$$
Q\left(\mu,\beta,\kappa\right)  \leq  - \log \left(\frac{\kappa^2}{1+\kappa^2}\right) 
- \log \left(-\bar{x}^{(n)}_w\right) - 1.
$$
When $ x_{(1)} < \mu < x_{(n)} $, then we use the result in iii) of \ref{app:Case 1}.
In particular we have
\begin{equation}
Q\left(\mu,\beta,\kappa\right)  \leq  Q\left(\mu,\ddot{\beta},\ddot{\kappa}\right),
\label{eq:bound}
\end{equation}
where the quantities $\ddot{\beta}$ and $\ddot{\kappa}$ are obtained as in \eqref{eq:solutions}.
Substituting these values into the right-hand side of \eqref{eq:bound} we obtain 
$$
Q\left(\mu,\beta,\kappa\right)  \leq m\left(\mu\right),
$$
where
 \begin{equation}
m(\mu) =  - 2 \log \left(\sqrt{a(\mu)} + \sqrt{b(\mu)} \right) - \sqrt{a(\mu)}   \sqrt{b(\mu)}.
\label{eq:functionm}
\end{equation}
Maximizing \eqref{eq:functionm} is equivalent to minimize the function 
 \begin{equation}
g(\mu) =   2 \log \left(\sqrt{a(\mu)} + \sqrt{b(\mu)} \right) + \sqrt{a(\mu)}   \sqrt{b(\mu)}.
\label{eq:functiong}
\end{equation}
It turns out that the minimum of $g\left(\cdot\right)$, on the set 
$$
A=\left\{\mu : x_{(1)} \leq \mu \leq  x_{(n)} \right\},
$$
is given by one on the values 
$$
\text{$g\left(x_{(j)}\right)$, $j=1,\ldots,n$.}
$$
This follows from the fact that $g\left(\cdot\right)$ is continuous on the closed interval $\left[x_{(1)}, x_{(n)}\right]$ and concave down on each of the subintervals $\left(x_{(j-1)},  x_{(j)}\right)$, $j=2,\ldots,n$ \citep[see Lemma 3.5.2 in][]{kotz2001laplace}.


\section{ECM algorithm for the two-point asymmetric Laplace distribution}
\label{app:ECM algorithm for the two-point asymmetric Laplace distribution}

To have closed-form updates for the parameters $\mu$, $\beta$, $\kappa$, $\theta_1$ and $\theta_2$ of the TP-AL distribution, at each iteration of the fitting algorithm, we make two changes at the EM algorithm described in Section~\ref{sec:Maximum likelihood estimation}.
The first one is related to the use of the expectation-conditional maximization (ECM) algorithm \citep{Meng:Rubin:Maxi:1993}.
The second change is related to the use of the missing variable 
$$
V=\frac{W-1/\theta_2}{1-1/\theta_2} = \begin{cases}
      1 & \text{with probability $\theta_1$}\\
      0 & \text{with probability $1-\theta_1$,}\\
 	\end{cases}
$$
which is a linear transformation of the missing variable $W$ in \eqref{eq:W Bernoulli}.
With the latter change, the complete-data are given by $\left(x_1,v_1\right),\ldots,\left(x_i,v_i\right),\ldots,\left(x_n,v_n\right)$ and the complete-data likelihood can be written as
\begin{equation*}
L_c\left(\mu,\beta,\theta_1,\theta_2\right)
=
\prod_{i=1}^n\left[\theta_1 f_{\text{AL}}\left(x_i;\mu,\beta,\kappa\right)\right]^{v_i}\left[\left(1-\theta_1\right) f_{\text{AL}}\left(x_i;\mu,\theta_2\beta,\kappa\right)\right]^{1-v_i}.
\end{equation*}
Simple algebra yields the following complete-data log-likelihood
\begin{equation}
l_c\left(\mu,\beta,\theta_1,\theta_2\right)=l_{1c}\left(\theta_1\right)+l_{2c}\left(\mu,\beta,\kappa,\theta_2\right),
\label{eq:complete-data log-likelihood}
\end{equation}
where 
\begin{equation}
l_{1c}\left(\theta_1\right)=\sum_{i=1}^{n}\left[v_i\ln \theta_1+\left(1-v_i\right)\ln \left(1-\theta_1\right)\right]
\label{eq:l1c}
\end{equation}
and
\begin{equation}
l_{2c}\left(\mu,\beta,\theta_2\right)=\sum_{i=1}^{n}\left[v_i\ln f_{\text{AL}}\left(x_i;\mu,\beta,\kappa\right)+\left(1-v_i\right)\ln f_{\text{AL}}\left(x_i;\mu,\theta_2\beta,\kappa\right)\right].
\label{eq:l2c}
\end{equation}

Working on the complete-data log-likelihood in \eqref{eq:complete-data log-likelihood} as for the EM, the ECM algorithm iterates between three steps, one E-step and two CM-steps, until convergence.
The two CM-steps are obtained by partitioning $\left\{\mu,\beta,\kappa,\theta_1,\theta_2\right\}$ in the two subsets $\left\{\mu,\beta,\kappa,\theta_1\right\}$ and $\left\{\theta_2\right\}$.
These steps, for the generic iteration of the algorithm, are detailed below. 



\paragraph{E-step}

In the E-step we update $v_i$, $i=1,\ldots,n$, as
\begin{align}
\dot{\mbox{E}}\left(V_{i}|X_i=x_i\right) & =   
\frac{
\dot{\theta}_1 f_{\text{AL}}\left(x_i;\dot{\mu},\dot{\beta},\dot{\kappa}\right) }
{
\dot{\theta}_1 f_{\text{AL}}\left(x_i;\dot{\mu},\dot{\beta},\dot{\kappa}\right) + \left(1-\dot{\theta}_1\right) f_{\text{AL}}\left(x_i;\dot{\mu},\dot{\theta_2}\dot{\beta},\dot{\kappa}\right)
} \nonumber \\
&\eqqcolon \dot{v}_{i}, 
\label{eq:wi expectation_contaminated}
\end{align}
which is the expected \textit{a posteriori} probability for $x_i$ to come from the reference asymmetric Laplace distribution.
By substituting $v_i$ with $\dot{v}_{i}$ in \eqref{eq:complete-data log-likelihood}, and based on \eqref{eq:l1c} and \eqref{eq:l2c}, we obtain the expected complete-data log-likelihood $Q\left(\mu,\beta,\kappa,\theta_1,\theta_2\left|\right.\dot{\mu},\dot{\beta},\dot{\kappa},\dot{\theta}_1,\dot{\theta}_2\right)=Q_1\left(\theta_1\left|\right.\dot{\mu},\dot{\beta},\dot{\kappa},\dot{\theta}_1,\dot{\theta}_2\right)+Q_2\left(\mu,\beta,\kappa,\theta_2\left|\right.\dot{\mu},\dot{\beta},\dot{\kappa},\dot{\theta}_1,\dot{\theta}_2\right)$.


\paragraph{CM-step 1}

The first CM-step requires the calculation of $\left\{\ddot{\mu},\ddot{\beta},\ddot{\kappa},\ddot{\theta}_1\right\}$ as the value of $\left\{\mu,\beta,\kappa,\theta_1\right\}$ that maximizes $Q\left(\mu,\beta,\kappa,\theta_1,\dot{\theta}_2\left|\right.\dot{\mu},\dot{\beta},\dot{\kappa},\dot{\theta}_1,\dot{\theta}_2\right)$.
The update for $\theta_1$ is calculated independently by maximizing $Q_1\left(\theta_1\left|\right.\dot{\mu},\dot{\beta},\dot{\kappa},\dot{\theta}_1,\dot{\theta}_2\right)$
with respect to $\theta_1$, subject to the constraint on this parameter.
Simple algebra yields
\begin{displaymath}
\ddot{\theta}_1=
\frac{1}{n}\sum_{i=1}^n\dot{v}_i.
\end{displaymath}
The updates of $\mu$, $\beta$, and $\kappa$ are obtained by the maximization of $Q_2\left(\mu,\beta,\kappa,\dot{\theta}_2\left|\right.\dot{\mu},\dot{\beta},\dot{\kappa},\dot{\theta}_1,\dot{\theta}_2\right)$.
It is straightforward to realize that these estimates are analogous to those given in Section~\ref{subsec:M-step} with the only difference that 
\begin{equation*}
a(\mu) = \sum_{i=1}^n \left(\dot{v}_i + \frac{1-\dot{v}_i}{\dot{\theta}_2}\right) \left(x_i-\mu \right)^+ \quad \text{and} \quad b(\mu) = \sum_{i=1}^n \left(\dot{v}_i + \frac{1-\dot{v}_i}{\dot{\theta}_2}\right) \left(x_i-\mu \right)^-, 
\label{eq:a e b TP-AL}
\end{equation*}  

%


\paragraph{CM-step 2}

In the second CM-step, at the same iteration, $\theta_2$ is chosen to maximize $Q_2\left(\ddot{\mu},\ddot{\beta},\ddot{\kappa},\theta_2\left|\right.\dot{\mu},\dot{\beta},\dot{\kappa},\dot{\theta}_1,\dot{\theta}_2\right)$.
Simple algebra yields the closed-form update
\begin{equation*}
\ddot{\theta}_2 = \frac{\displaystyle \sum_{i=1}^{n}\left(1-\dot{v}_i\right)\left[\kappa\left(x_i-\ddot{\mu} \right)^+ + \frac{1}{\kappa}\left(x_i-\ddot{\mu} \right)^-\right] }{\displaystyle \ddot{\beta} \sum_{i=1}^{n}\left(1-\dot{v}_i\right)} .
\label{eq:vartheta update bern}
\end{equation*}

\section{ECME algorithm for the uniform asymmetric Laplace distribution}
\label{app:ECME algorithm for the uniform asymmetric Laplace distribution}

To fit the U-AL distribution we use the ECME algorithm.
In our case, it iterates between three steps, one E-step and two CM-steps, until convergence. 
The two CM-steps arise from the partition of $\left\{\mu,\beta,\kappa,\theta\right\}$ as $\left\{\mu,\beta,\kappa\right\}$ and $\left\{\theta\right\}$.
From an operational point of view, the only difference with respect to the EM algorithm is the second CM-step, where we directly maximize the (observed-data) log-likelihood function instead of the complete-data one. 
We detail the three steps, for the generic iteration of the algorithm, below.

\paragraph{E-step}

To calculate the expectations involved in the E-step (refer to Section~\ref{subsec:E-step}), we observe that
\begin{align}
f\left(w_{i}|x_i;\mu,\beta,\kappa,\theta\right) 
& \propto  f\left(w_{i},x_i;\mu,\beta,\kappa,\theta\right) \nonumber \\
& \propto \frac{1}{\eta \left(x_i,\mu,\beta,\kappa,\theta\right)}  
f_{\text{G}}\left[w_i;2,\delta\left(x_i;\mu,\beta,\kappa\right) \right] ,  
\label{eq:posterior_uniform} 
\end{align} 
where 
$$
\eta \left(x_i;\mu,\beta,\kappa,\theta\right)= \Gamma \left[2, (1-\theta)\delta\left(x_i;\mu,\beta,\kappa\right) \right]  -\Gamma\left[2,\delta\left(x_i;\mu,\beta,\kappa\right) \right].
$$
This means that $W_{i}|X_i=x_i$ has a doubly-truncated gamma distribution (see, e.g., \citealp{coffey2000properties}), on the interval $\left(1-\theta,1\right)$, with shape $2$ and rate $\delta\left(x_i;\mu,\beta,\kappa\right)$, whose pdf is given in \eqref{eq:posterior_uniform}; in symbols 
\begin{equation}
W_{i}|X_i=x_i \sim \mathcal{DTG}_{\left(1-\theta,1\right)}\left(2,\delta\left(x_i;\mu,\beta,\kappa\right)\right).
\label{eq:Doubly Truncated Gamma TIN}
\end{equation}
The function $g(W_i)$ arising from \eqref{eq:ALSM complete-data log-likelihood mu, Gamma and Sigma} is $g_1\left(w\right)=w$.
Thanks to \eqref{eq:Doubly Truncated Gamma TIN} we obtain
\begin{align}
\dot{\mbox{E}}\left(W_{i}|X_i=x_i\right) & =
\frac{1}{\delta\left(x_i;\dot{\mu},\dot{\beta},\dot{\kappa}\right)}
\frac{\Gamma\left[3, (1-\dot{\theta}) \delta\left(x_i;\dot{\mu},\dot{\beta},\dot{\kappa}\right) \right] -\Gamma\left[3, \delta\left(x_i;\dot{\mu},\dot{\beta},\dot{\kappa}\right) \right] }{\Gamma\left[2, (1-\dot{\theta}) \delta\left(x_i;\dot{\mu},\dot{\beta},\dot{\kappa}\right) \right] -\Gamma\left[2, \delta\left(x_i;\dot{\mu},\dot{\beta},\dot{\kappa}\right) \right]}   ,  \nonumber \\
&\eqqcolon \dot{w}_{i} . \nonumber
\end{align}

\paragraph{CM-step 1} 

The first CM-step proceeds as the M-step of the EM algorithm, but limited to the update of $\mu$, $\beta$, and $\kappa$ by maximizing $Q_1\left(\mu,\beta,\kappa|\dot{\mu},\dot{\beta},\dot{\kappa},\dot{\theta}\right)$; refer to Section~\ref{subsec:M-step}. 

\paragraph{CM-step 2} 

The update of $\theta$ is obtained at the second CM-step by maximizing $l\left(\mu,\beta,\kappa,\theta|\ddot\mu,\ddot\beta,\ddot\kappa\right)$ over $\theta$, namely the observed-data log-likelihood function in \eqref{eq:compound loglik} with $\mu$, $\beta$, and $\kappa$ fixed at $\ddot\mu$, $\ddot\beta$, and $\ddot\kappa$, respectively.

\end{document}